\documentclass[reprint,superscriptaddress,amsmath,amssymb,aps,pra]{revtex4-2}
\usepackage{graphicx}
\usepackage{bm}
\usepackage[colorlinks=true,
  linkcolor=black,
  citecolor=blue,
  urlcolor =blue]{hyperref}
\usepackage{amsmath}
\usepackage{amssymb}
\usepackage{amsthm}
\usepackage{graphicx}
\usepackage{mathrsfs}
\usepackage{braket}
\usepackage{comment}
\usepackage[dvipsnames]{xcolor}
\usepackage{soul,xcolor}
\usepackage{color, xcolor, colortbl}
\usepackage{graphicx}
\usepackage{amsthm,amsmath,amssymb,amsfonts}
\usepackage{dcolumn}
\usepackage{url}
\usepackage{epstopdf}
\usepackage{enumitem}
\usepackage{algpseudocode}
\usepackage{bm}
\usepackage{appendix}
\usepackage{multirow}
\usepackage{braket}
\usepackage[english]{babel}
\usepackage{hyperref}
\usepackage[capitalize]{cleveref}
\usepackage{algorithm}
\usepackage{xpatch}
\usepackage{tikz}
\usepackage{adjustbox}
\usepackage{xspace}
\usetikzlibrary{quantikz}
\usepackage{pifont}

\newcommand{\ud}{\,\mathrm{d}}

\newcommand{\mc}[1]{\mathcal{#1}}

\newtheorem{thm}{\protect\theoremname}

\newtheorem{prop}[thm]{\protect\propositionname}

\providecommand{\definitionname}{Definition}
\providecommand{\assumptionname}{Assumption}
\providecommand{\corollaryname}{Corollary}
\providecommand{\lemmaname}{Lemma}
\providecommand{\propositionname}{Proposition}
\providecommand{\remarkname}{Remark}
\providecommand{\theoremname}{Theorem}
\usepackage{bbm}

\newcommand{\rev}[1]{\textcolor{black}{#1}}
\newcommand{\rrev}[1]{{\color{black}#1}}
\newcommand{\frev}[1]{{\color{black}#1}}
\NewDocumentCommand{\ketbra}{mG{#1}}{\mathinner{|{#1}\rangle\!\langle{#2}|}}

\usepackage{silence}


\usetikzlibrary{fit}
\tikzset{%
  highlight/.style={rectangle,rounded corners,fill=blue!15,draw,fill opacity=0.3,thick,inner sep=0pt}
}

\graphicspath{{figures/}}

\makeatletter
\renewcommand\section[1]{\par\addvspace{1.5ex}%
  \noindent\textbf{#1.}\quad}
\makeatother

\begin{document}

\title{Simple and efficient end-to-end quantum thermal and ground state preparation}

\author{Zhiyan Ding}
\email{zyding@umich.edu}
\thanks{These authors contributed equally to this work.}
\affiliation{Department of Mathematics, University of Michigan, Ann Arbor}
\author{Yongtao Zhan}
\email{yzhan@caltech.edu}
\thanks{These authors contributed equally to this work.}
\affiliation{Institute for Quantum Information and Matter, California Institute of Technology}
\affiliation{Division of Physics, Mathematics, and Astronomy, California Institute of Technology}
\author{John Preskill}
\email{preskill@caltech.edu}
\affiliation{Institute for Quantum Information and Matter, California Institute of Technology}
\affiliation{Division of Physics, Mathematics, and Astronomy, California Institute of Technology}
\affiliation{AWS Center for Quantum Computing}
\author{Lin Lin}
\email{linlin@math.berkeley.edu}
\affiliation{Department of Mathematics, University of California, Berkeley}
\affiliation{Applied Mathematics and Computational Research Division, Lawrence Berkeley National Laboratory}

\begin{abstract}
Quantum computing algorithms for many-body physics, chemistry and materials science typically require the preparation of thermal or ground states for a given Hamiltonian. We propose quantum algorithms based on system-bath interactions to prepare thermal and ground states for a range of physically relevant Hamiltonians. These algorithms require only forward evolution under a system-bath Hamiltonian in which the bath is a single reusable ancilla qubit, making them especially well-suited for early fault-tolerant quantum devices. By carefully designing the bath and interaction Hamiltonians, we prove that the fixed point of the dynamics accurately approximates the desired quantum state. Furthermore, we establish theoretical guarantees on the mixing time, and thereby provide a rigorous justification for the end-to-end efficiency of system-bath interaction models in thermal and ground state preparation for the physically relevant models we consider.
\end{abstract}
\maketitle

Preparing quantum thermal and ground states is a fundamental task in quantum computing, with wide-ranging applications in quantum many-body physics, quantum chemistry, and materials science. A variety of approaches have been developed in the literature, including variational algorithms~\cite{tilly_variational_2022}, adiabatic methods~\cite{farhi2000,RevModPhys.90.015002}, matrix function based methods using quantum phase estimation and quantum singular value transformation~\cite{PoulinWocjan2009,GilyenSuLowEtAl2019,LinTong2020a}, tomography-based quantum imaginary time evolution (QITE)~\cite{motta_determining_2020}, and dissipative state preparation schemes~\cite{VerstraeteWolfCirac2009,Temme_2011}, to name a few.

While each approach has its own strengths and limitations, there has been a persistent trade-off between simplicity, efficiency, generality, and rigor. Namely, we would like to design algorithms that are simple enough to be implemented on early fault-tolerant quantum devices. The runtime should scale at most polynomially with relevant parameters such as system size and precision.
The polynomial runtime scaling is not a trivial requirement: preparing the ground state of a general local Hamiltonian is QMA-hard in the worst case~\cite{KitaevShenVyalyi2002}, and even preparing classical Gibbs states can be NP-hard~\cite{Barahona1982,Sly2010}. Thus, polynomial-time quantum algorithms cannot be expected for all Hamiltonians. Our goal should be to develop algorithms that are efficient for broad classes of physically relevant Hamiltonians, and all of these should come with rigorous end-to-end performance guarantees.  To date, we are not aware of any approach that simultaneously satisfies all of these criteria.

Among existing approaches, dissipative algorithms appear to be the most promising candidates. They can avoid variational parameter tuning, state tomography procedures, and normalization factors that can grow exponentially with system size. In recent years, a new wave of approaches based on dissipative dynamics, such as the Lindblad dynamics, have been proposed to prepare thermal and ground states of general Hamiltonians~\cite{Cubitt2023,MiMichailidisShabaniEtAl2024,RallWangWocjan2023,ChenKastoryanoBrandaoEtAl2023,ChenKastoryanoGilyen2023,Ding_2025,Ding2024,li2024dissipative,zhan2025rapidquantumgroundstate,hahn2025,gilyen2024,jiang2024quantum,eder2025quantum,motlagh2024ground,Lloyd2025Quasiparticle}.
In a broad sense, dissipative algorithms mimic how nature prepares thermal states. The system of interest is coupled to an external bath that induces dissipation. Importantly, dissipation here does not mean that the system state is post-selected based on specific measurement outcomes.  This issue arises in direct implementations of imaginary-time evolution on quantum computers, where the probability of obtaining the desired thermal state in the system register decreases exponentially with system size. In contrast, dissipative algorithms are naturally described by quantum channels, ensuring that at each step the system register remains in a well-defined physical quantum state without post-selection. In this way, dissipative algorithms can be interpreted as quantum generalizations of Markov chain Monte Carlo (MCMC) methods. By carefully designing the dissipation mechanism, or equivalently the quantum moves, and repeatedly applying the resulting quantum channel that simulates the dissipative dynamics, the system can be driven toward its thermal or ground state.

A key metric for evaluating the efficiency of dissipative algorithms is the total simulation time required to reach the target state. The total simulation time is often characterized by the \emph{mixing time}, which refers to how long it takes for the system to evolve close enough to its steady state, regardless of the initial state. In other words, it measures the time needed for the system to ``forget'' its starting point and settle into thermal or ground state equilibrium.
Recently, there is substantial progress in rigorous understanding the mixing time of the \emph{continuous-time} Lindblad dynamics for several classes of physically relevant Hamiltonians~\cite{BardetCapelGaoEtAl2023,rouz2024,DingLiLinZhang2024,kochanowski2024rapid,rouze2024optimal,tong2024fast,smid2025polynomial}.

A main drawback of dissipative algorithms for thermal and ground state preparation lies in the complexity of their implementation, i.e., the implementation of quantum channels that accurately simulate the dissipative dynamics. Existing high-order 
simulation algorithms \cite{cleve_2017,li_et_2023,Ding2024,ChenKastoryanoBrandaoEtAl2023,PRXQuantum.5.020332} can achieve optimal or near-optimal scaling in terms of the simulation time, but they often involve complex logic gates, controlled or time-reversed Hamiltonian simulations, and a large number of ancilla qubits. These complex circuit structures can be challenging to implement on early fault-tolerant quantum devices, where quantum resources remain limited, error correction may be applied selectively to components most susceptible to noise, and hybrid architectures combining analog (such as forward Hamiltonian evolution) and digital elements may be used to reduce implementation costs and suppress errors. Consequently, algorithms must be significantly simplified before they can be considered for practical deployment. On the other hand, simplified protocols often deviate too far from the continuous-time Lindblad dynamics, making it difficult to rigorously analyze the mixing time.

\begin{figure*}[htbp!]
\begin{center}
\includegraphics[width=0.8\textwidth]{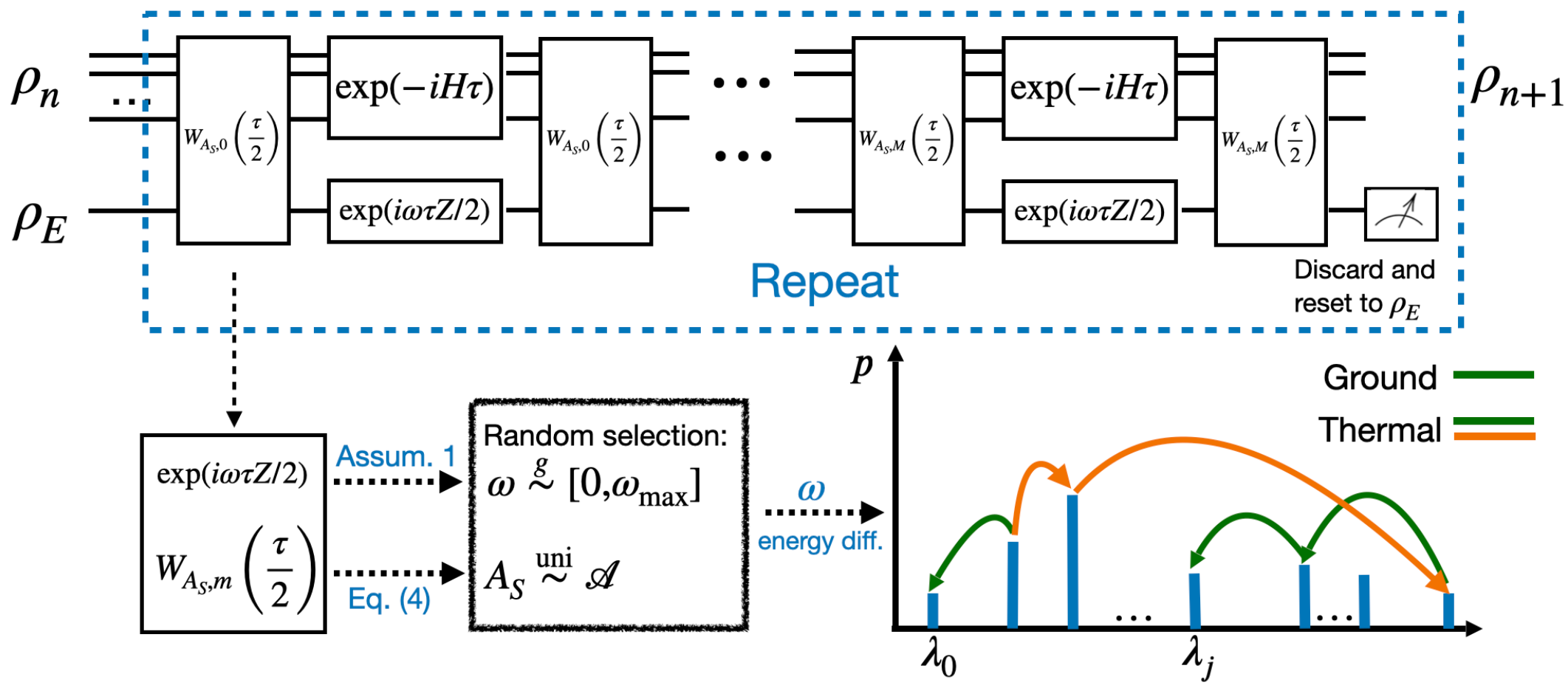}
\end{center}
\caption{Illustration of the quantum circuit for the proposed algorithm. The weak-coupling dynamics in~\eqref{eqn:Phi_alpha} are implemented via a second-order Trotter decomposition. The procedure requires only a single ancilla qubit. Here, $\rho_E\propto \exp(\beta \omega Z/2)$ for thermal state preparation, and $\rho_E=\ket{0}\bra{0}$ for ground state preparation. The system-bath interaction term $W_{A_S,m}(\tau/2)$ is defined in~\eqref{eqn:Wnformula} with $A_S$ randomly sampled from a set $\mc{A}$. For each $\omega > 0$ randomly sampled from a given distribution $g(\omega)$, the algorithm approximately generates a jump from an energy eigenstate $\ket{\psi_j}$ to another eigenstate $\ket{\psi_i}$ with Bohr frequency $\lambda_i - \lambda_j \approx \pm \omega$. That is, the change in energy of the system can be either $+\omega$ or $-\omega$. In the thermal state preparation algorithm, both energy-increasing and energy-decreasing transitions are allowed, while in the ground state preparation algorithm, only energy-decreasing transitions ($\lambda_i - \lambda_j \approx -\omega$) are permitted. The system-bath interaction block $W(\cdot)$ is explained in the Methods section.}
\label{fig:ourrepeatedint}
\end{figure*}

In this work we resolve the above tension by introducing a quantum algorithm that is simple to implement, and yet admits rigorous end-to-end performance guarantees. Our approach is a randomized system-bath interaction scheme that requires only forward evolution under the system Hamiltonian and a single reusable environment qubit. At a high level, the protocol proceeds as follows (Fig. 1). Given a system Hamiltonian $H$ and a small set of local system operators $\mathcal{A}$, prepare a one-qubit environment state. In each round, draw a coupling $A_S\in\mathcal{A}$ and a random bath frequency $\omega$ from a prescribed distribution, evolve forward under a weak system-bath interaction with a smooth temporal envelope while the system simultaneously evolves under $H$, then trace out and reset the environment qubit. The resulting quantum channel acting on the system is iterated, and the number of iterations is controlled by the mixing time.

To justify the efficiency of this protocol, we compare this quantum channel with a continuous-time Lindblad evolution whose fixed point  can be known to approximate the target state. By carefully choosing the environment state, coupling operators, and sampling distributions, we ensure that the fixed point of the channel can be made arbitrarily close to the target state. Moreover, we establish polynomial mixing-time bounds for several physically relevant Hamiltonians, thereby providing rigorous guarantees within a minimal architecture that is well-suited for early fault-tolerant devices.

We note that the idea of employing a weak system-bath interaction framework to cool a system and prepare low-temperature thermal or ground states has been explored in several previous and concurrent works~\cite{molpeceres2025,Lloyd2025Quasiparticle,hahn2025provably,langbehn2025universal,lloyd2025quantumthermal,scandi2025thermal}, from both theoretical and numerical perspectives. In particular, Refs.~\cite{hahn2025provably,lloyd2025quantumthermal} adopted similar principles to design quantum algorithms for thermal state preparation based on system-bath interaction models,  although some specific details of their protocols differ from ours. In this work, we carefully design both the initialization of the bath state and the interaction function, thereby providing, to the best of our knowledge, the first end-to-end efficiency guarantee within the system–bath interaction framework. Specifically, our study is the first to establish rigorous mixing-time guarantees for both thermal and ground state preparation in such a framework, laying a theoretical foundation for understanding the promising empirical results observed in related studies. A more detailed comparison with the previous work can be found in Supp. S2.


\section{Algorithm description}\label{sec:detail_algorithm}
Given a system Hamiltonian $H$, we define the total Hamiltonian as
\begin{equation}\label{eqn:Ham_total}
H_{\alpha}(t) = H + H_E + \alpha f(t) \left( A_S \otimes B_E + A^\dagger_S \otimes B_E^\dagger \right).
\end{equation}
Here, $A_S$ is a coupling operator acting on the system, which can be randomly chosen from a user-provided set of operators set $\mc{A}=\{A^i,-A^i\}_i$ with the property that $\{(A^i)^\dagger\}_i=\{A^i\}_i$. We assume without loss of generality $\|A^i\|\leq 1$. We set $B_E = (X-iY)/2 = \ket{1}\bra{0}$. The bath Hamiltonian $H_E=-\omega Z/2$ depends on a random frequency variable $\omega$, drawn from a distribution $g(\omega)$ related to the system Hamiltonian. \rrev{Here, $X,Y,Z$ denote the single-qubit Pauli operators.} The interaction envelope $f(t)$ is a real, positive, even, and normalized function of time, satisfying $f(t)=f(-t)\in\mathbb{R}$ and $\int_{\mathbb{R}} f^2(t)\,\ud t = 1$. We choose $f$ to be a Gaussian function
$f(t) = \frac{1}{(2\pi)^{1/4} \sigma^{1/2}} \exp\left(-\frac{t^2}{4\sigma^2}\right)$.

Starting from an initial state $\rho_0$ in the system register, in each step of the algorithm we set the state of the ancilla qubit to $\rho_E\propto \exp(-\beta H_E)$, we simulate the time evolution governed by the total Hamiltonian $H_\alpha(t)$ and then trace out the ancilla qubit. When $\beta=\infty$, we set $\rho_E=\ket{0}\bra{0}$.
Then
\begin{equation}\label{eqn:Phi_alpha}
\rho_{n+1} = \Phi(\rho_n) := \mathbb{E}_{A_S,\omega}\left(\operatorname{Tr}_E \left[ U^\alpha(T) \left( \rho_n \otimes \rho_E \right) U^{\alpha}(T)^\dagger \right]\right),
\end{equation}
where $U^\alpha(t)$ is the time-evolution operator of the total Hamiltonian $U^\alpha(t) := \mathcal{T} \exp\left(-i \int_{-T}^t H_{\alpha}(s) \, \mathrm{d}s \right)$. \rrev{Note that the Hamiltonian evolution starts at time $-T$. In the implementation, we simulate the time-dependent Hamiltonian in $\Phi$ using exisiting Hamiltonian simulation protocols. The detailed discussion can be found in Methods section and Supp. S6 D.}



\rev{To the best of our knowledge, this work provides the first rigorous analysis showing that, for physically relevant but analytically tractable classes of models such as free fermions and commuting Hamiltonians, the map $\Phi_{\alpha}$ possesses a unique fixed point that accurately approximates the desired thermal or ground state, and that its mixing time can be explicitly bounded. These results offer a concrete demonstration that the mechanism underlying $\Phi_{\alpha}$ can be made fully rigorous in nontrivial quantum many-body settings. More generally, we prove that whenever the mixing time $\tau_{\mathrm{mix}}$ of $\Phi_{\alpha}$ grows at most polynomially in the system size, one can choose parameters so that the state $\rho_{\tau_{\mathrm{mix}}}$ stays arbitrarily close to the fixed point $\rho_{\mathrm{fix}}$, and hence serves as a reliable approximation to the target thermal or ground state.}


\section{Rigorous guarantees}
Our first result is to show that the fixed point of the quantum channel $\Phi$ can be made arbitrarily close to the target  state. The mixing time $\tau_{\rm mix}(\epsilon)$ is defined as the minimum number of iterations required to ensure that $\Phi^k(\rho_0)$ is $\epsilon$-close to the fixed point for any initial state $\rho_0$ (see Definition S1 for the formal definition).

\begin{thm}[Informal main result]\label{thm:main_informal}
Let $\Phi$ be the quantum channel in \cref{eqn:Phi_alpha}. For any target precision $\epsilon>0$, there exist parameter choices such that the trace distance between the fixed point $\rho_{\mathrm{fix}}(\Phi)$ and the target state is at most $\epsilon$.  Starting from an arbitrary initial state, the end-to-end Hamiltonian evolution time is upper bounded by
\begin{equation}
T_{\mathrm{total}} =
\begin{cases}
\mathrm{poly}\big(\tau_{\mathrm{mix}}(\epsilon),\beta,1/\epsilon\big), & \text{thermal},\\
\mathrm{poly}\big(\tau_{\mathrm{mix}}(\epsilon),\Delta^{-1},1/\epsilon\big), & \text{ground state},
\end{cases}
\end{equation}
 where $\beta$ is the inverse temperature and $\Delta$ is the spectral gap of the system Hamiltonian $H$.
\end{thm}


We note that ground states can also be approximated by thermal states at low enough temperatures, $\beta = \mathrm{poly}(N,1/\Delta)$ \cite{PhysRevX.11.011047}. We provide a detailed discussion of this complexity discrepancy in Supp. S5. The total Hamiltonian evolution time $T_{\mathrm{total}}$ in~\cref{thm:main_informal} is the product of the number of iterations $\tau_{\rm mix}(\epsilon)$ and the Hamiltonian evolution time $2T$ in each iteration. The polynomial dependence on the mixing time $\tau_{\rm mix}(\epsilon)$ occurs because the Hamiltonian evolution time $T$ in each iteration needs to scale polynomially with $\tau_{\rm mix}(\epsilon)$, so that the fixed point of $\Phi$ can remain close
to the target state after $\tau_{\rm mix}(\epsilon)$ iterations.



To establish end-to-end efficiency in preparing thermal and ground states, we derive rigorous mixing-time bounds for a range of physically relevant Hamiltonians. These results are presented in the following two theorems:
\begin{thm}\label{thm:informal_thermal}
For thermal state preparation of $N$-qubit non-interacting spin systems or quadratic fermionic systems \rev{with $\beta=\Theta(1)$}, and of local commuting spin systems at high temperature, the mixing time $\tau_{\rm mix}(\epsilon)$ of $\Phi$ scales as \rev{$\mathrm{poly}(N,1/\epsilon)$}.
\end{thm}

\begin{thm}\label{thm:informal_ground}
For ground state preparation of $N$-qubit non-interacting spin systems or quadratic fermionic systems, the mixing time $\tau_{\rm mix}(\epsilon)$ of $\Phi$ scales as $\mathrm{poly}(N,1/\epsilon)$.
\end{thm}

By combining the mixing-time bounds with~\cref{thm:main_informal}, the algorithm achieves runtime scaling polynomially with system size.

\section{Analysis overview}
\frev{We first need to show that the fixed point of $\Phi$ is $\epsilon$-close to the target thermal or ground state, and then derive an upper bound on the mixing time of $\Phi$.}

First, we show that Eq.~\eqref{eqn:Phi_alpha} can be approximated by a continuous-time Lindbladian dynamics whose fixed point can be tuned to lie arbitrarily close to the target state by setting $T=\widetilde{\Omega}(\sigma)$ and \rev{$\sigma\gg 1, \alpha T\ll 1$}.  Specifically, we first show that $\Phi$ can be approximated by a Lindbladian dynamics $\exp(\mathcal{L}\alpha^2)$—up to a unitary evolution—with error $\mathcal{O}(\alpha^4\sigma^2)$, where $\mc{L}$ is the Lindbladian operator and $\alpha^2$ corresponds to the Lindbladian evolution time. This result is formally given in Supp. S3. Next, we show that the fixed point of the Lindbladian dynamics closely approximates the target thermal or ground state. Consequently, the fixed point of $\Phi$ is also close to the desired state. For the thermal state, the fixed-point error can be bounded by $\widetilde{\mathcal{O}}\left((\beta/\sigma)\alpha^2\tau_{\rm mix}(\epsilon)\right)$, while for the ground state, the error scales as $\widetilde{\mathcal{O}}\left(\exp(-\sigma^2\Delta^2/32)\alpha^2\tau_{\rm mix}(\epsilon)\right)$, provided the Hamiltonian has a spectral gap $\Delta$. The detailed versions of the approximation error are stated in Supp. S5 as Theorem S4 and Theorem S5, respectively.

According to the approximation error bound, when $\alpha^2\tau_{\rm mix}(\epsilon)$ remains essentially constant, an appropriate choice of the parameters $\sigma$ and $\alpha$ ensures that the fixed point $\rho_{\text{fix}}(\Phi)$ can be made arbitrarily close to the target state. Moreover, since $\Phi$ can be approximated by the Lindbladian dynamics, we show that the rescaled mixing time $\alpha^2\tau_{\rm mix}(\epsilon)$ closely matches the mixing time of the Lindbladian operator $\mc{L}$. This correspondence is a key step in establishing an upper bound on the mixing time of $\Phi$ in the second stage of our analysis.

The main technical challenge, however, is that in practice the rescaled mixing time $\alpha^2\tau_{\rm mix}(\epsilon)$ may diverge as $\sigma$ or $\alpha^{-1}$ grows. For instance, in thermal state preparation the fixed-point error bound often scales as $\alpha^2\tau_{{\rm mix}}(\epsilon)/\sigma$. Without a carefully designed filter function $f$ and system-bath interaction, $\alpha^2\tau_{{\rm mix}}(\epsilon)$ can \emph{grow exponentially} with $\sigma$,~\emph{even for} the simple single-qubit Hamiltonian $H=Z$ (see Section S7, Theorem S17.). Thus a key step in our work is to design the dissipative protocol to ensure that once $\sigma$ is sufficiently large, the mixing time $\tau_{\mathrm{mix}}(\epsilon)$ is \emph{independent} of $\sigma$.

To achieve this, we need to carefully design the initial bath state and the interaction function.~\rev{This construction provides significant flexibility in tuning $f$, allowing the resulting Lamb-shift term to approximately commute with the target thermal or ground state. Specifically, for a fixed $\omega$ and sufficiently large $\sigma$, the system-bath interaction is designed so that the induced energy transitions in the system concentrate around $\omega$ rather than near zero. Then, by sampling $\omega$ from a suitable distribution, our method enables large energy transitions while keeping the perturbation to the approximate fixed point small, thereby ensuring provably fast mixing without substantially disturbing the target state.}




\section{Applications of the algorithm}\label{sec:application}

\frev{\emph{Quantum materials and chemistry on early fault-tolerant hardware.}}
Efficient state preparation with end-to-end performance guarantees is essential for reliable property prediction in materials and molecules, especially when classical methods disagree. Our protocol prepares thermal and ground states using only forward evolution under a system-bath Hamiltonian and a single ancilla qubit. This minimal architecture reduces control overhead and avoids deep control logic, which is attractive for early fault-tolerant and hybrid analog-digital devices.

\medskip

\frev{\emph{Benchmarks for dissipative engineering and co-design with error correction.}}
Because the channel's fixed point can be tuned to approximate a target thermal or ground state with minimal circuitry, our protocol provides a benchmark that directly reflects the end-to-end performance of a quantum device. This perspective is especially relevant for architectures with mid-circuit measurement capabilities, which is a core ingredient of error correction. Rather than focusing on per-gate fidelity, the dissipative dynamics test how well the hardware and control stack jointly support long-time evolution toward a stable fixed point.
In practice, the most natural initial targets are Hamiltonians with simple native structure, such as free-fermion or commuting local models, and modest system sizes where forward system–bath evolution and repeated ancilla reset can be implemented with minimal ancilla counts and low overhead.
In this sense, the protocol is a natural setting for co-design: it allows hardware and algorithmic choices to be evaluated together in terms of their impact on the overall ability to prepare and stabilize complex quantum states.

\section{Conclusion and outlook}\label{sec:discussion}
In this work, we have presented a quantum algorithm for preparing both thermal and ground states with provable correctness and mixing-time guarantees, while requiring only modest quantum resources. The algorithm operates by repeatedly applying a simple quantum channel, obtained from forward time evolution under a system-bath Hamiltonian, followed by resetting the bath. The bath can be reduced to a single ancilla qubit by employing a randomized choice of bath frequency and coupling operator between the system and bath.
This design makes it well suited for early fault-tolerant quantum devices, including hybrid devices that can implement certain type of Hamiltonian simulations via analog methods.

Our analysis shows that, with a suitable choice of parameters, the fixed point of the quantum channel can accurately approximate the target thermal or ground state. Once the fixed-point error bound is established, it suffices to prove polynomial mixing time in order to achieve end-to-end efficiency. A key technical point is that the mixing time and simulation parameters are generally interdependent; it is therefore important to demonstrate that, for our channel construction, the mixing time remains bounded as parameters are tuned to reduce the fixed-point error. \rrev{While we have established polynomial mixing for several physically relevant models, similar behavior may also extend to broader classes of Hamiltonians, especially in the thermal-state setting, as indicated by two recent works~\cite{slezak2026polynomialtimethermalizationgibbssampling,wang2026lindbladdynamicsrigorousguarantees} that appeared after the initial completion of this paper.}


The weak system-bath interaction framework can be regarded as an approximate, low-order approach to Lindblad simulation. In this work, we propose a second-order Trotterization method for the simulation of $\Phi$, which circumvents the need for backward Hamiltonian evolution. Compared with high-order Lindblad simulation methods~\cite{cleve_2017,li_et_2023,PRXQuantum.5.020332}, this can lead to an asymptotic slowdown with respect to parameters such as system size and precision. This tradeoff is analogous to the use of low-order Trotter methods in Hamiltonian simulation, and is to some extent unavoidable. Nevertheless, there is considerable potential for improvement to reduce the cost. For instance, if backward Hamiltonian evolution can be efficiently implemented on the device, one may employ a higher-order Trotterization scheme to simulate $\Phi$, thereby substantially reducing the asymptotic gate count. This refinement further enhances the algorithms practicality for early fault-tolerant implementations across diverse quantum hardware platforms.

In the hardware implementation, conditioned on sampling a pure eigenstate from $\rho_E$ and on a measurement outcome for the discarded ancilla, each iteration of our algorithm maps a pure system state to a pure system state. This process naturally induces a classical Markov chain on the space of state vectors, with each application of $\Phi$ corresponding to a single transition step. Consequently, this protocol offers a quantum-inspired approach to sampling from Gibbs ensembles using classical state-vector or tensor-network simulations, in regimes where both the states and transitions admit efficient representations. This can also lead to a hybrid workflow: using classical resources to pre-screen couplings and frequency windows, followed by quantum preparation at scales where classical methods saturate.

In our implementation, the filter function has a narrow width $\sigma^{-1}$ in the frequency domain. This choice simplifies the fixed-point error analysis but may increase the total simulation time. Allowing broader energy transitions per channel application, as in the protocols of \cite{Ding2024,Ding_2025} for ground- and thermal-state preparation, could lead to a more efficient implementation. However, this also introduces new challenges, such as demonstrating that the approximating Lindbladian dynamics still fixes the target state.
We hope this work will motivate further research on optimized filter design, generalized mixing-time analysis, and experimental demonstrations of dissipative algorithms for quantum state preparation.

\vspace{0.5em}
\frev{\emph{Acknowledgments--}
The authors thank Anthony Chen, Dominik Hahn, Jiaqing Jiang, Jerome Lloyd, Siddharth Parameswaran, and Oles Shtanko for helpful discussions. The Institute for Quantum Information and Matter is an NSF Physics Frontiers Center. L.L. is a Simons Investigator in Mathematics.}

\vspace{0.5em}
\frev{\emph{Funding Statement--}
This material is based upon work supported by the U.S. Department of Energy, Office of Science, Accelerated Research in Quantum Computing Centers, Quantum Utility through Advanced Computational Quantum Algorithms, grant no. DE-SC0025572 (J.P., L.L.) and Fundamental Algorithmic Research toward Quantum Utility, grant no. DE‐SC0025535 (J.P.). Additional support is from the U.S. Department of Energy, Office of Science, National Quantum Information Science Research Centers, Quantum Systems Accelerator (Z.D., Y.Z., J.P., L.L.) and the National Science Foundation, grant no. PHY-2317110 (Y.Z., J.P.).}

\vspace{0.5em}
\frev{\emph{Author contributions--} Z.D. and Y.Z. developed the original algorithm. Z.D., Y.Z., and L.L. carried out the theoretical analysis. J.P. and L.L. supervised the project. All authors contributed to the discussion and interpretation of the results and to the preparation and revision of the manuscript.}

\vspace{0.5em}
\frev{\emph{Competing interests--} The authors declare no competing interests.}

\vspace{0.5em}
\frev{\emph{Tables--} There are no tables in this manuscript.}

\clearpage
\section{Methods}

Here we first present the implementation details of the proposed quantum algorithm and then summarize the proof roadmap and the core analytical techniques supporting the main theoretical result.

\section{Algorithm Implementation}

To facilitate early fault-tolerant implementation, we approximate the time-dependent Hamiltonian in $\Phi$ using a second-order Trotter formula with a small time step $\tau$. We choose not to employ higher-order Trotter formulas so as to avoid implementing backward Hamiltonian simulations of $H$. At each time step $n$, the frequency parameter $\omega$ and the operator $A_S$ are independently sampled. Let $M=\left\lceil 2T/\tau\right\rceil$, and $f_m=f((m+1/2)\tau-T)$. According to the preceding description, the evolution from $\rho_n$ to $\rho_{n+1}$ is implemented as follows:
\begin{equation}\label{eqn:Phi_alpha_approx}
\begin{aligned}
&\rho_{n+1} = \Phi^{\rm approx}(\rho_n)\\
&:=\mathbb{E}_{A_S,\omega}\left(\operatorname{Tr}_E \left[ \Pi^M_{m=0}U^\alpha_{m\tau}\left( \rho_n \otimes \rho_E \right) \left(\Pi^M_{m=0}U^\alpha_{m\tau}\right)^\dagger \right]\right)\,,
\end{aligned}
\end{equation}
where
\[
U^\alpha_{m\tau}=W_{A_S,m}(\tau/2)\exp(-iH\tau)\exp(i\omega \tau Z/2)W_{A_S,m}(\tau/2)
\]
implements the Trotterized approximation to $U^{\alpha}$ with
\begin{equation}\label{eqn:Wnformula}
  W_{A_S,m}(\tau/2)=\exp(-i\alpha f_m( A_S \otimes B_E+A^\dagger_S\otimes B^\dagger_E)\tau/2)\,.
\end{equation}
In practice, if the system Hamiltonian $H$ consists of noncommuting terms, one would further Trotterize $\exp(-iH\tau)$. \rrev{When first- or second-order Trotterization is used for the Hamiltonian evolution, the simulation involves only forward time evolution, which is well suited to analog quantum simulators. When backward Hamiltonian evolution can also be implemented efficiently, as in digital quantum computers, the simulation cost can be further reduced by using higher-order Trotterization or quantum singular value transformation methods to simulate $\Phi$ with improved precision.} The detailed complexity analysis of the second order Trotter can be found in Supp. S6 D.

\section{Key techniques}

The analysis of our method can be divided into three main steps: (1) Derive the effective Lindblad dynamics that approximates $\Phi$; (2) Show that the fixed point of the Lindblad dynamics is close to the target thermal (or ground) state, and hence also serves as the fixed point of $\Phi$; and (3) Establish an upper bound on the mixing time of $\Phi$.

\medskip
\noindent\emph{Derivation of effective Lindblad dynamics:}

We show that the quantum map $\Phi$ in the weak coupling regime $\alpha \ll 1$ can be approximated by an effective Lindblad dynamics up to a unitary transformation.

Define the time evolution operator by $U_S(t) := \exp(-iHt)$, and the associated superoperator by $\mc{U}_S(t)[\rho]=U_S(t)\rho U^\dagger_S(t)$. Given a jump operator $V$,  we define the associated dissipative operator as
\begin{equation}
\mathcal{D}_{V}(\rho)=V\rho V^\dagger- \frac12\{V^\dagger V,\rho\}\,.
\end{equation}
The following theorem establishes that the quantum channel $\Phi$ can be approximated by an effective Lindblad dynamics:

\begin{thm}[Informal]\label{thm:char_Phi_alpha} Under the choice of $H_E, A_S, B_E, f(t), g(\omega)$ in the main text, $\rho_{n+1}$ can be expressed as
\begin{equation}\label{eqn:Phi_map_approx}
\rho_{n+1}=\mc{U}_S(T)\,\circ\,\exp(\mathcal{L} \alpha^2)\,\circ\,\mc{U}_S(T)[\rho_n]+\mc{O}(\alpha^4T^4\|f\|^4_{L^\infty})\,.
\end{equation}
Here, $\mathcal{L}$ is a Lindbladian operator.
\begin{equation}\label{eqn:lindbladian_operator}
\begin{aligned}
\mc{L}(\rho)=&\mathbb{E}_{A_S}\left(\int^\infty_{-\infty} -i\left[g(\omega)H_{\mathrm{LS},A_S}(\omega),\rho\right]\right.\\
&\left.+\gamma(\omega)\mathcal{D}_{V_{A_S,f,T}(\omega)}(\rho)\mathrm{d}\omega\right)\,.
\end{aligned}
\end{equation}
Here, the Lamb shift term $H_{\mathrm{LS},A_S}(\omega)$ is a Hermitian matrix. $\gamma(\omega)=(g(\omega)+g(-\omega))/(1+\exp(\beta\omega))$ when $\beta<\infty$, and $\gamma(\omega)=(g(\omega)+g(-\omega))\textbf{1}_{\omega<0}+\rev{g(0)\textbf{1}_{\omega=0}}$ when $\beta=\infty$. The jump operator $V_{A_S,f,T}(\omega)$ is defined as
\begin{equation}\label{eqn:V_sf}
V_{A_S,f,T}(\omega)=\int^T_{-T}f(t)A_S(t)\exp(-i\omega t)\mathrm{d}t.
\end{equation}
\rrev{where $A_S(t)$ is the Heisenberg-picture evolution of the coupling operator $A_S$: $A_S(t)=\exp(iHt)A_S\exp(-iHt)$.}
\end{thm}

In above theorem, we omit the $f,T$ dependence of $H_{\mathrm{LS},A_S}(\omega)$ in the subindex for simplicity. The rigorous version of above theorem with the specific form of $\mc{L}$  can be found in Supp. S3 as Theorem S2.
The $\mathcal{O}(\alpha^2)$ term is obtained by analyzing the second-order term in the Dyson expansion of the time evolution operators, and the $\mathcal{O}(\alpha^4)$ remainder bounds the sum of higher-order terms.
The parameter $\alpha^2$ can be interpreted as the effective evolution time under the Lindbladian operator $\mathcal{L}$. Each application of $\Phi$ approximately evolves the system, up to the unitary transformation $\mc{U}_S(T)$, under the Lindbladian operator $\mathcal{L}$ for time $\alpha^2$ with error scaling as $\mc{O}(\alpha^4)$. Therefore, up to the unitary transformation $\mathcal{U}_S(T)$, the map $\Phi$ serves as a first-order approximation to the Lindblad dynamics $\partial_t \rho = \mathcal{L}(\rho)$.

\medskip
\noindent\emph{\rrev{Choice of interaction terms:}}

\rrev{In our algorithm, it is important to properly choose $H_E$, $\rho_E$, $A_S$, $B_E$, $\alpha$, $f(t)$, and $g(\omega)$ to ensure that the fixed point of $\Phi$ closely approximates the target thermal or ground state and that the resulting dynamics has a favorable mixing time. First, $H_E$ specifies the single-ancilla bath system. The operators $A_S$ and $B_E$ jointly induce energy transitions through the system--bath interaction, with the transition energy set by the parameter $\omega$ in $H_E$, thereby driving transitions between eigenstates of the system Hamiltonian. The bath is initialized in the thermal state $\rho_E$ of $H_E$, which ensures the correct balance between upward and downward transitions, as reflected by $\gamma(\omega)$ in~\eqref{eqn:lindbladian_operator}, and thus yields an approximate quantum detailed balance condition. The choice of coupling operator $A_S$ determines which transitions between system eigenvectors are enabled. \rev{In particular, if all operators in $\mathcal{A}$ commute with $H$, then the interaction cannot induce transitions between distinct energy eigenspaces, so such a choice is not suitable for thermalization or cooling. More generally, for the corresponding effective Lindbladian with a full-rank fixed point, a standard sufficient condition for uniqueness is that the commutant of the coherent term together with the jump operators is trivial~\cite{wolf2012quantum,gilyen2024quantum,Zhang_2024}.} When prior knowledge of the eigenstate structure is available, one can design $A_S$ to optimally induce transitions between relevant eigenstates. More generally, the Lindbladian literature provides generic choices of $A_S$ that yield efficient local transitions for broad classes of systems, such as $\mathcal{A}=\{\pm X_i,\pm Y_i,\pm Z_i\}$ for local spin systems \cite{temme2015fast, temme2014hypercontractivity} and $\mathcal{A}=\{\pm c_i,\pm c_i^\dagger\}$ for fermionic systems \cite{tong2024fast, temme2014hypercontractivity}. The Fourier transform $\hat{f}$ of the filter function controls the allowed energy-transition window through the jump operator $V_{A_S,f,T}(\omega)$ in~\eqref{eqn:V_sf}. As shown in~\cref{prop:first_step}, to guarantee that the fixed point is close to the target state, we need to choose $\sigma$ sufficiently large, which generates a narrow energy-transition window around $\omega$.}

\rrev{Therefore, to keep the mixing time bounded as $\sigma$ increases, we must carefully design the system--bath interaction so that the induced energy transitions can cover a wide energy range; this is ensured by the choice of $g(\omega)$. The choice of $g(\omega)$ is flexible, as long as it provides a sufficiently broad energy-transition window. \rev{This choice does not require computing the spectrum of $H$. For a local Hamiltonian $H=\sum_j h_j$ with $\|h_j\|=\mathcal{O}(1)$, one may simply take $\omega_{\max}=2\sum_j\|h_j\|=\mathcal{O}(N)$, which upper bounds all Bohr frequencies.} When $A_S$ induces only local transitions, \rev{we expect that it is sufficient to choose} $g(\omega)$ to be uniform on a smaller window with $\omega_{\max}=\mathcal{O}(1)$ to cover the local energy changes generated by $A_S$. In Supp. S7, we use a single qubit example to explicitly demonstrate the importance of this design choice.}

\medskip
\noindent\emph{Approximate fixed point:}

For a quantum channel $\Phi$ with a unique fixed point $\rho_{\text{fix}}$, a state $\rho$ is close to $\rho_{\text{fix}}$ if $\Phi(\rho)\approx\rho$. To ensure that $\rho_{\text{fix}}$ approximates the target thermal state $\rho_\beta$, it suffices to bound $\|\Phi(\rho_\beta)-\rho_\beta\|_1$. According to~\cref{thm:char_Phi_alpha} and the fact that $\mc{U}_S(T)$ preserves the targe state, we can show
\[
\|\Phi(\rho_\beta)-\rho_\beta\|_1 \approx\alpha^2\left\|\mc{L}(\rho_\beta)\right\|_1,\quad \alpha\ll 1\,,
\]
where $\mc{L}$ is the approximation Lindblad operator that appears in~\cref{thm:char_Phi_alpha}. Thus, it suffices to show that the Lindblad dynamics approximately preserve the thermal or ground state. This forms the most technical component of the proof. For the thermal state, we demonstrate that the Lindbladian operator approximately satisfies the detailed balance condition, up to an additional coherent term that nearly commutes with the thermal state. For the ground state, assuming the Hamiltonian has a spectral gap $\Delta$, we directly show that the operator approximately preserves the ground state.



The following proposition provides fixed point error bounds in preparing both thermal and ground states.
\rev{Guidance on parameter choices and a more detailed formulation of this proposition can be found in Theorem S4 (Theorem S7) and Theorem S5 (Theorem S11).}

\begin{prop}\label{prop:first_step} For the quantum channel defined in \cref{eqn:Phi_alpha}, For any inverse temperature $\beta>0$, we have
\begin{equation}\label{eqn:fixed_point_error_1_informal}
\|\rho_{\rm fix}(\Phi)-\rho_\beta\|_1
  = \widetilde{\mathcal{O}}\!\left(\Bigl(\tfrac{\beta}{\sigma}+\alpha^2\sigma^2\Bigr)\alpha^2\tau_{\rm mix}(\epsilon)\right).
\end{equation}
For ground state that $\beta=\infty$, if the Hamiltonian $H$ has spectral gap $\Delta$, we also have
\begin{equation}\label{eqn:fixed_point_error_2_informal}
\|\rho_{\rm fix}(\Phi)-\rho_\infty\|_1
  = \widetilde{\mathcal{O}}\!\left(\left(e^{-\sigma^2\Delta^2/32}+\alpha^2\sigma^2\right)\alpha^2\tau_{\rm mix}(\epsilon)\right).
\end{equation}
\end{prop}

\medskip
\noindent\emph{Mixing time bound:}

The ground-state and thermal-state cases require different proof strategies.

We start with thermal state case. When two quantum channels $\Phi_1$ and $\Phi$ are close, i.e., $\|\Phi_1-\Phi\|_{1\rightarrow1}\approx0$, their mixing times $\tau_{1,\text{mix}}(\epsilon),\tau_{\text{mix}}(\epsilon)$ are  also close. We therefore define the quantum channel $\Phi_1$
\[
\Phi_1=\mc{U}_S(T)\,\circ\,\exp\left(\mathcal{M}\alpha^2\right)\,\circ\,\mc{U}_S(T)\,
\]
where $\mathcal{M}$ contains a dissipative operator that approximately satisfies the quantum detailed balance condition, together with a coherent term that commutes with the thermal state. In the free-fermionic example, we show that the dissipative part approximately satisfies the KMS detailed balance condition and possesses a nonvanishing spectral gap of order $\mathrm{poly}(1/N)$ as $\sigma \to \infty$. For commuting local Hamiltonians, existing results have established that the Davies generator exhibits a bounded mixing time in the high-temperature regime. We note that even though the dissipative part possesses a spectral gap, existing techniques for analyzing the mixing time cannot be directly applied due to the presence of the coherent term and the accompanying unitary evolution. To address this challenge, we establish the contractivity of the channel under a novel weighted Hilbert–Schmidt norm, $\|\rho_{\beta}^{-1/4}[\cdot]\rho_{\beta}^{-1/4}\|_2$, inspired by~\cite{ChenKastoryanoBrandaoEtAl2023}. We prove contraction of the quantum channel $\exp\left(\mathcal{M}\alpha^2\right)$ under the weighted Hilbert-Schmidt norm and this contraction still holds in the presence of the unitary evolution $\mc{U}_S(T)$, and therefore also holds for the map $\Phi_1$. The Lindbladian $\mathcal{M}$ is chosen to both preserve the thermal state $\rho_\beta$ and approximate the effective Lindbladian $\mathcal{L}$. Hence, the mixing time of $\Phi_1$ is close to that of $\Phi$. Using this framework, we rigorously demonstrate that, at high temperature, the mixing time for the single-qubit case, free-fermion systems, and commuting local Hamiltonians scales polynomially with $N$ and remains independent of $\sigma$. The detailed results are summarized in Supp. S6.

For the ground state case, we only focus on the free fermion Hamiltonain and adopt the strategy from~\cite[Section IV]{zhan2025rapidquantumgroundstate}, which analyzes the Heisenberg evolution of the number operator. Following the argument in~\cite[Section IV]{zhan2025rapidquantumgroundstate}, the convergence of the Lindblad dynamics to the ground state can be established by showing that the expectation of the number operator converges to zero. Moreover, since the unitary evolution commutes with the number operator, it does not affect this convergence. Finally, the convergence of the number operator can be directly related to the trace distance between the current state and the ground state using the Fuchs--van de Graaf inequality; see Supp. S8 for details.

\vspace{0.5em}
\frev{\emph{Data availability--} No datasets were generated or analyzed during the current study.}


\begin{thebibliography}{58}
\providecommand{\natexlab}[1]{#1}
\providecommand{\url}[1]{\texttt{#1}}
\expandafter\ifx\csname urlstyle\endcsname\relax
  \providecommand{\doi}[1]{doi: #1}\else
  \providecommand{\doi}{doi: \begingroup \urlstyle{rm}\Url}\fi

\bibitem[Tilly et~al.(2022)Tilly, Chen, Cao, Picozzi, Setia, Li, Grant, Wossnig, Rungger, Booth, and Tennyson]{tilly_variational_2022}
Jules Tilly, Hongxiang Chen, Shuxiang Cao, Dario Picozzi, Kanav Setia, Ying Li, Edward Grant, Leonard Wossnig, Ivan Rungger, George~H. Booth, and Jonathan Tennyson.
\newblock The variational quantum eigensolver: A review of methods and best practices.
\newblock \emph{Phys. Rep.}, 986:\penalty0 1--128, 2022.
\newblock ISSN 0370-1573.
\newblock \doi{https://doi.org/10.1016/j.physrep.2022.08.003}.
\newblock URL \url{https://www.sciencedirect.com/science/article/pii/S0370157322003118}.

\bibitem[Farhi et~al.(2000)Farhi, Goldstone, Gutmann, and Sipser]{farhi2000}
Edward Farhi, Jeffrey Goldstone, Sam Gutmann, and Michael Sipser.
\newblock Quantum computation by adiabatic evolution.
\newblock \emph{arXiv:quant-ph/0001106}, 2000.

\bibitem[Albash and Lidar(2018)]{RevModPhys.90.015002}
Tameem Albash and Daniel~A. Lidar.
\newblock Adiabatic quantum computation.
\newblock \emph{Rev. Mod. Phys.}, 90:\penalty0 015002, Jan 2018.
\newblock \doi{10.1103/RevModPhys.90.015002}.
\newblock URL \url{https://link.aps.org/doi/10.1103/RevModPhys.90.015002}.

\bibitem[Poulin and Wocjan(2009)]{PoulinWocjan2009}
David Poulin and Pawel Wocjan.
\newblock Preparing ground states of quantum many-body systems on a quantum computer.
\newblock \emph{Phys. Rev. Lett.}, 102:\penalty0 130503, 2009.

\bibitem[Gily{\'e}n et~al.(2019)Gily{\'e}n, Su, Low, and Wiebe]{GilyenSuLowEtAl2019}
Andr{\'a}s Gily{\'e}n, Yuan Su, Guang~Hao Low, and Nathan Wiebe.
\newblock Quantum singular value transformation and beyond: exponential improvements for quantum matrix arithmetics.
\newblock In \emph{Proceedings of the 51st Annual ACM SIGACT Symposium on Theory of Computing}, pages 193--204, 2019.

\bibitem[Lin and Tong(2020)]{LinTong2020a}
Lin Lin and Yu~Tong.
\newblock Near-optimal ground state preparation.
\newblock \emph{Quantum}, 4:\penalty0 372, 2020.

\bibitem[Motta et~al.(2020)Motta, Sun, Tan, O’Rourke, Ye, Minnich, Brandão, and Chan]{motta_determining_2020}
Mario Motta, Chong Sun, Adrian T.~K. Tan, Matthew~J. O’Rourke, Erika Ye, Austin~J. Minnich, Fernando G. S.~L. Brandão, and Garnet Kin-Lic Chan.
\newblock Determining eigenstates and thermal states on a quantum computer using quantum imaginary time evolution.
\newblock \emph{Nat. Phys.}, 16\penalty0 (2):\penalty0 205--210, February 2020.
\newblock ISSN 1745-2481.
\newblock \doi{10.1038/s41567-019-0704-4}.
\newblock URL \url{https://doi.org/10.1038/s41567-019-0704-4}.

\bibitem[Verstraete et~al.(2009)Verstraete, Wolf, and Cirac]{VerstraeteWolfCirac2009}
Frank Verstraete, Michael~M. Wolf, and I.~Cirac.
\newblock {Quantum computation and quantum-state engineering driven by dissipation}.
\newblock \emph{Nat. Phys.}, 5\penalty0 (9):\penalty0 633--636, 2009.

\bibitem[Temme et~al.(2011)Temme, Osborne, Vollbrecht, Poulin, and Verstraete]{Temme_2011}
K.~Temme, T.~J. Osborne, K.~G. Vollbrecht, D.~Poulin, and F.~Verstraete.
\newblock Quantum {Metropolis} sampling.
\newblock \emph{Nature}, 471\penalty0 (7336):\penalty0 87–90, March 2011.
\newblock ISSN 1476-4687.
\newblock \doi{10.1038/nature09770}.
\newblock URL \url{http://dx.doi.org/10.1038/nature09770}.

\bibitem[Kitaev et~al.(2002)Kitaev, Shen, and Vyalyi]{KitaevShenVyalyi2002}
Alexei~Yu Kitaev, Alexander Shen, and Mikhail~N Vyalyi.
\newblock \emph{Classical and quantum computation}.
\newblock American Mathematical Soc., 2002.

\bibitem[Barahona(1982)]{Barahona1982}
Francisco Barahona.
\newblock On the computational complexity of ising spin glass models.
\newblock \emph{J. Phys. A: Math. Gen.}, 15:\penalty0 3241, 1982.

\bibitem[Sly(2010)]{Sly2010}
Allan Sly.
\newblock Computational transition at the uniqueness threshold.
\newblock In \emph{2010 IEEE 51st Annual Symposium on Foundations of Computer Science}, pages 287--296, 2010.


\bibitem[Cubitt(2023)]{Cubitt2023}
Toby~S. Cubitt.
\newblock Dissipative ground state preparation and the dissipative quantum eigensolver.
\newblock \emph{arXiv:2303.11962}, 2023.

\bibitem[Mi et~al.(2024)Mi, Michailidis, Shabani, Miao, Klimov, Lloyd, Rosenberg, Acharya, Aleiner, Andersen, Ansmann, Arute, Arya, Asfaw, Atalaya, Bardin, Bengtsson, Bortoli, Bourassa, Bovaird, Brill, Broughton, Buckley, Buell, Burger, Burkett, Bushnell, Chen, Chiaro, Chik, Chou, Cogan, Collins, Conner, Courtney, Crook, Curtin, Dau, Debroy, {Del Toro Barba}, Demura, {Di Paolo}, Drozdov, Dunsworth, Erickson, Faoro, Farhi, Fatemi, Ferreira, Burgos, Forati, Fowler, Foxen, Genois, Giang, Gidney, Gilboa, Giustina, Gosula, Gross, Habegger, Hamilton, Hansen, Harrigan, Harrington, Heu, Hoffmann, Hong, Huang, Huff, Huggins, Ioffe, Isakov, Iveland, Jeffrey, Jiang, Jones, Juhas, Kafri, Kechedzhi, Khattar, Khezri, Kieferov{\'{a}}, Kim, Kitaev, Klots, Korotkov, Kostritsa, Kreikebaum, Landhuis, Laptev, Lau, Laws, Lee, Lee, Lensky, Lester, Lill, Liu, Locharla, Malone, Martin, McClean, McEwen, Mieszala, Montazeri, Morvan, Movassagh, Mruczkiewicz, Neeley, Neill, Nersisyan, Newman, Ng, Nguyen, Nguyen, Niu, O'Brien, Opremcak, Petukhov, Potter, Pryadko, Quintana, Rocque, Rubin, Saei, Sank, Sankaragomathi, Satzinger, Schurkus, Schuster, Shearn, Shorter, Shutty, Shvarts, Skruzny, Smith, Somma, Sterling, Strain, Szalay, Torres, Vidal, Villalonga, Heidweiller, White, Woo, Xing, Yao, Yeh, Yoo, Young, Zalcman, Zhang, Zhu, Zobrist, Neven, Babbush, Bacon, Boixo, Hilton, Lucero, Megrant, Kelly, Chen, Roushan, Smelyanskiy, and Abanin]{MiMichailidisShabaniEtAl2024}
X.~Mi, A.~A. Michailidis, S.~Shabani, K.~C. Miao, P.~V. Klimov, J.~Lloyd, E.~Rosenberg, R.~Acharya, I.~Aleiner, T.~I. Andersen, M.~Ansmann, F.~Arute, K.~Arya, A.~Asfaw, J.~Atalaya, J.~C. Bardin, A.~Bengtsson, G.~Bortoli, A.~Bourassa, J.~Bovaird, L.~Brill, M.~Broughton, B.~B. Buckley, D.~A. Buell, T.~Burger, B.~Burkett, N.~Bushnell, Z.~Chen, B.~Chiaro, D.~Chik, C.~Chou, J.~Cogan, R.~Collins, P.~Conner, W.~Courtney, A.~L. Crook, B.~Curtin, A.~G. Dau, D.~M. Debroy, A.~{Del Toro Barba}, S.~Demura, A.~{Di Paolo}, I.~K. Drozdov, A.~Dunsworth, C.~Erickson, L.~Faoro, E.~Farhi, R.~Fatemi, V.~S. Ferreira, L.~F. Burgos, E.~Forati, A.~G. Fowler, B.~Foxen, {\'{E}}.~Genois, W.~Giang, C.~Gidney, D.~Gilboa, M.~Giustina, R.~Gosula, J.~A. Gross, S.~Habegger, M.~C. Hamilton, M.~Hansen, M.~P. Harrigan, S.~D. Harrington, P.~Heu, M.~R. Hoffmann, S.~Hong, T.~Huang, A.~Huff, W.~J. Huggins, L.~B. Ioffe, S.~V. Isakov, J.~Iveland, E.~Jeffrey, Z.~Jiang, C.~Jones, P.~Juhas, D.~Kafri, K.~Kechedzhi, T.~Khattar, M.~Khezri, M.~Kieferov{\'{a}}, S.~Kim, A.~Kitaev, A.~R. Klots, A.~N. Korotkov, F.~Kostritsa, J.~M. Kreikebaum, D.~Landhuis, P.~Laptev, K.-M. Lau, L.~Laws, J.~Lee, K.~W. Lee, Y.~D. Lensky, B.~J. Lester, A.~T. Lill, W.~Liu, A.~Locharla, F.~D. Malone, O.~Martin, J.~R. McClean, M.~McEwen, A.~Mieszala, S.~Montazeri, A.~Morvan, R.~Movassagh, W.~Mruczkiewicz, M.~Neeley, C.~Neill, A.~Nersisyan, M.~Newman, J.~H. Ng, A.~Nguyen, M.~Nguyen, M.~Y. Niu, T.~E. O'Brien, A.~Opremcak, A.~Petukhov, R.~Potter, L.~P. Pryadko, C.~Quintana, C.~Rocque, N.~C. Rubin, N.~Saei, D.~Sank, K.~Sankaragomathi, K.~J. Satzinger, H.~F. Schurkus, C.~Schuster, M.~J. Shearn, A.~Shorter, N.~Shutty, V.~Shvarts, J.~Skruzny, W.~C. Smith, R.~Somma, G.~Sterling, D.~Strain, M.~Szalay, A.~Torres, G.~Vidal, B.~Villalonga, C.~V. Heidweiller, T.~White, B.~W.~K. Woo, C.~Xing, Z.~J. Yao, P.~Yeh, J.~Yoo, G.~Young, A.~Zalcman, Y.~Zhang, N.~Zhu, N.~Zobrist, H.~Neven, R.~Babbush, D.~Bacon, S.~Boixo, J.~Hilton, E.~Lucero, A.~Megrant, J.~Kelly, Y.~Chen, P.~Roushan, V.~Smelyanskiy, and D.~A. Abanin.
\newblock {Stable quantum-correlated many-body states through engineered dissipation}.
\newblock \emph{Science}, 383\penalty0 (6689):\penalty0 1332--1337, 2024.

\bibitem[Rall et~al.(2023)Rall, Wang, and Wocjan]{RallWangWocjan2023}
Patrick Rall, Chunhao Wang, and Pawel Wocjan.
\newblock Thermal state preparation via rounding promises.
\newblock \emph{Quantum}, 7:\penalty0 1132, 2023.

\bibitem[Chen et~al.(2023{\natexlab{a}})Chen, Kastoryano, Brand{\~a}o, and Gily{\'e}n]{ChenKastoryanoBrandaoEtAl2023}
Chi-Fang Chen, Michael~J Kastoryano, Fernando~GSL Brand{\~a}o, and Andr{\'a}s Gily{\'e}n.
\newblock Quantum thermal state preparation.
\newblock \emph{arXiv:2303.18224}, 2023{\natexlab{a}}.

\bibitem[Chen et~al.(2023{\natexlab{b}})Chen, Kastoryano, and Gily{\'e}n]{ChenKastoryanoGilyen2023}
Chi-Fang Chen, Michael~J Kastoryano, and Andr{\'a}s Gily{\'e}n.
\newblock An efficient and exact noncommutative quantum {Gibbs} sampler.
\newblock \emph{arXiv:2311.09207}, 2023{\natexlab{b}}.

\bibitem[Ding et~al.(2025)Ding, Li, and Lin]{Ding_2025}
Zhiyan Ding, Bowen Li, and Lin Lin.
\newblock Efficient quantum {Gibbs} samplers with {Kubo--Martin--Schwinger} detailed balance condition.
\newblock \emph{Commun. Math. Phys.}, 406\penalty0 (3):\penalty0 67, 2025.

\bibitem[Ding et~al.(2024{\natexlab{a}})Ding, Chen, and Lin]{Ding2024}
Zhiyan Ding, Chi-Fang Chen, and Lin Lin.
\newblock Single-ancilla ground state preparation via {Lindbladians}.
\newblock \emph{Phys. Rev. Research}, 6:\penalty0 033147, 2024{\natexlab{a}}.

\bibitem[Li et~al.(2024)Li, Zhan, and Lin]{li2024dissipative}
Hao-En Li, Yongtao Zhan, and Lin Lin.
\newblock Dissipative ground state preparation in ab initio electronic structure theory.
\newblock \emph{arXiv:2411.01470}, 2024.

\bibitem[Zhan et~al.(2025)Zhan, Ding, Huhn, Gray, Preskill, Chan, and Lin]{zhan2025rapidquantumgroundstate}
Yongtao Zhan, Zhiyan Ding, Jakob Huhn, Johnnie Gray, John Preskill, Garnet Kin-Lic Chan, and Lin Lin.
\newblock Rapid quantum ground state preparation via dissipative dynamics.
\newblock \emph{arXiv:2503.15827}, 2025.
\newblock URL \url{https://arxiv.org/abs/2503.15827}.

\bibitem[Hahn et~al.(2025{\natexlab{a}})Hahn, Sweke, Deshpande, and Shtanko]{hahn2025}
Dominik Hahn, Ryan Sweke, Abhinav Deshpande, and Oles Shtanko.
\newblock Efficient quantum {Gibbs} sampling with local circuits.
\newblock \emph{arXiv:2506.04321}, 2025{\natexlab{a}}.

\bibitem[Gilyén et~al.(2024)Gilyén, Chen, Doriguello, and Kastoryano]{gilyen2024}
Andr\'as Gilyén, Chi-Fang Chen, Joao~F. Doriguello, and Michael~J. Kastoryano.
\newblock Quantum generalizations of {Glauber and Metropolis} dynamics.
\newblock \emph{arXiv:2405.20322}, 2024.

\bibitem[Jiang and Irani(2024)]{jiang2024quantum}
Jiaqing Jiang and Sandy Irani.
\newblock Quantum {Metropolis} sampling via weak measurement.
\newblock \emph{arXiv:2406.16023}, 2024.


\bibitem[Eder et~al.(2025)Eder, Fin{\v{z}}gar, Braun, and Mendl]{eder2025quantum}
Peter~J Eder, Jernej~Rudi Fin{\v{z}}gar, Sarah Braun, and Christian~B Mendl.
\newblock Quantum dissipative search via {Lindbladians}.
\newblock \emph{Phys. Rev. A}, 111\penalty0 (4):\penalty0 042430, 2025.


\bibitem[Motlagh et~al.(2024)Motlagh, Zini, Arrazola, and Wiebe]{motlagh2024ground}
Danial Motlagh, Modjtaba~Shokrian Zini, Juan~Miguel Arrazola, and Nathan Wiebe.
\newblock Ground state preparation via dynamical cooling.
\newblock \emph{arXiv preprint arXiv:2404.05810}, 2024.



\bibitem[Lloyd et~al.(2025)Lloyd, Michailidis, Mi, Smelyanskiy, and Abanin]{Lloyd2025Quasiparticle}
Jerome Lloyd, Alexios~A. Michailidis, Xiao Mi, Vadim Smelyanskiy, and Dmitry~A. Abanin.
\newblock Quasiparticle cooling algorithms for quantum many-body state preparation.
\newblock \emph{PRX Quantum}, 6:\penalty0 010361, March 2025.
\newblock \doi{10.1103/PRXQuantum.6.010361}.
\newblock URL \url{https://doi.org/10.1103/PRXQuantum.6.010361}.



\bibitem[Bardet et~al.(2023)Bardet, Capel, Gao, Lucia, P{\'e}rez-Garc{\'\i}a, and Rouz{\'e}]{BardetCapelGaoEtAl2023}
Ivan Bardet, {\'A}ngela Capel, Li~Gao, Angelo Lucia, David P{\'e}rez-Garc{\'\i}a, and Cambyse Rouz{\'e}.
\newblock Rapid thermalization of spin chain commuting {Hamiltonians}.
\newblock \emph{Phys. Rev. Lett.}, 130\penalty0 (6):\penalty0 060401, 2023.

\bibitem[Rouz\'{e} et~al.(2025)Rouz\'{e}, Fran\c{c}a, and Alhambra]{rouz2024}
Cambyse Rouz\'{e}, Daniel~Stilck Fran\c{c}a, and \'{A}lvaro~M. Alhambra.
\newblock Efficient thermalization and universal quantum computing with quantum gibbs samplers.
\newblock In \emph{STOC 25}, page 1488–1495, 2025.
\newblock ISBN 9798400715105.
\newblock \doi{10.1145/3717823.3718268}.
\newblock URL \url{https://doi.org/10.1145/3717823.3718268}.

\bibitem[Ding et~al.(2024{\natexlab{b}})Ding, Li, Lin, and Zhang]{DingLiLinZhang2024}
Zhiyan Ding, Bowen Li, Lin Lin, and Ruizhe Zhang.
\newblock Polynomial-time preparation of low-temperature {G}ibbs states for {2D Toric Code}.
\newblock \emph{arXiv:2410.01206}, 2024{\natexlab{b}}.

\bibitem[Kochanowski et~al.(2024)Kochanowski, Alhambra, Capel, and Rouz{\'e}]{kochanowski2024rapid}
Jan Kochanowski, Alvaro~M Alhambra, Angela Capel, and Cambyse Rouz{\'e}.
\newblock Rapid thermalization of dissipative many-body dynamics of commuting {H}amiltonians.
\newblock \emph{Commun. Math. Phys.}, 2024.

\bibitem[Rouz{\'e} et~al.(2024)Rouz{\'e}, Fran{\c{c}}a, and Alhambra]{rouze2024optimal}
Cambyse Rouz{\'e}, Daniel~Stilck Fran{\c{c}}a, and {\'A}lvaro~M Alhambra.
\newblock Optimal quantum algorithm for {Gibbs} state preparation.
\newblock \emph{arXiv:2411.04885}, 2024.

\bibitem[Tong and Zhan(2025)]{tong2024fast}
Yu~Tong and Yongtao Zhan.
\newblock Fast mixing of weakly interacting fermionic systems at any temperature.
\newblock \emph{PRX Quantum}, 6:\penalty0 030301, Jul 2025.
\newblock \doi{10.1103/h1dx-ps5p}.
\newblock URL \url{https://link.aps.org/doi/10.1103/h1dx-ps5p}.

\bibitem[Štěpán Šmíd et~al.(2025)Štěpán Šmíd, Meister, Berta, and Bondesan]{smid2025polynomial}
Štěpán Šmíd, Richard Meister, Mario Berta, and Roberto Bondesan.
\newblock Polynomial time quantum {Gibbs} sampling for {Fermi-Hubbard} model at any temperature.
\newblock \emph{arXiv:2501.01412}, 2025.
\newblock URL \url{https://arxiv.org/abs/2501.01412}.

\bibitem[Cleve and Wang(2017)]{cleve_2017}
Richard Cleve and Chunhao Wang.
\newblock Efficient quantum algorithms for simulating {Lindblad} evolution.
\newblock In \emph{ICALP 2017}, volume~80, pages 17:1--17:14, 2017.

\bibitem[Li and Wang(2023)]{li_et_2023}
Xiantao Li and Chunhao Wang.
\newblock Simulating {Markovian} open quantum systems using higher-order series expansion.
\newblock In \emph{ICALP 2023}, volume 261, pages 87:1--87:20, 2023.

\bibitem[Ding et~al.(2024{\natexlab{c}})Ding, Li, and Lin]{PRXQuantum.5.020332}
Zhiyan Ding, Xiantao Li, and Lin Lin.
\newblock Simulating open quantum systems using {Hamiltonian} simulations.
\newblock \emph{PRX Quantum}, 5:\penalty0 020332, 2024{\natexlab{c}}.

\bibitem[Molpeceres et~al.(2025)Molpeceres, Lu, Cirac, and Kraus]{molpeceres2025}
Daniel Molpeceres, Sirui Lu, J.~Ignacio Cirac, and Barbara Kraus.
\newblock Quantum algorithms for cooling: a simple case study.
\newblock \emph{arXiv:2503.24330}, 2025.

\bibitem[Hahn et~al.(2025{\natexlab{b}})Hahn, Parameswaran, and Placke]{hahn2025provably}
Dominik Hahn, S.~A. Parameswaran, and Benedikt Placke.
\newblock Provably efficient quantum thermal state preparation via local driving.
\newblock \emph{arXiv:2505.22816}, 2025{\natexlab{b}}.
\newblock URL \url{https://arxiv.org/abs/2505.22816}.

\bibitem[Langbehn et~al.(2025)Langbehn, Mouloudakis, King, Menu, Gornyi, Morigi, Gefen, and Koch]{langbehn2025universal}
Josias Langbehn, George Mouloudakis, Emma King, Raphaël Menu, Igor Gornyi, Giovanna Morigi, Yuval Gefen, and Christiane~P. Koch.
\newblock Universal cooling of quantum systems via randomized measurements.
\newblock \emph{arXiv:2506.11964}, 2025.

\bibitem[Lloyd and Abanin(2025)]{lloyd2025quantumthermal}
Jerome Lloyd and Dmitry~A. Abanin.
\newblock Quantum thermal state preparation for near-term quantum processors.
\newblock \emph{arXiv:2506.21318}, 2025.

\bibitem[Scandi and Álvaro M.~Alhambra(2025)]{scandi2025thermal}
Matteo Scandi and Álvaro M.~Alhambra.
\newblock Thermalization in open many-body systems and {KMS} detailed balance.
\newblock \emph{arXiv:2505.20064}, 2025.
\newblock URL \url{https://arxiv.org/abs/2505.20064}.

\bibitem[Kuwahara et~al.(2021)Kuwahara, Alhambra, and Anshu]{PhysRevX.11.011047}
Tomotaka Kuwahara, \'Alvaro~M. Alhambra, and Anurag Anshu.
\newblock Improved thermal area law and quasilinear time algorithm for quantum gibbs states.
\newblock \emph{Phys. Rev. X}, 11:\penalty0 011047, Mar 2021.
\newblock \doi{10.1103/PhysRevX.11.011047}.
\newblock URL \url{https://link.aps.org/doi/10.1103/PhysRevX.11.011047}.

\bibitem[Slezak et~al.(2026)Slezak, Scandi, Álvaro M.~Alhambra, França, and Rouzé]{slezak2026polynomialtimethermalizationgibbssampling}
Samuel Slezak, Matteo Scandi, Álvaro M.~Alhambra, Daniel~Stilck França, and Cambyse Rouzé.
\newblock Polynomial-time thermalization and gibbs sampling from system-bath couplings.
\newblock \emph{arXiv/2601.16154}, 2026.

\bibitem[Wang and Ding(2026)]{wang2026lindbladdynamicsrigorousguarantees}
Ke~Wang and Zhiyan Ding.
\newblock Beyond lindblad dynamics: Rigorous guarantees for thermal and ground state preservation under system bath interactions.
\newblock \emph{arXiv/2512.03457}, 2026.

\bibitem[Wolf(2012)]{wolf2012quantum}
Michael~M. Wolf.
\newblock Quantum channels \& operations: Guided tour, 2012.
\newblock Lecture notes, available online.

\bibitem[Gily{\'e}n et~al.(2024)Gily{\'e}n, Chen, Doriguello, and Kastoryano]{gilyen2024quantum}
Andr{\'a}s Gily{\'e}n, Chi-Fang Chen, Joao~F Doriguello, and Michael~J Kastoryano.
\newblock Quantum generalizations of {G}lauber and {M}etropolis dynamics.
\newblock \emph{arXiv preprint arXiv:2405.20322}, 2024.

\bibitem[Zhang and Barthel(2024)]{Zhang_2024}
Yikang Zhang and Thomas Barthel.
\newblock Criteria for davies irreducibility of markovian quantum dynamics.
\newblock \emph{Journal of Physics A: Mathematical and Theoretical}, 57\penalty0 (11):\penalty0 115301, mar 2024.
\newblock \doi{10.1088/1751-8121/ad2a1e}.
\newblock URL \url{https://doi.org/10.1088/1751-8121/ad2a1e}.

\bibitem[Temme and Kastoryano(2015)]{temme2015fast}
Kristan Temme and Michael~J Kastoryano.
\newblock How fast do stabilizer hamiltonians thermalize?
\newblock \emph{arXiv preprint arXiv:1505.07811}, 2015.

\bibitem[Temme et~al.(2014)Temme, Pastawski, and Kastoryano]{temme2014hypercontractivity}
Kristan Temme, Fernando Pastawski, and Michael~J Kastoryano.
\newblock Hypercontractivity of quasi-free quantum semigroups.
\newblock \emph{Journal of Physics A: Mathematical and Theoretical}, 47\penalty0 (40):\penalty0 405303, 2014.

\end{thebibliography}
\end{document}


\title{Supplementary information for ``Simple and efficient end-to-end quantum thermal and ground state preparation"}
\author{Zhiyan Ding}
\email{zyding@umich.edu}
\thanks{These authors contributed equally to this work.}
\affiliation{Department of Mathematics, University of Michigan, Ann Arbor}
\author{Yongtao Zhan}
\email{yzhan@caltech.edu}
\thanks{These authors contributed equally to this work.}
\affiliation{Institute for Quantum Information and Matter, California Institute of Technology}
\affiliation{Division of Physics, Mathematics, and Astronomy, California Institute of Technology}
\author{John Preskill}
\email{preskill@caltech.edu}
\affiliation{Institute for Quantum Information and Matter, California Institute of Technology}
\affiliation{Division of Physics, Mathematics, and Astronomy, California Institute of Technology}
\affiliation{AWS Center for Quantum Computing}
\author{Lin Lin}
\email{linlin@math.berkeley.edu}
\affiliation{Department of Mathematics, University of California, Berkeley}
\affiliation{Applied Mathematics and Computational Research Division, Lawrence Berkeley National Laboratory}

\maketitle
\let\section\oldsection
\onecolumngrid

The supplement is organized as follows. We begin by introducing the notations and reviewing related work in~\cref{sec:notation} and~\cref{sec:related_works}. Then, we present our theoretical analysis in three steps:
\begin{itemize}
    \item \textbf{First step: Derivation of the effective Lindbladian dynamics.} The main result is stated in~\cref{thm:char_Phi_alpha_rigor}, which is a rigorous version of Theorem 4, with the proof given in~\cref{sec:proof_rhoevolution}.

    After deriving the Lindblad dynamics, we demonstrate that two close CPTP maps have close fixed points and mixing times in~\cref{sec:almost_fixed_point}~\cref{thm:almost_fixed_point}, which provides a useful tool for analyzing the fixed point and mixing time of $\Phi$.

    \item \textbf{Second step: Fixed point error bounds for thermal and ground state preparation.} The main results are presented in~\cref{sec:thermal_ground_prep}~\cref{thm:fix_thermal} and~\cref{thm:fix_ground}, corresponding to the thermal and ground states, respectively, with proofs provided in~\cref{sec:app_fix_thermal} and~\cref{sec:app_fix_gs}. Combining the results from the first two steps, we show that the fixed point of $\Phi$ is close to the target thermal or ground state when properly adjusting the parameters.

    \item \textbf{Third step: Mixing time and End-to-end efficiency analysis.} We present mixing-time results for several physically relevant models in~\cref{thm:toy_model},~\cref{thm:gs_mixing},~\cref{thm:thermal_mixing}, and ~\cref{thm:thermal_commuting_local} and derive end-to-end runtime estimates for our state preparation algorithm in~\cref{cor:complete}. The proofs of these results are collected in~\cref{sec:toy_uniform}–\cref{sec:thermal_mixing_general}.

\end{itemize}


\section{Notations and detailed balance condition}\label{sec:notation}
For a matrix $A\in\CC^{N\times N}$, let $A^*, A^T, A^{\dag}$ be the complex conjugation, transpose, and Hermitian transpose (or adjoint) of $A$, respectively. $\|A\|_p=\mathrm{Tr}\left(\left(\sqrt{A^\dagger A}\right)^{p}\right)^{1/p}$ denotes the Schatten $p$-norm. The Schatten $1$-norm $\norm{A}_1$ is also called the trace norm, the Schatten $2$-norm $\norm{A}_2$ is also called the Hilbert--Schmidt norm (or Frobenius norm for matrices), and the Schatten $\infty$-norm $\norm{A}_\infty$ is the same as the operator norm $\norm{A}$.
 The trace distance between two states $\rho,\sigma$ is $D(\rho,\sigma):=\frac12 \norm{\rho-\sigma}_1$. Given a superoperator $\Phi:\CC^{N\times N}\rightarrow \CC^{N\times N}$, we define the induced trace norm as
\[
\|\Phi\|_{1\leftrightarrow 1}=\sup_{\|A\|_1=1}\|\Phi(A)\|_{1}\,.
\]

We denote eigenstates of the Hamiltonian $H$ by $\{\ket{\psi_i}\}$ and the corresponding eigenvalues by $\{\lambda_i\}$. Each difference of eigenvalues $\lambda_i-\lambda_j$ is called a Bohr frequency, and $B(H)$ denotes the set of all Bohr frequencies.
Also, given $\nu\in B(H)$ and a matrix $A$, we define
\begin{equation}
A(\nu)=\sum_{\lambda_j-\lambda_i=\nu}\ket{\psi_j}\bra{\psi_j}A\ket{\psi_i}\bra{\psi_i},
\end{equation}
where $\ket{\psi_i}$ is an eigenvector of $H$ with eigenvalue $\lambda_i$.

Given the thermal state $\sigma_\beta\propto \exp(-\beta H)$, we define the $s$-inner product on operator space as
\[
\left\langle A,B\right\rangle_{s,\sigma_\beta}=\mathrm{Tr}\left(A^\dagger\sigma_{\beta}^{1-s}B\sigma_{\beta}^{s}\right)
\]
for $0<s<1$. Given a Lindbladian operator $\mc{L}$, we say $\mc{L}$ satisfies the KMS detailed balance condition (KMS DBC) if $\mc{L}^\dagger$ is self-adjoint under $\left\langle A,B\right\rangle_{1/2,\sigma_\beta}$ and $\mc{L}$ satisfies the GNS detailed balance condition (GNS DBC) if $\mc{L}^\dagger$ is self-adjoint under $\left\langle A,B\right\rangle_{s,\sigma_\beta}$ for any $s\neq 1/2$. We note that, if $\mc{L}$ satisfies GNS DBC, it must also satisfy KMS DBC and take a generic form of the Davies generator. Given $\mc{L}$ satisfies GNS DBC or KMS DBC, we define the spectral gap as
\[
\mathrm{Gap}(\mc{L})=\inf_{\mathrm{Tr}(A\sigma_\beta)=0,A\neq 0}\frac{-\left\langle A,\mathcal{L}^\dagger(A)\right\rangle_{1/2,\sigma_\beta}}{\left\langle A,A\right\rangle_{1/2,\sigma_\beta}}\,.
\]


We adopt the following asymptotic notations beside the usual big $\Or$ one. We write $f=\Omega(g)$ if $g=\Or(f)$; $f=\Theta(g)$ if $f=\Or(g)$ and $g=\Or(f)$. The notations $\wt{\Or}$, $\wt{\Omega}$, $\wt{\Theta}$ are used to suppress subdominant polylogarithmic factors. If not specified, $f = \wt{\Or}(g)$ if $f = \Or(g\operatorname{polylog}(g))$; $f = \wt{\Omega}(g)$ if $f = \Omega(g\operatorname{polylog}(g))$; $f = \wt{\Theta}(g)$ if $f = \Theta(g\operatorname{polylog}(g))$. Note that these tilde notations do not remove or suppress dominant polylogarithmic factors. For instance, if $f=\Or(\log g \log\log g)$, then we write $f=\wt{\Or}(\log g)$ instead of $f=\wt{\Or}(1)$.

In addition, we note that when analyzing the approximate fixed point of $\Phi$ in~\cref{sec:app_fix_thermal} and~\cref{sec:app_fix_gs}, we define the limiting generator of $\mc{L}$ as  $\widetilde{\mc{L}}$ after letting $T \rightarrow \infty$, and set $\widetilde{\Phi}= \mc{U}_S \circ \exp\left(\widetilde{\mc{L}} \alpha^2\right) \circ \mc{U}_S$. 
Furthermore, in the proofs of the mixing times in~\cref{sec:thermal_mixing} and~\cref{sec:gs_mixing}, we further approximate $\widetilde{\mc{L}}$ by $\widehat{\mc{L}}$, which exactly fixes the thermal state or ground state.

\rev{Given a quantum channel $\Phi$, the integer mixing time of $\Phi$ describes the minimum number of iterations required so that, starting from any initial state, the evolved state is guaranteed to be $\epsilon$-close to the target state. In this sense, it characterizes the worst-case convergence time over all initial states.
\begin{defn}\label{def:mixing_time} Given a CPTP map $\Phi$ with a unique fixed point $\rho_{\rm fix}(\Phi)$ and $\epsilon>0$, the integer mixing time $\tau_{{\rm mix},\Phi}(\epsilon)$ is defined as
\begin{equation}
\tau_{{\rm mix},\Phi}(\epsilon)=\min \left\{t\in \mathbb{N}\middle|\sup_{\rho}\|\Phi^{t}(\rho)-\rho_{\rm fix}(\Phi)\|_1\leq \epsilon\right\}\,.
\end{equation}
For $\Phi$ that takes the form of (2), the parameter $\alpha^2$ can be interpreted as the effective Lindbladian evolution time per application, and we define the (rescaled) mixing time as
\begin{equation}
t_{{\rm mix},\Phi}(\epsilon)=\alpha^2 \tau_{{\rm mix},\Phi}(\epsilon).
\end{equation}
\end{defn}}

Besides \cref{def:mixing_time}, other definitions of the mixing time are also used in the literature such as
\[
\tau_{\rm mix;c} = \min \left\{ t \in \mathbb{N} \,\middle|\, \sup_{\rho_1 \neq \rho_2} \frac{\| \Phi^{t}(\rho_1) - \Phi^{t}(\rho_2) \|_1}{\| \rho_1 - \rho_2 \|_1} \leq \frac{1}{2} \right\}.
\]
It is well known that $\tau_{\rm mix,\Phi}(\epsilon) \leq \tau_{\rm mix;c} \left(\log_2(1/\epsilon) + 1\right)$, indicating that $\tau_{\rm mix,\Phi}(\epsilon)$ scales logarithmically in $1/\epsilon$ whenever $\tau_{\rm mix;c} < \infty$~\cite{KastoryanoTemme2013}.

\section{Related works}\label{sec:related_works}

In this section, we review the related works on thermal and ground state preparation, focusing on the recent developments in Lindblad dynamics and weak-interaction dissipative systems.


Lindblad dynamics, originally developed to model the evolution of weakly coupled open quantum systems, has garnered significant attention in the past two years as a protocol for preparing thermal~\cite{RallWangWocjan2023,ChenKastoryanoBrandaoEtAl2023,ChenKastoryanoGilyen2023,Ding_2025} and ground states~\cite{Ding2024,zhan2025rapidquantumgroundstate}, due to its mathematical simplicity and analytical tractability. Given a Hamiltonian $H$, one can construct appropriate Lindblad operators (typically of the form $K = \int_{-\infty}^{\infty} f(s) e^{iHs} A e^{-iHs} \, \mathrm{d}s$) along with a suitable coherent term, such that the resulting dynamics drive any initial state toward the thermal or ground state. The convergence rate is governed by the mixing time of the dynamics. Recently, the mixing time analysis of Lindblad dynamics has been successfully carried out for various physically relevant Hamiltonians in both the thermal~\cite{TemmeKastoryanoRuskaiEtAl2010,KastoryanoTemme2013,BardetCapelGaoEtAl2023,rouz2024,DingLiLinZhang2024,kochanowski2024rapid,rouze2024optimal,tong2024fast} and ground state~\cite{zhan2025rapidquantumgroundstate} regimes. Leveraging well-developed Lindbladian simulation algorithms~\cite{cleve_2017,li_et_2023,Ding2024,ChenKastoryanoBrandaoEtAl2023,PRXQuantum.5.020332}, such dynamics can be efficiently simulated on a fault-tolerant quantum computer. However, due to the complexity of the jump operator, most simulation algorithms require a large number of ancilla qubits, controlled or time-reversed Hamiltonian evolutions, and intricate quantum control logic for clock registers, making them unsuitable for early fault-tolerant quantum devices. To mitigate the cost of simulating the detailed balanced Lindblad dynamics, very recently~\cite{hahn2025} proposes a variational compilation strategy to construct an approximation to the jump operator and to simulate the Lindblad dynamics using local gates.

In contrast to the Lindblad dynamics, the implementation of weak-interaction dissipative systems is more straightforward. Once the bath and system-bath interaction are specified, the dynamics can be simulated using forward Hamiltonian evolution followed by partial trace (or repeated interactions). Similar to our work, several concurrent works~\cite{molpeceres2025,hahn2025provably,hagan2025thermodynamic,langbehn2025universal,lloyd2025quantumthermal,scandi2025thermal} have also proposed quantum algorithms for thermal state preparation based on system-bath interaction models. While these works offer valuable insights, they do not provide rigorous end-to-end performance guarantees, and/or may face challenges in early fault-tolerant implementation.
In the following, we provide a brief overview of these works that are more relevant to ours and highlight the differences with our approach and summarize them in~\cref{tab:comparison}:

\begin{table*}[htbp!]
\begin{adjustbox}{width=\textwidth}
\centering
\begin{tabular}{c|c|c|c|c}
\hline
\hline
& \multicolumn{3}{c|}{\textbf{Properties}} &
\\
\cline{2-4}
\textbf{Algorithms} & Fixed-point    & Mixing time & Early fault-tolerant & \textbf{Remarks}\\
& error bound & guarantee& Implementation & \\ \hline
\begin{tabular}{ll}Lindblad dynamics based\\
thermal state preparation \cite{RallWangWocjan2023,ChenKastoryanoBrandaoEtAl2023,ChenKastoryanoGilyen2023}\end{tabular}& \cmark & \cmark  & \xmark & \begin{tabular}{ll}Controlled Hamiltonian\ simulation;\\ Complex logic gates\end{tabular}\\ \hline
\begin{tabular}{l}Discrete dynamics simulating\\
Metropolis-type sampling  \cite{Temme_2011,jiang2024quantum}\end{tabular}& \cmark & ? &\xmark & \begin{tabular}{ll}Controlled Hamiltonian simulation;\\ Complex logic gates\end{tabular}
\\
\hline
\begin{tabular}{ll}Lindblad dynamics based\\
ground state preparation \cite{Ding2024,Ding_2025}\end{tabular}& \cmark & \cmark  & ? & \begin{tabular}{ll}Time-reversed Hamiltonian \\simulation\end{tabular}\\
\hline
\begin{tabular}{ll}Hahn \textit{et al} \cite{hahn2025}\end{tabular}& ? & ?  & \cmark  & \begin{tabular}{ll}Variational compilation\end{tabular}\\
\hline
\multicolumn{1}{c}{\begin{tabular}{l}Weakly-coupled system\\
bath interaction\end{tabular}}
\\
\hline
Hagan \textit{et al} \cite{hagan2025thermodynamic}&\cmark &?&?& \begin{tabular}{ll}Haar-random system-bath coupling;\\
Exponential simulation time
\end{tabular}
\\
\hline
Hahn \textit{et al} \cite{hahn2025provably}& \cmark & ?  &\cmark & \begin{tabular}{ll} Only allow small
energy transitions
\end{tabular}
\\
\hline
Langbehn \textit{et al} \cite{langbehn2025universal}& ? & ?  &\cmark & Rotating wave approximation
\\
\hline
Lloyd \textit{et al} \cite{lloyd2025quantumthermal}& \cmark &  ? &\cmark &
\begin{tabular}{ll}Similar structure as~\cite{hahn2025provably}
\end{tabular}
 \\
\hline
Scandi \textit{et al} \cite{scandi2025thermal} & \cmark & ?& ? & \begin{tabular}{ll} Gaussian bath coupling\end{tabular}
\\
\hline
Shtanko \textit{et al} \cite{shtanko2021preparing}, Chen~\textit{et al} \cite{chen2023fastthermalizationeigenstatethermalization}& \cmark  & \cmark & ? & \begin{tabular}{ll}ETH hypothesis\end{tabular}
\\ \hline
\textbf{This work} & \textcolor{green}{\cmark}  & \cmark
 & \cmark & \begin{tabular}{ll}Large energy  transitions;\\Can prepare ground state \end{tabular}
\\ \hline\hline
\end{tabular}
\end{adjustbox}
\vspace{0.5em}
\caption{Comparison of recent quantum thermal and ground state preparation algorithms based on Lindblad dynamics or weakly coupled system-bath interaction.
``Fixed-point error bound'' refers to whether there is a rigorous fixed-point error bound for a general Hamiltonian $H$. ``Mixing time guarantee'' indicates whether the mixing time of the algorithm can be theoretically established at least for certain interacting Hamiltonians (see \cref{sec:mixing_time}). }
\label{tab:comparison}
\end{table*}


\begin{itemize}

\item In~\cite{molpeceres2025}, the authors study the weak-interaction algorithm in the regime of small $\alpha$ and constant $f(t)$, and rigorously establish its correctness and efficiency for a specific free fermion model. To the best of our knowledge, it remains unclear whether their approach extends to a general Hamiltonian $H$.

\item In~\cite{hagan2025thermodynamic}, the authors assume Haar-random system-bath coupling and establish a rigorous fixed-point error bound for the thermal state. According to their theoretical results, for general systems, the algorithm may require impractical parameter choices to resolve exponentially close eigenvalues. For instance, as discussed in~\cite[Section I.A]{hagan2025thermodynamic}, the required coupling strength $\alpha$ might be exponentially small, which in turn requires the simulation time $T$ in each step to scale exponentially with the number of qubits. Consequently, the total simulation time becomes exponentially long to guarantee the correctness of the fixed point.

  \item In~\cite{hahn2025provably}, the authors prove a result similar to~\cref{thm:fix_thermal} for the thermal state preparation. Although their work presents a result similar to ours in the thermal state setting, the authors do not provide theoretical guarantees on the mixing time—an essential component for establishing the end-to-end complexity of the algorithm (see the detailed discussion in~\cref{sec:toy_uniform} and~\cref{rem:filter_contrast}). In contrast, in~\cref{sec:mixing_time}, we prove that for commuting local Hamiltonians and free fermion systems, the mixing time admits a well-defined limit as $\sigma \to \infty$, thereby yielding a complete fixed-point error bound for these models, as stated in Corollary~\ref{cor:complete}.

\item The algorithmic structure in~\cite{lloyd2025quantumthermal} is similar to that in~\cite{hahn2025provably}. In both works, the bath state is initialized as $\ket{0}\bra{0}$, and the interaction function $f$ is carefully tuned so that the resulting jump operator in the approximate Lindblad dynamics satisfies the detailed balance condition. Ref.~\cite{lloyd2025quantumthermal} justifies the fixed-point error bound in the perturbative regime. Although the paper does not provide a fully rigorous error bound, its numerical results support both the efficiency of the algorithm and the validity of the perturbative analysis. We note that, unlike the two works~\cite{lloyd2025quantumthermal,hahn2025provably}, our approach employs a nontrivial initial bath state—specifically, the thermal bath state. This choice ensures that the dissipative part of our approximate Lindbladian dynamics automatically satisfies the detailed balance condition. Consequently, the interaction function $f$ in our framework can be designed with a flexibly tunable variance $\sigma$ (independent of $\beta$), without the need to impose additional constraints or formulation to maintain detailed balance. This differs from the interaction functions used in~\cite{lloyd2025quantumthermal} and~\cite{hahn2025provably}. Thanks to this flexibility, our algorithm can accommodate large energy transitions and achieve rigorous mixing times, all while maintaining a provable bound on the fixed-point error.


  \item \rev{In~\cite{scandi2025thermal}, the authors prove a result similar to Theorem 4, showing that the corresponding Lindbladian dynamics approximately satisfy the KMS detailed balance condition. This, in turn, implies~\cref{thm:fix_thermal} as a corollary. In contrast to our result, their analysis only considers the thermal state preparation and relies on the assumption of a Gaussian bath. Their algorithm also requires a detailed characterization of the two-point correlation functions.}

  \item \rev{In~\cite{langbehn2025universal,ramonescandell2025thermalstatepreparationrepeated,shtanko2021preparing,chen2023fastthermalizationeigenstatethermalization}, the authors investigate bath and system–bath interaction models similar to ours. However, the theoretical analyses in~\cite{langbehn2025universal,ramonescandell2025thermalstatepreparationrepeated} are primarily limited to small-scale systems, while~\cite{shtanko2021preparing,chen2023fastthermalizationeigenstatethermalization} rely on the Eigenstate Thermalization Hypothesis (ETH). In particular, under the ETH assumption,~\cite{chen2023fastthermalizationeigenstatethermalization} demonstrates that the repeated interaction dynamics can be effectively approximated by a Davies generator for thermal state preparation.}

  \item Our choice of $f$ is inspired by~\cite{ChenKastoryanoBrandaoEtAl2023}, where the authors construct a Lindbladian dynamics using the same filter function in the jump operators. Under this framework, they also establish a fixed-point error bound for the thermal state similar to~\cref{thm:fix_thermal}.

   \item $\Phi$ to Lindbladian dynamics: There is extensive literature supporting the convergence of $\Phi$ to Lindbladian dynamics under the weak-interaction assumption. Notably,~\cite{PhysRevA.88.012103,Mozgunov2020} derive the Coarse-Grained Master Equation (CGME) in the presence of a general bath. More recently,~\cite[Appendix D]{ChenKastoryanoBrandaoEtAl2023} rigorously shows that the resulting Lindbladian dynamics with $f(t) = \frac{1}{T} \mathbf{1}_{[-T/2, T/2]}(t)$ approximately fixes the thermal state, yielding a result similar to our~\cref{thm:fix_thermal}. In contrast to the general setting of~\cite{PhysRevA.88.012103,Mozgunov2020}, we provide a simple and explicit choice of bath and coupling operators that allows the Lindbladian dynamics to be derived more easily. Moreover, our use of a Gaussian filter $f(t)$ leads to a better fixed-point error bound compared to the flat choice of $f$ in~\cite[Appendix D]{ChenKastoryanoBrandaoEtAl2023}.

    \item In~\cite{Ding2024}, the authors proposed a Lindbladian-dynamics-based algorithm for ground state preparation. As demonstrated in~\cite{zhan2025rapidquantumgroundstate}, both theoretically and numerically, the dynamics exhibits rapid mixing for several physical Hamiltonians. We note that the algorithm in~\cite{Ding2024} simulates the Lindbladian dynamics using a single ancilla qubit but requires time-reversed Hamiltonian evolution. In contrast, our algorithm involves only forward Hamiltonian evolution, which leads to a nontrivial Lamb shift term in the dynamics that must be carefully handled in the convergence analysis.

\end{itemize}

\rev{After completing this work, two recent studies~\cite{slezak2026polynomialtimethermalizationgibbssampling,wang2026lindbladdynamicsrigorousguarantees} have further advanced the theoretical understanding of the system-bath interaction algorithm developed in this work. Ref.~\cite{slezak2026polynomialtimethermalizationgibbssampling} substantially extends our mixing-time analysis for thermal-state preparation to a broader class of Hamiltonians, including local spin Hamiltonians at high temperature, local weakly interacting fermionic Hamiltonians with sufficiently small interaction strength at constant inverse temperature $\beta$, and one-dimensional nearest-neighbor Hamiltonians at constant $\beta$. The authors also establish a monotonicity result for the spectral gap of the approximate Lindbladian dynamics. These ideas are then used in~\cite{wang2026lindbladdynamicsrigorousguarantees} to analyze thermal-state preparation even beyond the Lindblad regime.}

\section{Derivation of Effective Lindblad dynamics}\label{sec:phi_alpha_char}

Recall the time evolution operator by $U_S(t) := \exp(-iHt)$, and the associated superoperator by $\mc{U}_S(t)[\rho]=U_S(t)\rho U^\dagger_S(t)$. We then show that the quantum map $\Phi$ can be approximated by an effective Lindblad dynamics in the following theorem:

\begin{thm}[Rigorous version of Theorem 4]\label{thm:char_Phi_alpha_rigor} Under the choice of $H_E, A_S, B_E, f(t), g(\omega)$ in the main text, $\rho_{n+1}$ can be expressed as
\begin{equation}\label{eqn:general_formula}
\begin{aligned}
\rho_{n+1/3}=&U_S(T)\rho_{n}U^\dagger_S(T)=\mathcal{U}_S(T)[\rho_n]\\
\rho_{n+2/3}=&\rho_{n+1/3}\\
&+\alpha^2\underbrace{\mathbb{E}_{A_S,\omega}\left\{-i[\rrev{g(\omega)}H_{\mathrm{LS},A_S}(\omega),\rho_{n+1/3}]+\frac{\rrev{g(\omega)}}{1+\exp(\beta\omega)}\mathcal{D}_{V_{A^\dagger_S,f,T}(\omega)}(\rho_{n+1/3})+\frac{\rrev{g(\omega)}}{1+\exp(-\beta\omega)}\mathcal{D}_{V_{A_S,f,T}(-\omega)}(\rho_{n+1/3})\right\}}_{:=\mathcal{L}[\rho]}\\
&+\mc{O}(\alpha^4\|A_S\|^4T^4\|f\|^4_{L^\infty})\\
=&\exp(\mc{L}\alpha^2)\rho_{n+1/3}+\mc{O}(\alpha^4\|A_S\|^4T^4\|f\|^4_{L^\infty})\\
\rho_{n+1}=&U_S(T)\rho_{n+2/3}U^\dagger_S(T)=\mathcal{U}_S(T)[\rho_{n+2/3}]
\end{aligned}\,,
\end{equation}
Here,
\[
H_{\mathrm{LS},A_S}(\omega)=\rev{-}\mathrm{Im}\left(\frac{\exp(-\beta \omega)}{1+\exp(-\beta\omega)}\mc{G}_{A^\dagger_S,f}(\omega)+\frac{1}{1+\exp(-\beta\omega)}\mc{G}_{A_S,f}(-\omega)\right)\,,
\]
with
\begin{equation}\label{eqn:G_S}
\mc{G}_{A_S,f}(\omega)=\int^T_{-T}\int^{s_1}_{-T}f(s_2)f(s_1) A^\dagger_S(s_2)A_S(s_1)\exp(-i\omega(s_1-s_2))\mathrm{d}s_2\mathrm{d}s_1\,.
\end{equation}
\rrev{The jump operator $V_{A_S,f,T}(\omega)$ is defined as
\[
V_{A_S,f,T}(\omega)=\int^T_{-T}f(t)A_S(t)\exp(-i\omega t)\mathrm{d}t.
\]
where $A_S(t)$ is the Heisenberg-picture evolution of the coupling operator $A_S$.}
\end{thm}

We put the proof of the above theorem in \cref{sec:proof_rhoevolution}.~\rev{In our work, because we assume $A_S$ is uniformly sampled from $\mc{A}=\{A^i,-A^i\}_i$ with the property that $\{(A^i)^\dagger\}_i=\{A^i\}_i$ and $\omega$ is sampled from $g$, we obtain
\[
\begin{aligned}
&\mathcal{L}(\rho)=\mathbb{E}_{A_S}\left\{-i\int^\infty_{-\infty}[g(\omega)H_{\mathrm{LS},A_S}(\omega),\rho_{n+1/3}]\mathrm{d}\omega+\int^\infty_{-\infty}\frac{g(\omega)}{1+\exp(\beta\omega)}\mathcal{D}_{V_{A^\dagger_S,f,T}(\omega)}(\rho_{n+1/3})\mathrm{d}\omega\right.\\
&\left.+\int^\infty_{-\infty}\frac{g(\omega)}{1+\exp(-\beta\omega)}\mathcal{D}_{V_{A_S,f,T}(-\omega)}(\rho_{n+1/3})\mathrm{d}\omega\right\}\\
=&\mathbb{E}_{A_S}\left\{-i\int^\infty_{-\infty}[g(\omega)H_{\mathrm{LS},A_S}(\omega),\rho_{n+1/3}]\mathrm{d}\omega+\int^\infty_{-\infty}\frac{g(\omega)}{1+\exp(\beta\omega)}\mathcal{D}_{V_{A_S,f,T}(\omega)}(\rho_{n+1/3})\mathrm{d}\omega\right.\\
&\left.+\int^\infty_{-\infty}\frac{g(\omega)}{1+\exp(-\beta\omega)}\mathcal{D}_{V_{A_S,f,T}(-\omega)}(\rho_{n+1/3})\mathrm{d}\omega\right\}\\
=&\mathbb{E}_{A_S}\left\{-i\int^\infty_{-\infty}[g(\omega)H_{\mathrm{LS},A_S}(\omega),\rho_{n+1/3}]\mathrm{d}\omega+\int^\infty_{-\infty}\frac{g(\omega)+g(-\omega)}{1+\exp(\beta\omega)}\mathcal{D}_{V_{A_S,f,T}(\omega)}(\rho_{n+1/3})\mathrm{d}\omega\right\}\\
\end{aligned}\,,
\]
This gives (8) in Theorem 4.} According to the above theorem, another perspective on our algorithm is that it can be viewed as a simulation method that reproduces (8) using at most two forward evolutions with a single ancilla qubit and randomness. It is worth noting that related results on a given Lindbladian simulation (without forward evolution) have also been obtained in~\cite{chen2025,yu2024exponentiallyreducedcircuitdepths,PhysRevResearch.6.043321}. However, we emphasize that our main contribution lies in presenting a particularly simple choice of environment and bath, such that the resulting Lindbladian dynamics naturally generate a jump operator in integral form. This construction eliminates the need for block encoding or explicit decomposition of the jump operator.

\rev{In addition, we clarify the relation between our derivation of the effective Lindbladian and the standard first-principles derivation of a quantum Markovian master equation~\cite{davies_markovian_1974,Davies1976,BreuerPetruccione2002,agerskov2026markovianquantummasterequations,PhysRevB.102.115109}. The main difference is the timescale being analyzed. In each application of our channel, the Hamiltonian is evolved only over the finite interval $[-T,T]$, while the dissipative correction per application appears only at order $\alpha^2$. Accordingly, we use a short-time perturbative expansion of the reduced channel,
\[
\Phi(\rho)=\rho+\alpha^2\mathcal{L}(\rho)+\mathcal{O}(\alpha^4),
\]
where the $\mathcal{O}(\alpha^4)$ term is the higher-order remainder of the weak-coupling expansion. This argument does not rely on decay of bath correlation functions. By contrast, first-principles derivations aim at an effective master equation valid on long system--bath timescales, and the Markov approximation is justified by rapid decay of bath correlations~\cite{davies_markovian_1974,Davies1976,BreuerPetruccione2002}. A second difference is that in our setting the generator $\mathcal{L}$ already has Lindblad form before any rotating-wave or secular approximation is invoked. Those approximations enter later only as analytical tools for comparing our generator with detailed-balance generators such as the Davies generator, rather than as steps needed to make the dynamics CPTP.}

\subsection{Proof of \texorpdfstring{\cref{thm:char_Phi_alpha_rigor}}{Lg}}\label{sec:proof_rhoevolution}

In this section, we prove \cref{thm:char_Phi_alpha_rigor}.

\begin{proof}[Proof of~\cref{thm:char_Phi_alpha_rigor}]

Define $\rho_{\rm ini}=\rho_n \otimes \rho_E$, $\rho(T)=U^\alpha(T) \rho_{\rm ini} U^{\alpha}(T)^\dagger$ and $G(t)=f(t) \left( A_S \otimes B_E + A^\dagger_S \otimes B_E^\dagger \right)$. We first expand $U^{\alpha}(t):= \mathcal{T} \exp\left(-i \int_{-T}^t H_{\alpha}(s) \, \mathrm{d}s \right)$ into Dyson series:
\[
U^{\alpha}(t)=U_{0}(t;-T)-i\alpha U_{1}(t;-T)+(-i\alpha)^2U_{2}(t;-T)+(-i\alpha)^3U_{3}(t;-T)+\mathcal{O}\left(\alpha^4T^4\|f\|^4_{L^\infty}\left(\|A_S\|\|B_E\|\right)^4\right)\,.
\]
Here $U_{0}(t;-T)=\exp(-i(H+H_E)(t-(-T))$. Let $\mc{G}(t)=U^\dagger_{0}(t;-T)G(t)U_0(t;-T)$, which is the evolution of $G(t)$ under the Heisenberg picture. Then,
\[
U_n(t;-T)=U_{0}(t;-T)\int^t_{-T}\int^{s_1}_{-T}\dots \int^{s_{n-1}}_{-T}\mc{G}(s_1)\mc{G}(s_2)\dots \mc{G}(s_n)\mathrm{d}s_n\mathrm{d}s_{n-1}\dots\mathrm{d}s_1\,.
\]

According to the above expansion, it is straightforward to see that
\[
\begin{aligned}
\rho(T)=&U_0(T;-T)\rho_{\rm ini}U^\dagger_0(T;-T)-i\alpha \underbrace{\left(U_1(T;-T)\rho_{\rm ini}U^\dagger_0(T;-T)-U_0(T;-T)\rho_{\rm ini} U^\dagger_1(T;-T)\right)}_{\mathbb{E}(\cdot)=0}\\
&+\alpha^2\left(-U_0(T;-T)\rho_{\rm ini} U^\dagger_2(T;-T)-U_2(T;-T)\rho_{\rm ini}U^\dagger_0(T;-T)+U_1(T;-T)\rho_{\rm ini}U^\dagger_1(T;-T)\right)\\
&+\alpha^3\underbrace{\left(\cdots\right)}_{\mathbb{E}(\cdot)=0}+\mathcal{O}\left(\alpha^4T^4\|f\|^4_{L^\infty}\left(\|A_S\|\|B_E\|\right)^4\right)
\end{aligned}\,,
\]
Here, for the first order and third order term, we have expectation equals to zero because $\mathbb{E}(G(t))=0$.

Now, we only care about the second order term. Let $\widehat{\rho}(T)=U_0(T;-T)\rho_{\rm ini}U^\dagger_0(T;-T)$. Then,
\[
\begin{aligned}
&U_0(T;-T)\rho_{\rm ini} U^\dagger_2(T;-T)=\widehat{\rho}(T)U_0(T;-T)\int^T_{-T}\int^{s_1}_{-T}\mc{G}(s_2)\mc{G}(s_1)\mathrm{d}s_2\mathrm{d}s_1U^\dagger_0(T;-T)\\
=&\widehat{\rho}(T)U_0(T;-T)\frac{1}{2}\int^T_{-T}\int^{T}_{-T}\mc{G}(s_2)\mc{G}(s_1)\mathrm{d}s_2\mathrm{d}s_1U^\dagger_0(T;-T)+\widehat{\rho}(T)U_0(T;-T)\frac{1}{2}\int^T_{-T}\int^{s_1}_{-T}[\mc{G}(s_2),\mc{G}(s_1)]\mathrm{d}s_2\mathrm{d}s_1U^\dagger_0(T;-T)\,,
\end{aligned}
\]
where we use $\int^T_{-T}\int^{s_1}_{-T}\mc{G}(s_1)\mc{G}(s_2)\mathrm{d}s_2\mathrm{d}s_1=\int^T_{-T}\int^{T}_{s_1}\mc{G}(s_2)\mc{G}(s_1)\mathrm{d}s_2\mathrm{d}s_1$ in the last equality. Similarly,
\[
\begin{aligned}
&U_2(T;-T)\rho_{\rm ini} U^\dagger_0(T;-T)=U_0(T;-T)\int^T_{-T}\int^{s_1}_{-T}\mc{G}(s_1)\mc{G}(s_2)\mathrm{d}s_2\mathrm{d}s_1U^\dagger_0(T;-T)\widehat{\rho}(T)\\
=&U_0(T;-T)\frac{1}{2}\int^T_{-T}\int^{T}_{-T}\mc{G}(s_1)\mc{G}(s_2)\mathrm{d}s_2\mathrm{d}s_1U^\dagger_0(T;-T)\widehat{\rho}(T)+U_0(T;-T)\frac{1}{2}\int^T_{-T}\int^{s_1}_{-T}[\mc{G}(s_1),\mc{G}(s_2)]\mathrm{d}s_2\mathrm{d}s_1U^\dagger_0(T;-T)\widehat{\rho}(T)\,,
\end{aligned}
\]
and
\[
U_1(T;-T)\rho_{\rm ini}U^\dagger_1(T;-T)=\left(U_0(T;-T)\underbrace{\int^T_{-T}\mc{G}(s_1)\mathrm{d}s_1}_{:=V}U^\dagger_0(T;-T)\right)\widehat{\rho}(T)\left(U_0(T;-T)\int^T_{-T}\mc{G}(s_2)\mathrm{d}s_2U^\dagger_0(T;-T)\right)^\dagger\,.
\]
Combining the above three equalities and noticing $U^\dagger_0(T;-T)\widehat{\rho}(T)U_0(T;-T)=\rho_{\rm ini}$, this implies
\begin{equation}\label{eqn:evolution}
\begin{aligned}
&\rho(T)=U_0(T;-T)\rho_{\rm ini}U^\dagger_0(T;-T)\\
&+\alpha^2U_0(T;-T)\left(\underbrace{V\rho_{\rm ini} V^\dagger-\frac{1}{2}\left\{V^\dagger V,\rho_{\rm ini}\right\}}_{:=\text{Term I}}\underbrace{-i\left[\frac{i}{2}\int^T_{-T}\int^{s_1}_{-T}[\mc{G}(s_2),\mc{G}(s_1)]\mathrm{d}s_2\mathrm{d}s_1,\rho_{\rm ini}\right]}_{:=\text{Term II}}\right)U^\dagger_0(T;-T)+\mathcal{O}\left(\alpha^4T^4\|f\|^4_{L^\infty}\|A_S\|^4\right)\,.
\end{aligned}
\end{equation}
Here the expectation is taken over $A_S$ and $\omega$. We notice that $\rho_{n+1}=\mathbb{E}\left(\mathrm{Tr}_E\left(\rho(T)\right)\right)$. Let $\rho_{n+2/3}=U^\dagger_S(T)\rho_{n+1}U_S(T)$ and $\rho_{n+1/3}=U_S(T)\rho_{n}U^\dagger_S(T)$ as defined in~\cref{eqn:general_formula}. Applying $U^\dagger_0(0;-T)[\cdot]U_0(0;-T)$ on both sides of the above equality, tracing out the ancilla qubits, and taking the expectation over $A_S ,\omega$, we have
\begin{equation}\label{eqn:evolution_2}
\begin{aligned}
&\rho_{n+2/3}=\rho_{n+1/3}\\
&+\alpha^2\mathbb{E}\left(\mathrm{Tr}_E\left(U_0(0;-T)\left(\underbrace{V\rho_{\rm ini} V^\dagger-\frac{1}{2}\left\{V^\dagger V,\rho_{\rm ini}\right\}}_{:=\text{Term I}}\underbrace{-i\left[\frac{i}{2}\int^T_{-T}\int^{s_1}_{-T}[\mc{G}(s_2),\mc{G}(s_1)]\mathrm{d}s_2\mathrm{d}s_1,\rho_{\rm ini}\right]}_{:=\text{Term II}}\right)U^\dagger_0(0;-T)\right)\right)\\
+&\mathcal{O}\left(\alpha^4T^4\|f\|^4_{L^\infty}\|A_S\|^4\right)\,.
\end{aligned}
\end{equation}
Here, we note $U_0(0;-T)=\exp(-i(H+H_E)T)$.

Now, we deal with two terms separately:
\begin{itemize}
\item For the first term, we have
\[
\begin{aligned}
V=&\int^T_{-T}f(t)\exp(i\omega(t-(-T))\left(A_S(t;-T)\otimes \ket{1}\bra{0}\right)\mathrm{d}t\\
&+\int^T_{-T}f(t)\exp(-i\omega(t-(-T))\left(A^\dagger_S(t;-T)\otimes \ket{0}\bra{1}\right)\mathrm{d}t
\end{aligned}
\]
where
\[
A_S(t;-T)=\exp(iH(t+T))A_S\exp(-iH(t+T))\,,
\]

Let $A_{S,f}(\omega)=\int^T_{-T}f(t)A_S(t;-T)\exp(i\omega (t+T))\mathrm{d}t$. We have
\[
V=A_{S,f}(\omega)\otimes \ket{1}\bra{0}+A^\dagger_{S,f}(\omega)\otimes \ket{0}\bra{1}\,.
\]
This implies that
\[
\begin{aligned}
&\mathrm{Tr}_E\left(V\rho_{\rm ini} V^\dagger-\frac{1}{2}\left\{V^\dagger V,\rho_{\rm ini}\right\}\right)\\
=& \frac{\exp(-\beta\omega)}{1+\exp(-\beta\omega)}\left(A^\dagger_{S,f}(\omega)\rho_n A_{S,f}(\omega)-\frac{1}{2}\left\{A_{S,f}(\omega) A^\dagger_{S,f}(\omega),\rho_n\right\}\right)\\
&+\frac{1}{1+\exp(-\beta\omega)}\left(A_{S,f}(\omega)\rho_n A^\dagger_{S,f}(\omega)-\frac{1}{2}\left\{A^\dagger_{S,f}(\omega) A_{S,f}(\omega),\rho_n\right\}\right)
\end{aligned}\,.
\]
Recall $\rho_{n+1/3}=U_S(T)\rho_{n}U^\dagger_S(T)$ and $\rho_{\rm ini}=\rho_n\otimes \rho_E$. We can rewrite the above equality as
\[
\begin{aligned}
&\mathrm{Tr}_E\left(U_0(0;-T)\left(V\rho_{\rm ini} V^\dagger-\frac{1}{2}\left\{V^\dagger V,\rho_{\rm ini}\right\}\right)U^\dagger_0(0;-T)\right)\\
=&\mathrm{Tr}_E\left(\left(U_0(0;-T)VU_0(0;-T)^\dagger\right)U_0(0;-T)\rho_{\rm ini}U^\dagger_0(0;-T)\left(U_0(0;-T)V^\dagger U_0(0;-T)^\dagger\right)\right.\\
&\left.-\frac{1}{2}U_0(0;-T)\left\{V^\dagger V,\rho_{\rm ini}\right\}U_0(0;-T)^\dagger U^\dagger_0(0;-T)\right)\\
=& \frac{\exp(-\beta\omega)}{1+\exp(-\beta\omega)}\left(V^\dagger_{A_S,f}(\omega)\rho_{n+1/3} V_{A_S,f}(\omega)-\frac{1}{2}\left\{V_{A_S,f}(\omega) V^\dagger_{A_S,f}(\omega),\rho_{n+1/3}\right\}\right)\\
&+\frac{1}{1+\exp(-\beta\omega)}\left(V_{A_S,f}(\omega)\rho_{n+1/3} V^\dagger_{A_S,f}(\omega)-\frac{1}{2}\left\{V^\dagger_{A_S,f}(\omega) V_{A_S,f}(\omega),\rho_{n+1/3}\right\}\right)
\end{aligned}\,,
\]
Here, $V_{A_S,f}(\omega)=\int^T_{-T}f(t)A_S(t;0)\exp(i\omega t)\mathrm{d}t$.   This gives the Lindbladian operators in~\eqref{eqn:general_formula}.

\item For the second term: We first notice
\[
\begin{aligned}
&\int^T_{-T}\int^{s_1}_{-T}\mc{G}(s_2)\mc{G}(s_1)\mathrm{d}s_2\mathrm{d}s_1\\
=&\int^T_{-T}\int^{s_1}_{-T} (f(s_2)f(s_1)\exp(i\omega(s_2-s_1))\left(A_S(s_2;-T)A^\dagger_S(s_1;-T)\otimes\ket{1}\bra{1}\right)\mathrm{d}s_2\mathrm{d}s_1\\
&+\int^T_{-T}\int^{s_1}_{-T}(f(s_2)f(s_1)\exp(-i\omega(s_2-s_1))\left(A^\dagger_S(s_2;-T)A_S(s_1;-T)\otimes\ket{0}\bra{0}\right)\mathrm{d}s_2\mathrm{d}s_1
\end{aligned}
\]
We notice that
\[
A_S(s_2;-T)=\exp(iH(s_2+T))A_S\exp(-iH(s_2+T)),\quad A^\dagger_S(s_1;-T)=\exp(iH(s_1+T))A^\dagger_S\exp(-iH(s_1+T))\,.
\]
This implies
\[
\begin{aligned}
&A_S(s_2;-T)A^\dagger_S(s_1;-T)=\exp(iH(s_2+T))A_S\exp(-iH(s_2))\exp(iH(s_1))A^\dagger_S\exp(-iH(s_1+T))\\
=&\exp(iHT)A_S(s_2;0)A^\dagger_S(s_1;0)\exp(-iHT)\,.
\end{aligned}
\]
Recall $U_{0}(0;-T)=\exp(-i(H+H_E)T)$. Therefore, we have
\[
U_0(0;-T)\left(A_S(s_2;-T)A^\dagger_S(s_1;-T)\otimes\ket{1}\bra{1}\right)U^\dagger_0(0;-T)=A_S(s_2;0)A^\dagger_S(s_1;0)\otimes\ket{1}\bra{1}
\]
and
\[
U_0(0;-T)\left(A^\dagger_S(s_2;-T)A_S(s_1;-T)\otimes\ket{0}\bra{0}\right)U^\dagger_0(0;-T)=A^\dagger_S(s_2;0)A_S(s_1;0)\otimes\ket{0}\bra{0}
\]

Define
\[
\mc{G}_{A_S,f}(\omega)=\int^T_{-T}\int^{s_1}_{-T}f(s_2)f(s_1) A^\dagger_S(s_2;0)A_S(s_1;0)\exp(i\omega(s_2-s_1))\mathrm{d}s_2\mathrm{d}s_1\,.
\]
We then have
\[
\begin{aligned}
&U_0(0;-T)\int^T_{-T}\int^{s_1}_{-T}\mc{G}(s_2)\mc{G}(s_1)\mathrm{d}s_2\mathrm{d}s_1U^\dagger_0(0;-T)\\
=&\int^T_{-T}\int^{s_1}_{-T} (f(s_2)f(s_1)\exp(i\omega(s_2-s_1))\left(A_S(s_2;0)A^\dagger_S(s_1;0)\otimes\ket{1}\bra{1}\right)\mathrm{d}s_2\mathrm{d}s_1\\
&+\int^T_{-T}\int^{s_1}_{-T}(f(s_2)f(s_1)\exp(-i\omega(s_2-s_1))\left(A^\dagger_S(s_2;0)A_S(s_1;0)\otimes\ket{0}\bra{0}\right)\mathrm{d}s_2\mathrm{d}s_1\\
=&\mathcal{G}_{A^\dagger_S,f}(\omega)\otimes \ket{1}\bra{1}+\mathcal{G}_{A_S,f}(-\omega)\otimes \ket{0}\bra{0}\,.
\end{aligned}
\]
Because $\mc{G}(s)$ is a Hermitian matrix, we have
We then have
\[
\begin{aligned}
&U_0(0;-T)\int^T_{-T}\int^{s_1}_{-T}\mc{G}(s_1)\mc{G}(s_2)\mathrm{d}s_2\mathrm{d}s_1U^\dagger_0(0;-T)=\left(U_0(0;-T)\int^T_{-T}\int^{s_1}_{-T}\mc{G}(s_2)\mc{G}(s_1)\mathrm{d}s_2\mathrm{d}s_1U^\dagger_0(0;-T)\right)^\dagger\\
=&\mathcal{G}^\dagger_{A^\dagger_S,f}(\omega)\otimes \ket{1}\bra{1}+\mathcal{G}^\dagger_{A_S,f}(-\omega)\otimes \ket{0}\bra{0}\,.
\end{aligned}
\]

The above calculation gives
\[
\begin{aligned}
&\mathrm{Tr}_E\left(U_0(0;-T)\left(\text{Term II}\right)U^\dagger_0(0;-T)\right)\\
=&-i\left[\frac{i}{2}\frac{\exp(-\beta \omega)}{1+\exp(-\beta\omega)}\left(\mc{G}_{A^\dagger_S,f}(\omega)-\mc{G}^\dagger_{A^\dagger_S,f}(\omega)\right)+\frac{1}{1+\exp(-\beta\omega)}\left(\mc{G}_{A_S,f}(-\omega)-\mc{G}^\dagger_{A_S,f}(-\omega)\right),\rho_{n+1/3}\right]\\
=&-i\left[\frac{i}{2}\left(\frac{\exp(-\beta \omega)}{1+\exp(-\beta\omega)}\mc{G}_{A^\dagger_S,f}(\omega)+\frac{1}{1+\exp(-\beta\omega)}\mc{G}_{A_S,f}(-\omega)-\left(\dots\right)^\dagger\right),\rho_{n+1/3}\right]\\
=&-i\left[-\mathrm{Im}\left(\frac{\exp(-\beta \omega)}{1+\exp(-\beta\omega)}\mc{G}_{A^\dagger_S,f}(\omega)+\frac{1}{1+\exp(-\beta\omega)}\mc{G}_{A_S,f}(-\omega)
\right),\rho_{n+1/3}\right]
\end{aligned}\,.
\]
This gives the formula of $H_{\mathrm{LS},A_S}$ in the theorem.
\end{itemize}
\end{proof}

\section{Approximate CPTP map has close fixed point and mixing time}\label{sec:almost_fixed_point}

In this section, we show that the closeness of two CPTP maps $\Phi_1$ and $\Phi_2$ implies the closeness of their fixed points and mixing times. This provides a crucial link between the fixed point and mixing time of the Lindbladian dynamics in Theorem 4 and those of $\Phi$. The result is summarized in the following:

\begin{thm}\label{thm:almost_fixed_point}
Given two CPTP maps $\Phi_1,\Phi_2$ with unique fixed points $\rho_1,\rho_2$. Let $\tau_{1,\rm mix}(\epsilon),\tau_{2,\rm mix}(\epsilon)$ be the mixing time of $\Phi_1,\Phi_2$ respectively, defined as~\cref{def:mixing_time}. Then
\begin{itemize}
\item $\rho_1,\rho_2$ are close if the maps themselves are close: For any $\epsilon>0$,
\begin{equation}\label{eqn:fixed_point_error_1}
\left\|\rho_1-\rho_2\right\|_1\leq \epsilon+\tau_{1,\rm mix}(\epsilon)\|\Phi_1-\Phi_2\|_{1\leftrightarrow 1}\,.
\end{equation}
\item $\rho_1,\rho_2$ are close if $\Phi_1(\rho_2)$ is close to $\rho_2$: For any $\epsilon>0$,
\begin{equation}\label{eqn:fixed_point_error_2}
\left\|\rho_1-\rho_2\right\|_1\leq \epsilon+\tau_{1,\rm mix}(\epsilon)\|\Phi_1(\rho_2)-\rho_2\|_{1}\,.
\end{equation}

\item $\Phi_2$ has comparable mixing time with $\Phi_1$ if $\Phi_2$ is close to $\Phi_1$: Given any $\epsilon>0$, if $\tau_{1,\rm mix}(\epsilon/2)\|\Phi_1-\Phi_2\|_{1\leftrightarrow 1}\leq \epsilon/2$, then
\begin{equation}\label{eqn:mixing_time_bound}
\tau_{2,\rm mix}(2\epsilon)\leq \tau_{1,\rm mix}(\epsilon/2).
\end{equation}
\end{itemize}
\end{thm}
\noindent Importantly, \cref{eqn:fixed_point_error_2} makes no reference to the map $\Phi_2$, and applies for an arbitrary state $\rho_2$.~\rev{Perturbation bounds for quantum channels and their fixed points have been studied previously in the literature, e.g., in~\cite{Szehr_2013}. However, in~\cref{thm:almost_fixed_point}, we rely only on the mixing time of the quantum channel, which is a weaker assumption than the standard contraction conditions typically used in the literature, such as~\cite[Theorem~4]{Szehr_2013}. For completeness, we provide a full proof of the theorem below.}

\begin{proof}[Proof of~\cref{thm:almost_fixed_point}] To prove~\cref{eqn:fixed_point_error_1}, we notice that
\[
\begin{aligned}
&\left\|\rho_1-\rho_2\right\|_1=\left\|\rho_1-\Phi^{\tau_{1,\rm mix}(\epsilon)}_2(\rho_2)\right\|_1\\
\leq &\left\|\rho_1-\Phi^{\tau_{1,\rm mix}(\epsilon)}_1(\rho_2)\right\|_1+\left\|\Phi^{\tau_{1,\rm mix}(\epsilon)}_1(\rho_2)-\Phi^{\tau_{1,\rm mix}(\epsilon)}_2(\rho_2)\right\|_1\leq \epsilon+\tau_{1,\rm mix}(\epsilon)\|\Phi_1-\Phi_2\|_{1\leftrightarrow 1}
\end{aligned}
\]

\[
\begin{aligned}
&\left\|\rho_1-\rho_2\right\|_1\leq \left\|\rho_1-\Phi^{\tau_{1,\rm mix}(\epsilon)}_1(\rho_2)\right\|_1+\left\|\Phi^{\tau_{1,\rm mix}(\epsilon)}_1(\rho_2)-\rho_2\right\|_1\\
\leq &\left\|\rho_1-\Phi^{\tau_{1,\rm mix}(\epsilon)}_1(\rho_2)\right\|_1+\sum^{\tau_{1,\rm mix}(\epsilon)-1}_{n=0}\|\Phi^{n+1}_1(\rho_2)-\Phi^{n}_1(\rho_2)\|_1\leq \epsilon+\tau_{1,\rm mix}(\epsilon)\|\Phi_1(\rho_2)-\rho_2\|_{1}\,,
\end{aligned}
\]
where we use $\|\Phi_1\|_{1\leftrightarrow1}\leq 1$ in the last inequality.

Finally, to show the comparable mixing time, we note that for any $\rho$,
\[
\begin{aligned}
&\|\Phi^{\tau_{1,\rm mix}(\epsilon/2)}_2(\rho)-\rho_2\|_1\leq \|\Phi^{\tau_{1,\rm mix}(\epsilon/2)}_2(\rho)-\Phi^{\tau_{1,\rm mix}(\epsilon/2)}_1(\rho)\|_1+\|\Phi^{\tau_{1,\rm mix}(\epsilon/2)}_1(\rho)-\rho_1\|_1+\|\rho_1-\rho_2\|_1\\
\leq &\|\Phi^{\tau_{1,\rm mix}(\epsilon/2)}_2(\rho)-\Phi^{\tau_{1,\rm mix}(\epsilon/2)}_1(\rho)\|_1+\|\Phi^{\tau_{1,\rm mix}(\epsilon/2)}_1(\rho)-\rho_1\|_1+\tau_{1,\rm mix}(\epsilon/2)\|\Phi_1-\Phi_2\|_{1\leftrightarrow 1}+\epsilon/2\\
\leq &2\tau_{1,\rm mix}(\epsilon/2)\|\Phi_1-\Phi_2\|_{1\leftrightarrow 1}+\epsilon\leq 2\epsilon
\end{aligned}
\]
where we use~\eqref{eqn:fixed_point_error_1} in the second equality. This concludes the proof.

\end{proof}

\section{Fixed point error bounds for thermal and ground state preparation}\label{sec:thermal_ground_prep}

Under of $H_E, A_S, B_E, f(t), g(\omega)$ in the main text, the quantum channel $\Phi$ defined in~(2) can be engineered to approximately preserve the thermal or ground state of the system Hamiltonian.~\rev{The integer mixing time of $\Phi$ is defined in~\cref{def:mixing_time}. According to Theorem 4, the mixing time of $\Phi$ should be governed by the underlying Lindbladian operator $\mc{L}$.} The quantity $t_{{\rm mix},\Phi}(\epsilon)$ approximately captures the total Lindbladian evolution time required for mixing. When $\alpha$ is sufficiently small, this mixing time does not diverge as $\alpha \to 0$, but instead remains bounded above by a finite constant that depends only on properties of the Lindbladian.

The following theorem shows that
by properly choosing the parameters $\sigma,T,\alpha$ related to the mixing time, the fixed point of $\Phi$ is approximately the thermal state. We also omit some dependence on $\|H\|$ and $\|A_S\|$ for simplicity. The general version of~\cref{thm:fix_thermal} is stated in~\cref{sec:app_fix_thermal} as~\cref{thm:fix_thermal_rigor}, followed by the proof of both theorems.

\begin{thm}[Thermal state, informal]\label{thm:fix_thermal} Assume $0\leq \beta<\infty$ and $g(\omega)=\frac{1}{\omega_{\max}}\mathbf{1}_{[0,\omega_{\max}]}$. Then, for any $\epsilon>0$, if
\[
\sigma=\widetilde{\Omega}\left(\beta\omega^{-1}_{\max} \epsilon^{-1}t_{{\rm mix},\Phi}(\epsilon)\right),\quad T=\Omega(\sigma\log(\sigma/\epsilon))\,,
\]
and $\alpha=\mc{O}\left(\sigma T^{-2}\epsilon^{1/2}t^{-1/2}_{{\rm mix},\Phi}(\epsilon)\right)$,
then
\[\|\rho_{\rm fix}(\Phi)-\rho_\beta\|_1<\epsilon\,.
\]
\end{thm}
\noindent \cref{thm:fix_thermal} shows that if we set $\sigma = \widetilde{\Theta}\left(\omega_{\max}^{-1} \beta \epsilon^{-1} t_{{\rm mix},\Phi}(\epsilon)\right)$, we ensure that the fixed point is $\epsilon$-close to the thermal state by choosing
\[
T = \widetilde{\Theta}\left(\omega_{\max}^{-1} \beta \epsilon^{-1} t_{{\rm mix},\Phi}(\epsilon)\right),\
\alpha = \widetilde{\Theta}\left(\omega_{\max} \beta^{-1} t_{{\rm mix},\Phi}^{-3/2} \epsilon^{3/2}\right).
\]

Analogously, we can establish a corresponding result for the ground state as follows.

\begin{thm}[Ground state, informal]\label{thm:fix_ground}  Assume $H$ has a spectral gap $\Delta$ and let $\ket{\psi_0}$ be the ground state of $H$. Then, for any $\epsilon>0$, if
\[
\sigma=\widetilde{\Omega}\left(\Delta^{-1}\log(\|H\|/\epsilon)\right),\quad T=\Omega(\sigma\log(\sigma/\epsilon))\,,
\]
and $\alpha=\mc{O}\left(\sigma T^{-2}\epsilon^{1/2}t^{-1/2}_{{\rm mix},\Phi}(\epsilon)\right)$,
then
\[\|\rho_{\rm fix}(\Phi)-\ket{\psi_0}\bra{\psi_0}\|_1<\epsilon.
\]
\end{thm}
\noindent
\cref{thm:fix_ground} shows that if we set $\sigma=\widetilde{\Theta}\left(\Delta^{-1}\right)$, it suffices to choose
\[
T=\widetilde{\Theta}(\Delta^{-1}),\ \alpha=\widetilde{\Theta}\left(\Delta \epsilon^{1/2} t^{-1/2}_{{\rm mix},\Phi}(\epsilon)\right).
\]
The rigorous version of~\cref{thm:fix_ground} is given in Supp~\ref{sec:app_fix_gs} as~\cref{thm:fix_ground_rigor}. \cref{thm:fix_ground} shows that if we set $\sigma=\widetilde{\Theta}\left(\Delta^{-1}\right)$, it suffices to choose
\[
T=\widetilde{\Theta}(\Delta^{-1}),\ \alpha=\widetilde{\Theta}\left(\Delta \epsilon^{1/2} t^{-1/2}_{{\rm mix},\Phi}(\epsilon)\right).
\]

The result of~\cref{thm:fix_thermal} applies to all values of $\beta$ and does not require $\Delta = \mathrm{poly}(N^{-1})$ for efficient state preparation, whereas~\cref{thm:fix_ground} does rely on this assumption to ensure efficient preparation. On the other hand, the dependence on $\beta$ in~\cref{thm:fix_thermal} may not be sharp, particularly in the large-$\beta$ regime. For instance, at very low temperatures, where $\beta =\Omega(\poly(N,1/\Delta, 1/\epsilon))$, preparing the $\epsilon$-approximate thermal state effectively reduces to preparing the ground state. In such cases, one may directly adopt the parameter choices in~\cref{thm:fix_ground} rather than those in~\cref{thm:fix_thermal}.

\rev{According to the approximation-error bounds in~\cref{thm:fix_thermal} and~\cref{thm:fix_ground}, once the effective mixing time $t_{\rm mix,\Phi_{\alpha}}$ is upper bounded, an appropriate choice of parameters guarantees that the fixed point $\rho_{\text{fix}}(\Phi)$ can be made arbitrarily close to the target state. However, as discussed in the main text, the main difficulty is that $t_{\rm mix,\Phi_{\alpha}}$ itself depends on the parameters $\sigma$ and $\alpha$ used in the construction of $\Phi$. Consequently, it may happen that as $\sigma$ or $\alpha^{-1}$ tends to $+\infty$, the mixing time $t_{\rm mix,\Phi_{\alpha}}$ also diverges, causing the conditions in~\cref{thm:fix_ground} and~\cref{thm:fix_thermal} to become unsatisfiable (see~\cref{sec:toy_uniform} and~\cref{rem:filter_contrast}). To circumvent this issue, we carefully design the dissipative protocol that allows large energy transition between eigenvectors, which further ensures that, once $\sigma$ is sufficiently large, the mixing time $t_{\rm mix,\Phi_{\alpha}}$ becomes \emph{independent} of~$\sigma$. In~\cref{sec:toy_uniform}–\cref{sec:thermal_mixing_general}, we rigorously prove that for certain classes of physical models such as  free-fermion systems, and local commuting Hamiltonians, the mixing time does not blow up with~$\sigma$ and can be upper bounded by a quantity that scales polynomially with the number of qubits.
}

To prove \cref{thm:fix_thermal}, according to \cref{thm:almost_fixed_point} in \cref{sec:almost_fixed_point}, it suffices to bound $\|\Phi(\rho_\beta)-\rho_\beta\|_1$ . This consists of two main steps:
\begin{enumerate}
    \item Approximate the map $\Phi$ by choosing $\alpha \ll 1$.
    \item Show that the limiting map approximately fixes the thermal or ground state when $\sigma,T\gg 1$.
\end{enumerate}
In the first step, using the result of Theorem 4, we have
\[
\|\Phi(\rho_\beta)-\rho_\beta\|_1 \approx\left\|\alpha^2\mc{L}(\rho_\beta)\right\|_1,\quad \alpha\ll 1
\]
with the approximation error quantified in (7). Thus, it suffices to show the Lindblad dynamics approximately fix the thermal/ground state.  This constitutes the most technical part of the proof. For thermal states, it has been shown that the dissipative part of the Lindbladian $\mc{L}$ in (8) is approximately detailed-balanced \cite{ChenKastoryanoBrandaoEtAl2023} when $\{(A^i)^\dagger\}_i = \{A^i\}_i$, and therefore approximately fixes the thermal state (see~\cref{sec:app_fix_thermal} Lemma~\ref{lem:dissipative_thermal}). When $\sigma\gg 1$, we show that the Lamb shift Hamiltonian $H_{\mathrm{LS},A_S}(\omega)$ approximately commutes with the thermal state (see~\cref{sec:app_fix_thermal} Lemma~\ref{lem:Lamb_shift_commute}). These two properties together imply that $\left\|\mc{L}(\rho_\beta)\right\|_1\approx 0$.

Note that, in order to ensure a small error $\epsilon$, \cref{thm:fix_thermal} requires that the parameters defining our algorithm satisfy conditions that depend on the mixing time $t_{{\rm mix},\Phi}$. The mixing time enters the proof because the relationship between $\|\rho_{\rm fix}(\Phi)-\rho_\beta\|_1$ and $\|\Phi(\rho_\beta)-\rho_\beta\|_1$ involves the mixing time, as shown in~\cref{sec:almost_fixed_point}.

The proof of~\cref{thm:fix_ground} is similar; however, under the spectral gap assumption, the ground state case allows a direct upper bound on $\|\mathcal{L}(\rho_\beta)\|_1$, and the fixed-point error bound is independent of the choice of $g$. Specifically, the $\gamma$-dependent term in $\mathcal{L}(\rho_\beta)$ takes the form $\int \gamma(\omega) \mathcal{E}(\omega) \, \ud \omega$
for some error operator $\mathcal{E}(\omega)$, which by normalization of $\gamma$ satisfies $\int \gamma(\omega) \norm{\mathcal{E}(\omega)}_1 \, \ud \omega \le \sup_{\omega\in\mathrm{supp}(\gamma)} \norm{\mathcal{E}(\omega)}_1$.
This last term can be bounded directly (see~\cref{thm:fix_ground_rigor}), allowing $\gamma$ (and $g$) to be optimized to reduce $t_{\rm mix}$. In contrast, for thermal state preparation, the Lamb shift term cannot be uniformly bounded for all $\omega$; instead, one must estimate the integral itself to show that it approximately commutes with the thermal state
 (see Lemma~\ref{lem:Lamb_shift_commute}).

\subsection{Approximate fixed point -- Thermal state}\label{sec:app_fix_thermal}
 In this section, we provide a rigorous version of~\cref{thm:fix_thermal}
 in~\cref{thm:fix_thermal_rigor} and provide the proof. We consider~(2) with $f(t)= \frac{1}{(2\pi)^{1/4} \sigma^{1/2}} \exp\left(-\frac{t^2}{4\sigma^2}\right)$. First, we can rewrite $\mc{L}$ in~(8) as
\begin{equation}\label{eqn:lindbladian_operator_refine}
\mc{L}(\rho)=\mathbb{E}_{A_S}\left(\int^\infty_{-\infty} -i\left[g(\omega)H_{\mathrm{LS},A_S}(\omega),\rho\right]+\gamma(\omega)\mathcal{D}_{V_{A_S,f}(\omega)}(\rho)\mathrm{d}\omega\right)\,,
\end{equation}
where $\gamma(\omega)=(g(\omega)+g(-\omega))/(1+\exp(\beta\omega))$. In the case when $\beta=\infty$, $\gamma(\omega)=(g(\omega)+g(-\omega))\textbf{1}_{\omega<0}+g(0)\textbf{1}_{\omega=0}$.

Before presenting the rigorous version of~\cref{thm:fix_thermal}, we first consider a simplified CPTP map $\widetilde{\Phi}$ defined as follows:
\begin{equation}\label{eqn:Phi_tilde_thermal}
\widetilde{\Phi}=\mc{U}_S(T)\,\circ\,\exp\left(\widetilde{\mathcal{L}}\alpha^2\right)\,\circ\,\mc{U}_S(T)\,.
\end{equation}
Compared to $\Phi$ in~(2), we omit the error terms in Theorem 4 and take the limit $T \to \infty$ in $\mc{L}$. Specifically, as mentioned in~\cref{sec:notation},
\begin{equation}\label{eqn:L_tilde_thermal}
\begin{aligned}
\widetilde{\mc{L}}(\rho)=-i\left[\widetilde{H}_{\mathrm{LS}},\rho\right]+\mathbb{E}_{A_S}\left(\int^\infty_{-\infty}\gamma(\omega)\mathcal{D}_{\widetilde{V}_{A_S,f}(\omega)}(\rho)\mathrm{d}\omega\right)\,,
\end{aligned}
\end{equation}
where
\[
\widetilde{H}_{\mathrm{LS}}=\rev{-\mathbb{E}_{A_S}\left(\mathrm{Im}\left(\int^\infty_{-\infty}\gamma(\omega)\widetilde{\mc{G}}_{A_S,f}(-\omega)\mathrm{d}\omega\right)\right)},\quad \widetilde{V}_{A_S,f}(\omega)=\int^\infty_{-\infty}f(t)A_S(t)\exp(-i\omega t)\mathrm{d}t\,,
\]
with
\begin{equation}\label{eqn:G_S_thermal}
\widetilde{\mc{G}}_{A,f}(\omega)=\int^\infty_{-\infty}\int^{s_1}_{-\infty}f(s_2)f(s_1) A^\dagger(s_2)A(s_1)\exp(-i\omega(s_1-s_2))\mathrm{d}s_2\mathrm{d}s_1\,.
\end{equation}
In the formula of $\widetilde{H}_{\mathrm{LS}}$, we use the fact that $\{(A^i)^\dagger\}_i=\{A^i\}_i$ and $\widetilde{\mc{G}}_{A^i,f}=\widetilde{\mc{G}}_{-A^i,f}$.

 The distance between $\Phi$ and $\widetilde{\Phi}$ can be controlled in the following lemma:
\begin{lem}\label{lem:thermal_approx_first} When $T>\sigma$, we have
\[
\left\|\Phi-\widetilde{\Phi}\right\|_{1\leftrightarrow1}=\mathcal{O}\left(\alpha^2\sigma\exp\left(-T^2/(4\sigma^2)\right)\mathbb{E}(\|A_S\|^2)+\alpha^4T^4\sigma^{-2}\mathbb{E}\left(\|A_S\|^4\right)\right)
\]
\end{lem}
\begin{proof}[Proof of Lemma~\ref{lem:thermal_approx_first}] According to~\cref{thm:char_Phi_alpha_rigor} and $\|\gamma(\omega)\|_{L^1}=1$, we have
\[
\begin{aligned}
&\left\|\Phi-\widetilde{\Phi}\right\|_{1\leftrightarrow1}\leq \alpha^2\|\mc{L}-\widetilde{\mc{L}}\|_{1\leftrightarrow1}+\mc{O}\left(\alpha^4T^4\sigma^{-2}\mathbb{E}\left(\|A_S\|^4\right)\right)\\
=&\mc{O}\left(\alpha^2\sup_{\omega}\left(\|\mc{G}_{A_S,f}(\omega)-\widetilde{\mc{G}}_{A_S,f}(\omega)\|+\|V_{A_S,f}(\omega)-\widetilde{V}_{A_S,f}(\omega)\|\underbrace{\|V_{A_S,f}(\omega)\|}_{=\mc{O}(\sigma^{1/2}\|A_S\|)}\right)\right)\\
&+\mc{O}\left(\alpha^4T^4\sigma^{-2}\mathbb{E}\left(\|A_S\|^4\right)\right)
\end{aligned}\,.
\]
Thus, it suffices to consider $\|V_{A_S,f}(\omega)-\widetilde{V}_{A_S,f}(\omega)\|$ and $\|\mc{G}_{A_S,f}(\omega)-\widetilde{\mc{G}}_{A_S,f}(\omega)\|$. For the first term, we have
\[
\begin{aligned}
&\|V_{A_S,f}(\omega)-\widetilde{V}_{A_S,f}(\omega)\|\leq \|A_S\|\int_{|t|>T}f(t)\mathrm{d}t=\mc{O}\left((\sigma^{3/2}/T)\exp(-T^2/(4\sigma^2))\|A_S\|\right)\\
=&\mc{O}\left(\sigma^{1/2}\exp(-T^2/(4\sigma^2))\|A_S\|\right)\,,
\end{aligned}
\]
where we use $T>\sigma$ in the second equality.
For the second term, we have
\[
\|\mc{G}_{A_S,f}(\omega)-\widetilde{\mc{G}}_{A_S,f}(\omega)\|\leq \|A_S\|^2\left(\int_{\abs{s_1}\ge T}\int^{s_1}_{-\infty}+\int^{T}_{-T}\int^{-T}_{-\infty}f(s_2)f(s_1)\mathrm{d}s_2\mathrm{d}s_1\right)=\mc{O}\left(\sigma\exp(-T^2/(4\sigma^2))\|A_S\|^2\right)\,.
\]
Combining these two bounds, we conclude the proof.
\end{proof}

Using $\widetilde{\Phi}$, we are ready to state the rigorous version of~\cref{thm:fix_thermal} and provide the proof:
\begin{thm}\label{thm:fix_thermal_rigor} Define
\begin{equation}\label{eqn:condition_R_thermal}
  R:=\int^\infty_{0}\left|\int^\infty_{-\infty}\gamma(\omega)\exp(i\omega \sigma q)\mathrm{d}\omega\right|\exp(-q^2/8)\mathrm{d}q\,.
\end{equation}
When $T>\sigma>\beta$, we have
\[
\begin{aligned}
&\|\rho_{\rm fix}(\Phi)-\rho_\beta\|_1
\\
\leq &\left(\mathbb{E}_{A_S}\left(\left\| \left[\rho_\beta, \int \gamma(\omega)\widetilde{\mc{G}}_{A_S,f}(-\omega)\,\mathrm{d}\omega \right] \right\|_1+\left\|\left[\rho_\beta, \int \gamma(\omega)\left(\widetilde{\mc{G}}_{A_S,f}(-\omega)\right)^\dagger\,\mathrm{d}\omega \right] \right\|_1+\left\|\int^\infty_{-\infty}\gamma(\omega)\mathcal{D}_{\widetilde{V}_{A_S,f}(\omega)}(\rho_\beta)\right\|_1\mathrm{d}\omega\right)\right)\\
=&\widetilde{\mathcal{O}}\left(\left(\left(R+\|\gamma(\omega)\|_\infty\frac{1}{\sigma}\sqrt{\log(\sigma/\beta)}\right)\beta\mathbb{E}\left(\|A_S\|^2\right)+\sigma\exp\left(-T^2/(4\sigma^2)\right)\mathbb{E}(\|A_S\|^2)+\alpha^2T^4\sigma^{-2}\mathbb{E}\left(\|A_S\|^4\right)\right)\alpha^2\tau_{{\rm mix},\Phi}(\epsilon)+\epsilon\right)
\end{aligned}
\]
\end{thm}
According to~\cref{thm:fix_thermal_rigor}, to ensure a small fixed-point error, we require $R$ to vanish as $\sigma \to \infty$. This, in turn, imposes a constraint on the choice of $\gamma(\omega)$ (and hence $g(\omega)$). We prove below that it suffices to choose $g$ to be a uniform distribution. We emphasize that this constraint arises from the need to control the fixed-point error associated with the Lamb shift term in Lemma~\ref{lem:Lamb_shift_commute}. Specifically, we cannot directly prove that each term in the $\omega$-expansion of $\widetilde{H}_{\rm LS}$ commutes with the thermal state. Instead, we prove that the entire term approximately commutes with the thermal state after integrating over $\omega$.

Before proving~\cref{thm:fix_thermal_rigor}, we first use it to prove~\cref{thm:fix_thermal}.
\begin{proof}[Proof of~\cref{thm:fix_thermal}]
When $g(\omega)=\frac{1}{\omega_{\max}}\textbf{1}_{\omega\in[0,\omega_{\max}]}$ with $\omega_{\max}=\Omega(1)$. In this case, we have $\gamma(\omega)=\frac{1}{\omega_{\max}(1+\exp(\beta \omega))}\textbf{1}_{\omega\in[-\omega_{\max},\omega_{\max}]}$. Thus, $\|\gamma\|_{\infty}=\frac{1}{\omega_{\max}}$ and
\[
R=\underbrace{\int^{(\sigma \omega_{\max})^{-1}}_{0}\left|\int^\infty_{-\infty}\gamma(\omega)\exp(i\omega \sigma q)\mathrm{d}\omega\right|\exp(-q^2/8)\mathrm{d}q}_{=\Or((\sigma \omega_{\max})^{-1})}+\int^{\infty}_{(\sigma \omega_{\max})^{-1}}\left|\int^\infty_{-\infty}\gamma(\omega)\exp(i\omega \sigma q)\mathrm{d}\omega\right|\exp(-q^2/8)\mathrm{d}q
\]
For the second term, we have
\[
\begin{aligned}
&\left|\int^\infty_{-\infty}\gamma(\omega)\exp(i\omega \sigma q)\mathrm{d}\omega\right|=\left|\frac{1}{i\sigma \omega_{\max} q}\int^{\omega_{\max}}_{-\omega_{\max}}\frac{1}{1+\exp(\beta\omega)}\mathrm{d}\left(\exp(i\omega \sigma q)\right)\right|\\
\leq &\frac{2}{\omega_{\max}\sigma q}+\frac{1}{\omega_{\max}\sigma q}\left|\int^{\omega_{\max}}_{-\omega_{\max}}\frac{\beta\exp(\beta \omega)}{(1+\exp(\beta\omega))^2}\exp(i\omega\sigma q)\mathrm{d}\omega\right|=\mc{O}\left(\frac{1}{\omega_{\max}\sigma q}\right)
\end{aligned}
\]
Here, we note $\left|\int^{\omega_{\max}}_{-\omega_{\max}}\frac{\beta\exp(\beta \omega)}{(1+\exp(\beta\omega))^2}\exp(i\omega\sigma q)\mathrm{d}\omega\right|\leq \left|\int^{\infty}_{-\infty}\frac{\exp(u)}{(1+\exp(u))^2}\mathrm{d}u\right|=\mc{O}(1)$. Plugging this back into the expression for $R$, we obtain
\[
R=\mc{O}\left(\frac{1}{\sigma\omega_{\max}}\log(\sigma \omega_{\max})\right)\,.
\]
Combining this,~\cref{thm:fix_thermal_rigor}, and $\|A_S\|\leq 1$, we have
\[
\begin{aligned}
&\|\rho_{\rm fix}(\Phi)-\rho_\beta\|_1\\
=&\widetilde{\mathcal{O}}\left(\left(\frac{\beta}{\omega_{\max}\sigma}\left(\sqrt{\log(\sigma/\beta)}+\log(\sigma \omega_{\max})\right)\mathbb{E}(\|A_S\|^2)+\sigma\exp\left(-T^2/(4\sigma^2)\right)\mathbb{E}(\|A_S\|^2)+\alpha^2T^4\sigma^{-2}\mathbb{E}\left(\|A_S\|^4\right)\right)\underbrace{\alpha^2\tau_{{\rm mix},\Phi}(\epsilon)}_{=t_{{\rm mix},\Phi}}+\epsilon\right)\,.
\end{aligned}
\]
Now, to achieve $\epsilon$-precision, we first need
\[
\left(\frac{\beta}{\omega_{\max}\sigma}\left(\sqrt{\log(\sigma/\beta)}+\log(\sigma \omega_{\max})\right)+\sigma\exp\left(-T^2/(4\sigma^2)\right)\right)\mathbb{E}(\|A_S\|^2)t_{{\rm mix},\Phi}=\mc{O}(\epsilon)\,,
\]
which implies
\[
\sigma=\widetilde{\Omega}\left(\beta \mathbb{E}(\|A_S\|^2)\omega^{-1}_{\max}t_{{\rm mix},\Phi}\epsilon^{-1}\right),\quad T=\widetilde{\Omega}\left(\sigma\right)\,.
\]
In addition, we also require
\[
\alpha^2T^4\sigma^{-2}\mathbb{E}\left(\|A_S\|^4\right)t_{{\rm mix},\Phi}=\mc{O}(\epsilon)\,,
\]
which implies
\[
\alpha=\mc{O}\left(\sigma T^{-2} t^{-1/2}_{{\rm mix},\Phi} \mathbb{E}^{-1/2}\left(\|A_S\|^4\right)\epsilon^{1/2}\right)\,.
\]
Plugging in $\sigma=\widetilde{\Theta}\left(\beta \mathbb{E}(\|A_S\|^2)\omega^{-1}_{\max}t_{{\rm mix},\Phi}\epsilon^{-1}\right)$, we conclude that
\[
\sigma=\widetilde{\Theta}\left(\beta \mathbb{E}(\|A_S\|^2)\omega^{-1}_{\max}t_{{\rm mix},\Phi}\epsilon^{-1}\right),\quad T=\widetilde{\Theta}\left(\beta \mathbb{E}(\|A_S\|^2)\omega^{-1}_{\max}t_{{\rm mix},\Phi}\epsilon^{-1}\right)\,,
\]
and
\[
\alpha=\widetilde{\Theta}\left(\sigma^{-1}t^{-1/2}_{{\rm mix},\Phi} \mathbb{E}^{-1/2}\left(\|A_S\|^4\right)\epsilon^{1/2}\right)=\widetilde{\Theta}\left(\beta^{-1}\omega_{\max}t^{-3/2}_{{\rm mix},\Phi}\epsilon^{3/2}\mathbb{E}^{-1}\left(\|A_S\|^2\right)\mathbb{E}^{-1/2}\left(\|A_S\|^4\right)\right)\,.
\]
This concludes~\cref{thm:fix_thermal}.
\end{proof}

Next, we prove~\cref{thm:fix_thermal_rigor}.
According to~\cref{thm:almost_fixed_point}~\cref{eqn:fixed_point_error_2}, we need to show the upper bound of $\|\Phi(\rho_\beta)-\rho_\beta\|_1$. According to Lemma~\ref{lem:thermal_approx_first} and
\begin{equation}\label{eqn:first_approx}
\|\Phi(\rho_\beta)-\rho_\beta\|_1\leq \left\|\Phi-\widetilde{\Phi}\right\|_{1\leftrightarrow1}+\|\widetilde{\Phi}(\rho_\beta)-\rho_\beta\|_1\,,
\end{equation}
it suffices to show $\left\|\widetilde{\Phi}(\rho_\beta)-\rho_\beta\right\|_1$ is small. Let $d$ be the dimension of $H$ and $H$ have an eigendecomposition $\{(\lambda_i,\ket{\psi_i})\}^{d-1}_{i=0}$ with $\lambda_0\leq \lambda_1\leq \dots,\lambda_{d-1}$. Because the unitary evolution $\mc{U}_S(T)$ preserves the thermal state, we have
\begin{equation}\label{eqn:tilde_phi_bound_thermal}
\left\|\widetilde{\Phi}\left(\rho_{\beta}\right)-\rho_{\beta}\right\|_1\leq \alpha^2\left\|\widetilde{\mc{L}}\left(\rho_{\beta}\right)\right\|_1\,,
\end{equation}
where \rev{$\widetilde{\mc{L}}$ is defined in~\eqref{eqn:L_tilde_thermal}}. In $\widetilde{\mc{L}}$, we consider the Lamb shift term and dissipative term separately. For the Lamb shift term, we have the following lemma:
\begin{lem}\label{lem:Lamb_shift_commute} When $T>\sigma$, we have
\[
\left\|\left[\widetilde{H}_{\rm LS},\rho_{\beta}\right]\right\|_1\leq \mathbb{E}_{A_S}\left(\left\| \left[\rho_\beta, \int \gamma(\omega)\widetilde{\mc{G}}_{A_S,f}(-\omega)\,\mathrm{d}\omega \right] \right\|_1+\left\|\left[\rho_\beta, \int \gamma(\omega)\left(\widetilde{\mc{G}}_{A_S,f}(-\omega)\right)^\dagger\,\mathrm{d}\omega \right] \right\|_1\right)=\mathcal{O}\left(R\beta\mathbb{E}(\|A_S\|^2)\right)
\]
\end{lem}
For the dissipative term, we have the following lemma:
\begin{lem}\label{lem:dissipative_thermal} When $T>\sigma>
\beta$, we have
\[
\left\|\mathbb{E}_{A_S}\left(\int^\infty_{-\infty}\gamma(\omega)\mathcal{D}_{\widetilde{V}_{A_S,f}(\omega)}(\rho_\beta)\mathrm{d}\omega\right)\right\|_1\leq \mathbb{E}_{A_S}\left(\left\|\int^\infty_{-\infty}\gamma(\omega)\mathcal{D}_{\widetilde{V}_{A_S,f}(\omega)}(\rho_\beta)\mathrm{d}\omega\right\|_1\right)=\mathcal{O}\left(\|\gamma(\omega)\|_\infty\mathbb{E}\left(\|A_S\|^2\right)\frac{\beta}{\sigma}\sqrt{\log(\sigma/\beta)}\right)
\]
\end{lem}
We put the proof of the above lemmas in the end of this section. Now, we are ready to prove~\cref{thm:fix_thermal_rigor}.
\begin{proof}[Proof of~\cref{thm:fix_thermal_rigor}]
Combining Lemma~\ref{lem:Lamb_shift_commute} and Lemma~\ref{lem:dissipative_thermal}, we have
\[
\left\|\widetilde{\mc{L}}\left(\rho_{\beta}\right)\right\|_1=\mathcal{O}\left(\left(R+\|\gamma(\omega)\|_\infty\frac{1}{\sigma}\sqrt{\log(\sigma/\beta)}\right)\beta\mathbb{E}\left(\|A_S\|^2\right)\right)
\]
Plugging this into~\cref{eqn:tilde_phi_bound_thermal} and using Lemma~\ref{lem:thermal_approx_first} and~\cref{thm:almost_fixed_point} with~\eqref{eqn:first_approx}, we conclude the proof.
\end{proof}

Finally, we complete the proof of Lemma~\ref{lem:Lamb_shift_commute} and Lemma~\ref{lem:dissipative_thermal}.
\begin{proof}[Proof of Lemma~\ref{lem:Lamb_shift_commute}]
Recall that
\[
\begin{aligned}
&\widetilde{H}_{\mathrm{LS}}=-\mathbb{E}_{A_S}\left(\mathrm{Im}\left(\int^\infty_{-\infty}\gamma(\omega)\widetilde{\mc{G}}_{A_S,f}(-\omega)\mathrm{d}\omega\right)\right)\\
=&\frac{-1}{\rev{2i}}\mathbb{E}_{A_S}\left(\int^\infty_{-\infty}\gamma(\omega)\widetilde{\mc{G}}_{A_S,f}(-\omega)\mathrm{d}\omega-\int^\infty_{-\infty}\gamma(\omega)\left(\widetilde{\mc{G}}_{A_S,f}(-\omega)\right)^\dagger\mathrm{d}\omega\right)
\end{aligned}\,.
\]

This implies
\[
\|[\rho_\beta,\widetilde{H}_{\rm LS}]\|_1\leq \rev{\frac{1}{2}}\mathbb{E}_{A_S}\left(\left\| \left[\rho_\beta, \int \gamma(\omega)\widetilde{\mc{G}}_{A_S,f}(-\omega)\,\mathrm{d}\omega \right] \right\|_1+\left\|\left[\rho_\beta, \int \gamma(\omega)\left(\widetilde{\mc{G}}_{A_S,f}(-\omega)\right)^\dagger\,\mathrm{d}\omega \right] \right\|_1\right)
\]
Thus, to show that $\|[\rho_\beta,\widetilde{H}_{\rm LS}]\|_1$ is small, it suffices to bound
\begin{equation}\label{eqn:commutator_small}
\left\| \left[\rho_\beta, \int \gamma(\omega)\widetilde{\mc{G}}_{A_S,f}(-\omega)\,\mathrm{d}\omega \right] \right\|_1
\quad \text{and} \quad
\left\| \left[\rho_\beta, \int \gamma(\omega)\left(\widetilde{\mc{G}}_{A_S,f}(-\omega)\right)^\dagger\,\mathrm{d}\omega \right] \right\|_1
\end{equation}
for all $\|A_S\| \leq 1$. The argument proceeds in two steps. First, we show that both
\[
\left[H, \int \gamma(\omega)\widetilde{\mc{G}}_{A_S,f}(-\omega)\,\mathrm{d}\omega \right]
\quad \text{and} \quad
\left[H, \int \gamma(\omega)\left(\widetilde{\mc{G}}_{A_S,f}(-\omega)\right)^\dagger\,\mathrm{d}\omega \right]
\]
are small (we omit the proof of the latter as it is analogous), which implies that $[H, \widetilde{H}_{\rm LS}]$ is small. Then, we expand $\rho_\beta$ as a polynomial in $H$ and express the commutators $\|[\rho_\beta, \cdot]\|_1$ as sums of nested commutators, from which we establish the smallness of~\eqref{eqn:commutator_small}.

We first calculate $\|[H,\widetilde{H}_{\rm LS}]\|$. For simplicity, we only consider $\widetilde{\mc{G}}_{A_S,f}(\omega)$. The calculation with $\left(\widetilde{\mc{G}}_{A_S,f}(\omega)\right)^\dagger$ should be quite similar.  Using change of variable $p=(s_1+s_2)/\sigma$ and $q=(s_1-s_2)/\sigma$,we have
\[
\begin{aligned}
&\widetilde{\mc{G}}_{A_S,f}(\omega)\\
=&\frac{\sigma^2}{2}\int^{\infty}_{-\infty}\mathrm{d}p\int^{+\infty}_{0}\mathrm{d}q f\left(\frac{\sigma(p+q)}{2}\right)f\left(\frac{\sigma(p-q)}{2}\right)A^\dagger_S\left(\frac{\sigma(p-q)}{2}\right)A_S\left(\frac{\sigma(p+q)}{2}\right)\exp(-i\omega\sigma q)
\end{aligned}
\]
Notice that
\[
\begin{aligned}
\left[H,A^\dagger_S\left(\frac{\sigma(p-q)}{2}\right)A_S\left(\frac{\sigma(p+q)}{2}\right)\right]=\frac{-2i}{\sigma}\frac{\mathrm{d}}{\mathrm{d}p}\left(A^\dagger_S\left(\frac{\sigma(p-q)}{2}\right)A_S\left(\frac{\sigma(p+q)}{2}\right)\right)
\end{aligned}\,.
\]
Thus,
\[
\begin{aligned}
  &[H,\widetilde{\mc{G}}_{A_S,f}(\omega)]
  =\frac{-i\sigma}{2\sqrt{2\pi}}\int^{\infty}_{-\infty}\int^{\infty}_{0}\exp(-p^2/8)\exp(-q^2/8)\\
  &\cdot\frac{2}{\sigma}\frac{\mathrm{d}}{\mathrm{d}p}\left(A^\dagger_S\left(\frac{\sigma(p-q)}{2}\right)A_S\left(\frac{\sigma(p+q)}{2}\right)\right)\exp(-i\omega \sigma q)\mathrm{d}q\mathrm{d}p\\
  =&\frac{\rev{-i}}{\sqrt{2\pi}}\int^{\infty}_{-\infty}\int^{\infty}_{0}\frac{p}{4}\exp(-p^2/8)\exp(-q^2/8)\\
  &\cdot A^\dagger_S\left(\frac{\sigma(p-q)}{2}\right)A_S\left(\frac{\sigma(p+q)}{2}\right)\exp(-i\omega \sigma q)\mathrm{d}q\mathrm{d}p\\
\end{aligned}
\]
We notice that
\[
\left\|A^\dagger_S(\sigma(p-q)/2)A_S(\sigma(p+q)/2)\right\|\leq \|A_S\|^2\,.
\]
thus,
\rev{
\[
\left\|\left[H, \int \gamma(\omega)\widetilde{\mc{G}}_{A_S,f}(-\omega)\,\mathrm{d}\omega \right]\right\|=\mathcal{O}\left(R\mathbb{E}(\|A_S\|^2)\right)\,.
\]
This implies that
}
\[
\left\|\left[H, \widetilde{H}_{\rm LS}\right]\right\|=\mathcal{O}\left(R\mathbb{E}(\|A_S\|^2)\right)\,.
\]

Next, we notice that
\[
\begin{aligned}
&\|[\rho_\beta,\widetilde{H}_{\rm LS}]\|_1\leq \left\|\rho_
\beta\right\|_1\|\rho^{-1}_{\beta}\widetilde{H}_{\rm LS}\rho_{\beta}-\widetilde{H}_{\rm LS}\|=\|\rho^{-1}_{\beta}\widetilde{H}_{\rm LS}\rho_{\beta}-\widetilde{H}_{\rm LS}\|\\
=&\rev{\mathbb{E}_{A_S}\left\|\mathrm{Im}\left(\int^\infty_{-\infty}\gamma(\omega)\left(\rho^{-1}_\beta\widetilde{\mc{G}}_{A_S,f}(\omega)\rho_\beta-\widetilde{\mc{G}}_{A_S,f}(\omega)\right)\mathrm{d}\omega\right)\right\|}\,.
\end{aligned}
\]
We can use the BCH formula to expand the term $\rho_\beta^{-1}\widetilde{\mc{G}}_{A_S,f}(\omega)\rho_\beta$ as a series.
\[
\begin{aligned}
\rho_\beta^{-1}\widetilde{\mc{G}}_{A_S,f}(\omega)\rho_\beta-\widetilde{\mc{G}}_{A_S,f}(\omega)&=\mathrm{e}^{\beta H}\widetilde{\mc{G}}_{A_S,f}(\omega) \mathrm{e}^{-\beta H}-\widetilde{\mc{G}}_{A_S,f}(\omega)\\
&=\beta[H, \widetilde{\mc{G}}_{A_S,f}(\omega)]+\frac{\beta^2}{2}[H,[H, \widetilde{\mc{G}}_{A_S,f}(\omega)]]+\ldots \frac{\beta^n}{n!}[\overbrace{H,[H, \ldots[H}^{n H^{\prime} s}, \widetilde{\mc{G}}_{A_S,f}(\omega)] . .]+\ldots .
\end{aligned}
\]
Using change of variable $p=(s_1+s_2)/\sigma$ and $q=(s_1-s_2)/\sigma$,
similar to the previous calculation
\[
\begin{aligned}
\left[H,\widetilde{\mc{G}}_{A_S,f}(\omega)\right]=&\frac{\sigma}{2\sqrt{2\pi}}\frac{-2i}{\sigma}\int^{\infty}_{-\infty}\mathrm{d}p\int^{+\infty}_{0}\mathrm{d}q \exp(-p^2/8)\exp(-q^2/8)\\
&\cdot\exp(-i\omega\sigma q) \frac{\mathrm{d}}{\mathrm{d}p}\left(A^\dagger_S\left(\frac{\sigma(p-q)}{2}\right)A_S\left(\frac{\sigma(p+q)}{2}\right)\right)\\
=&\frac{\sigma}{2\sqrt{2\pi}}\frac{2i}{\sigma}\int^{\infty}_{-\infty}\mathrm{d}p\int^{+\infty}_{0}\mathrm{d}q \frac{\mathrm{d}}{\mathrm{d}p}\exp(-p^2/8)\exp(-q^2/8)\\
&\cdot \exp(-i\omega\sigma q) A^\dagger_S\left(\frac{\sigma(p-q)}{2}\right)A_S\left(\frac{\sigma(p+q)}{2}\right)\,.
\end{aligned}
\]
Applying this iteratively, we have the commutator form:
\[
\begin{aligned}
[\overbrace{H,[H, \ldots[H}^{n H' s}, \widetilde{\mc{G}}_{A_S,f}(\omega)] . .]=&\frac{\sigma}{2\sqrt{2\pi}}\left(\frac{2i}{\sigma}\right)^n\int^{\infty}_{-\infty}\mathrm{d}p\int^{+\infty}_{0}\mathrm{d}q \frac{\mathrm{d}^n}{\mathrm{d}p^n}\exp(-p^2/8)\exp(-q^2/8)\\
&\cdot \exp(-i\omega\sigma q) A^\dagger_S\left(\frac{\sigma(p-q)}{2}\right)A_S\left(\frac{\sigma(p+q)}{2}\right)
\end{aligned}\,.
\]
Notice that:
\[
\left|\frac{\mathrm{d}^n}{\mathrm{d}p^n}\exp(-p^2/8)\right|<2\sqrt{n!}2^{-n}\exp(-p^2/16)\,.
\]
As a result, following the proof of the previous lemma, the $n$-th term of the series can be bounded by
\[
\begin{aligned}
\left\|\int^{\infty}_{-\infty}\gamma(\omega)\mathrm{d}\omega \frac{\beta^n}{n!}[\overbrace{H,[H, \ldots[H}^{n H^{\prime} s}, \widetilde{\mc{G}}_{A_S,f}(\omega)] . .]\right\|&=\mathcal{O}\left(\frac{\sigma}{2\sqrt{2\pi}}\left(\frac{2\beta}{\sigma}\right)^n \frac{\rev{1}}{\sqrt{n!}}2^{-n}R\|A_S\|^2\right)\\
&=\mathcal{O}\left( \frac{\beta^{n-1}}{\sigma^{n-1}\sqrt{n!}}R\beta\|A_S\|^2\right)
\end{aligned}
\]
Summing all terms still gives
\[
\left\|\int^\infty_{-\infty}\gamma(\omega)\left(\rho^{-1}_\beta\widetilde{\mc{G}}_{A^i,f}(\omega)\rho_\beta-\widetilde{\mc{G}}_{A^i,f}(\omega)\right)\mathrm{d}\omega\right\|=\mathcal{O}\left(R\beta\|A_S\|^2\right)
\]
Similar result can be proved for $\left(\widetilde{\mc{G}}_{A_S,f}(\omega)\right)^\dagger$. We conclude the proof.
\end{proof}

\begin{proof}[Proof of Lemma~\ref{lem:dissipative_thermal}] Let $\mc{B}$ be a Lindbladian, define
\[
\mc{K}(\rho_\beta,\mc{B})=\rho_\beta^{-1/4}\mc{B}[\rho^{1/4}_\beta\cdot \rho^{1/4}_\beta]\rho_\beta^{-1/4}
\]
with
\[
\left(\mc{K}(\rho_\beta,\mc{B})\right)^\dagger=\rho_\beta^{1/4}\mc{B}^
\dagger[\rho^{-1/4}_\beta\cdot \rho^{-1/4}_\beta]\rho_\beta^{1/4}\,.
\]
We note that if $\mc{K}(\rho_\beta,\mc{B})=\left(\mc{K}(\rho_\beta,\mc{B})\right)^\dagger$, we have
\[
\rho_\beta^{-1/4}\mc{B}[\rho_\beta]\rho_\beta^{-1/4}=\mc{K}(\rho_\beta,\mc{B})[\sqrt{\rho_\beta}]=(\mc{K}(\rho_\beta,\mc{B}))^{\dag}[\sqrt{\rho_\beta}]=\rho_\beta^{1/4}\mc{B}^{\dagger}[I]\rho_\beta^{1/4}=0\,.
\]
This implies $\mc{B}$ fixes the thermal state $\rho_\beta$. Furthermore, \begin{equation}\label{eqn:B_rho_bound}
\begin{aligned}
&\left\|\mc{K}(\rho_\beta,\mc{B})-\left(\mc{K}(\rho_\beta,\mc{B})\right)^\dagger\right\|_{2\leftrightarrow 2}=\left\|\mc{K}(\rho_\beta,\mc{B})-\left(\mc{K}(\rho_\beta,\mc{B})\right)^\dagger\right\|_{2\leftrightarrow 2}\|\sqrt{\rho_\beta}\|_2\\
\geq &\left\|\mc{K}(\rho_\beta,\mc{B})[\sqrt{\rho_\beta}]-\left(\mc{K}(\rho_\beta,\mc{B})\right)^\dagger[\sqrt{\rho_\beta}]\right\|_{2}=\left\|\rho_\beta^{-1/4}\mc{B}[\rho_\beta]\rho_\beta^{-1/4}\right\|_2\\
= &\left\|\rho_\beta^{-1/4}\mc{B}[\rho_\beta]\rho_\beta^{-1/4}\right\|_2\|\rho_{\beta}^{1/4}\|^2_4\geq \|\mc{B}[\rho_\beta]\|_1
\end{aligned}\,.
\end{equation}
Here $\|\cdot\|_p$ is the Schattern-$p$ norm defined in~\cref{sec:notation}. In the last inequality, we use H\"older's inequality $\|BAB\|_1\leq \|B\|^2_4\|A\|_2$. This inequality implies that, if $\mc{K}(\rho_\beta,\mc{B})$ is approximately self-adjoint, $\mc{B}$ can also \rev{approximately} preserve the thermal state.

The rest of the proof follows a similar procedure as the proof of~
\cite[Theorem I.3]{ChenKastoryanoBrandaoEtAl2023} to show that $\|\mc{D}[\rho_\beta]\|_1$ is small. Let $\mc{D} =\mathbb{E}_{A_S}\left(\int^\infty_{-\infty}\gamma(\omega)\mathcal{D}_{\widetilde{V}_{A_S,f}(-\omega)}(\rho_\beta)\mathrm{d}\omega\right)$. In the proof of~\cite[Theorem I.3]{ChenKastoryanoBrandaoEtAl2023}, the authors first approximate $\mc{D}$ with secular version $\mc{D}_{sec}$ (See~\cite[Lemma A.2]{ChenKastoryanoBrandaoEtAl2023}). The secular approximation is an artificial cutoff in frequency space on the transition energies induced by Lindblad jump operators, which causes only a small error when $\sigma$ is sufficiently large. Following the proof of~\cite[Theorem I.3]{ChenKastoryanoBrandaoEtAl2023},we have
\[
\left\|\mc{D}_{sec}-\mc{D}\right\|_{1\leftrightarrow 1}+\left\|\mc{K}(\rho_\beta,\mc{D}_{sec})(\rho_\beta)-\left(\mc{K}(\rho_\beta,\mc{D}_{sec})(\rho_\beta)\right)^\dagger\right\|_{2\leftrightarrow 2}=\mathcal{O}\left(\|\gamma(\omega)\|_\infty\mathbb{E}\left(\|A_S\|^2\right)\frac{\beta}{\sigma}\sqrt{\log(\sigma/\beta)}\right)\,.
\]
This implies that
\[
\begin{aligned}
&\|\mc{D}(\rho_\beta)\|_1\leq \left\|\mc{D}_{sec}-\mc{D}\right\|_{1\leftrightarrow 1}+ \|\mc{D}_{sec}(\rho_\beta)\|_1\\
\leq &\left\|\mc{D}_{sec}-\mc{D}\right\|_{1\leftrightarrow 1}+\left\|\mc{K}(\rho_\beta,\mc{D}_{sec})(\rho_\beta)-\left(\mc{K}(\rho_\beta,\mc{D}_{sec})(\rho_\beta)\right)^\dagger\right\|_{2\leftrightarrow 2}=\mathcal{O}\left(\|\gamma(\omega)\|_\infty\mathbb{E}\left(\|A_S\|^2\right)\frac{\beta}{\sigma}\sqrt{\log(\sigma/\beta)}\right)\,.
\end{aligned}
\]
In the \rev{second inequality}, we use~\eqref{eqn:B_rho_bound}.
\end{proof}

\subsection{Approximate fixed point -- Ground state}\label{sec:app_fix_gs}
In this section, we provide a rigorous version of~\cref{thm:fix_ground}
in~\cref{thm:fix_ground_rigor} and provide the proof. We consider~(2) with  $\beta=\infty$ and $f(t)= \frac{1}{(2\pi)^{1/4} \sigma^{1/2}} \exp\left(-\frac{t^2}{4\sigma^2}\right)$. Similar to the thermal state case, we first consider a simplified CPTP map by removing the error terms in Theorem 4 and take the limit $T \to \infty$, as mentioned in~\cref{sec:notation},
\begin{equation}\label{eqn:Phi_tilde}
\widetilde{\Phi}=\mc{U}_S(T)\,\circ\,\exp\left(\widetilde{\mathcal{L}}\alpha^2\right)\,\circ\,\mc{U}_S(T)\,.
\end{equation}
Here
\begin{equation}\label{eqn:L_tilde}
\begin{aligned}
\widetilde{\mc{L}}(\rho)=\mathbb{E}_{A_S}\left(-i\left[\widetilde{H}_{\mathrm{LS},A_S},\rho\right]+\int^0_{-\infty}(g(\omega)+g(-\omega))\mathcal{D}_{\widetilde{V}_{A_S,f}(\omega)}(\rho)\mathrm{d}\omega\right)\,,
\end{aligned}
\end{equation}
where
\[
\widetilde{H}_{\mathrm{LS},A_S}=\rev{-}\mathrm{Im}\left(\int^0_{-\infty}g(\omega)\widetilde{\mc{G}}_{A^\dagger_S,f}(\omega)\mathrm{d}\omega+\int^\infty_{0}g(\omega)\widetilde{\mc{G}}_{A_S,f}(-\omega)\mathrm{d}\omega\right),\quad \widetilde{V}_{A_S,f}(\omega)=\int^\infty_{-\infty}f(t)A_S(t)\exp(-i\omega t)\mathrm{d}t\,,
\]
with
\[
\widetilde{\mc{G}}_{A_S,f}(\omega)=\int^\infty_{-\infty}\int^{s_1}_{-\infty}f(s_2)f(s_1) A^\dagger_S(s_2)A_S(s_1)\exp(i\omega(s_2-s_1))\mathrm{d}s_2\mathrm{d}s_1\,.
\]
Same as Lemma~\ref{lem:thermal_approx_first}, the error between $\Phi$ and $\widetilde{\Phi}$ can be controlled in the following lemma:
\begin{lem}\label{lem:gd_approx_first} When $T>\sigma$, We have
\[
\left\|\Phi-\widetilde{\Phi}\right\|_{1\leftrightarrow1}=\mathcal{O}\left(\alpha^2\sigma\exp\left(-T^2/(4\sigma^2)\right)\mathbb{E}(\|A_S\|^2)+\alpha^4T^4\sigma^{-2}\mathbb{E}\left(\|A_S\|^4\right)\right)
\]
\end{lem}
The proof of Lemma~\ref{lem:gd_approx_first} is almost the same as the proof of Lemma~\ref{lem:thermal_approx_first}. Thus, we omit it.  Using $\widetilde{\Phi}$, we are ready to state the rigorous version of~\cref{thm:fix_ground} and provide the proof:
\begin{thm}\label{thm:fix_ground_rigor}
Assume $H$ has a spectral gap $\Delta$ and $T>\sigma$. Then, for any $\epsilon>0$,

\[
\begin{aligned}
&\|\rho_{\rm fix}(\Phi)-\rev{\ket{\psi_0}\bra{\psi_0}}\|_1\\
=&\mathcal{O}\left(\left(\|H\|^{1/2}\sigma^{3/2}\exp\left(-\sigma^2\Delta^2/\rev{8}\right)\mathbb{E}(\|A_S\|^2)+\sigma\exp\left(-T^2/(4\sigma^2)\right)\mathbb{E}(\|A_S\|^2)+\alpha^2T^4\sigma^{-2}\mathbb{E}\left(\|A_S\|^4\right)\right)\alpha^2\tau_{{\rm mix},\Phi}(\epsilon)+\epsilon\right)\\
\end{aligned}
\]
\end{thm}
The proof of this theorem follows a similar approach to that of~\cref{thm:fix_thermal_rigor}, where we demonstrate that both the Lamb shift term and the dissipative term approximately preserve the ground state. Although the overall proof strategy is similar, we adopt a different technique in the proof below. Specifically, using the spectral gap $\Delta$, we directly establish a small fixed-point error when the number of Bohr frequencies is constant. In the general case, where the number of Bohr frequencies cannot be bounded, we approximate the Hamiltonian by a rounded version with a controllable number of eigenvalues, inspired by the secular approximation idea in~\cite{ChenKastoryanoGilyen2023}. The errors introduced in the Lamb shift and dissipative terms due to this rounding can also be controlled by exploiting the Gaussian structure of $f$. Furthermore, when handling the Lamb shift term, the rounding technique and the spectral gap assumption allow us to establish a uniform fixed-point error bound for $\left\| \left[\rev{\ket{\psi_0}\bra{\psi_0}}, \widetilde{\mc{G}}_{A_S,f}(-\omega)\right] \right\|_1+\left\|\left[\rev{\ket{\psi_0}\bra{\psi_0}}, \left(\widetilde{\mc{G}}_{A_S,f}(-\omega)\right)^\dagger\right] \right\|_1+\left\|\mathcal{D}_{\widetilde{V}_{A_S,f}(\omega)}(\rev{\ket{\psi_0}\bra{\psi_0}})\right\|_1$ in $\omega$ prior to taking the expectation over $\omega$. This enables the use of an arbitrary distribution $g$ in the theorem above.

\begin{proof}[Proof of~\cref{thm:fix_ground_rigor}]


Similar to the proof of~\cref{thm:fix_thermal_rigor}, it suffices to show $\left\|\widetilde{\Phi}(\rev{\ket{\psi_0}\bra{\psi_0}})-\rev{\ket{\psi_0}\bra{\psi_0}}\right\|_1$ is small. Because the unitary evolution $\mc{U}_S(T)$ preserves the ground state, we have
\begin{equation}\label{eqn:tilde_phi_bound}
\left\|\widetilde{\Phi}\left(\ket{\psi_0}\bra{\psi_0}\right)-\ket{\psi_0}\bra{\psi_0}\right\|_1\leq \alpha^2\left\|\widetilde{\mc{L}}\left(\ket{\psi_0}\bra{\psi_0}\right)\right\|_1\,,
\end{equation}
where $\widetilde{\mc{L}}$ is defined in~\cref{eqn:L_tilde}.

Now, we consider the Lamb shift term and dissipative term separately. For simplicity, we consider a fixed $A_S$ in the following calculation. Recall that
\[
A_{S}(\nu)=\sum_{\lambda_j-\lambda_i=\nu}\ket{\psi_j}\bra{\psi_j}A_S\ket{\psi_i}\bra{\psi_i},\quad
A^\dagger_{S}(\nu)=\sum_{\lambda_j-\lambda_i=\nu}\ket{\psi_j}\bra{\psi_j}A^\dagger_S\ket{\psi_i}\bra{\psi_i}.
\]
\begin{itemize}
\item Lamb shift term: Recall the definition of $\widetilde{\mc{G}}_{A_S,f}$: \[
\widetilde{\mc{G}}_{A_S,f}(-\omega)=\int^\infty_{-\infty}\int^{s_1}_{-\infty}f(s_2)f(s_1) A^
\dagger_S(s_2)A_S(s_1)\exp(-i\omega(s_2-s_1))\mathrm{d}s_2\mathrm{d}s_1\,.
\]

Using change of variable $p=(s_1+s_2)/\sigma$ and $q=(s_1-s_2)/\sigma$,
\[
\begin{aligned}
&\int^\infty_{-\infty}\int^{s_1}_{-\infty}f(s_2)f(s_1) A^
\dagger_S(s_2)A_S(s_1)\exp(-i\omega(s_2-s_1))\mathrm{d}s_2\mathrm{d}s_1\\
=&\sum_{\nu_1,\nu_2\in B(H)}A^\dagger_S(\nu_2)A_S(\nu_1)\int^{\infty}_{-\infty}\int^{s_1}_{-\infty}f(s_2)f(s_1)\exp(i\nu_2s_2)\exp(i\nu_1s_1)\exp(-i\omega(s_2-s_1))\mathrm{d}s_2\mathrm{d}s_1\\
=&\frac{\sigma}{2\sqrt{2\pi}}\sum_{\nu_1,\nu_2\in B(H)}A^\dagger_S(\nu_2)A_S(\nu_1)\\
&\cdot\underbrace{\int^{\infty}_{-\infty}\exp\left(i\frac{\sigma p}{2}(\nu_1+\nu_2)\right)\exp\left(-\frac{p^2}{8}\right)\mathrm{d}p}_{=\mc{O}(\exp(-\sigma^2(\nu_1+\nu_2)^2/2))}\underbrace{\int^{\infty}_{0}\exp\left(-\frac{q^2}{8}\right)\exp\left(i\frac{\sigma q}{2}(\nu_1-\nu_2)\right)\exp(i\sigma\omega q)\mathrm{d}q}_{=\mc{O}(1)}
\end{aligned}
\]
where $B(H)$ is the set of Bohr frequencies.

We note that
\[
\left[\ket{\psi_0}\bra{\psi_0},A^\dagger_{S}(\nu_2)A_S(\nu_1)\right]=0
\]
when $\left|\nu_2+\nu_1\right|<\Delta$. We show this using the proof by contradiction: When $\left[\ket{\psi_0}\bra{\psi_0},A^\dagger_{S}(\nu_2)A_S(\nu_1)\right]\neq 0$, we must have $\ket{\psi_0}\bra{\psi_0}A^\dagger_{S}(\nu_2)A_S(\nu_1)\neq 0$ or $A^\dagger_S(\nu_2)A_{S}(\nu_1)\ket{\psi_0}\bra{\psi_0}\neq 0$. We consider these two cases separately:
\begin{itemize}
    \item In the first case, we have $\left(A^\dagger_{S}(\nu_2)\right)^\dagger\ket{\psi_0}=A_{S}(-\nu_2)\ket{\psi_0}\neq 0$, which implies $\nu_2\leq 0$. Now, since $\left|\nu_2+\nu_1\right|<\Delta$ and $\ket{\psi_0}\bra{\psi_0}A^\dagger_{S}(\nu_2)A_S(\nu_1)\neq 0$, we have $\nu_1=-\nu_2$. This implies $\left[\ket{\psi_0}\bra{\psi_0},A^\dagger_{S}(\nu_2)A_S(\nu_1)\right]=0$.
    \item In the second case, we have $A_{S}(\nu_1)\ket{\psi_0}\neq 0$, which implies $\nu_1\geq 0$. Now, since $\left|\nu_2+\nu_1\right|<\Delta$ and $A^\dagger_{S}(\nu_2)A_S(\nu_1)\ket{\psi_0}\neq 0$, we have $\nu_1=-\nu_2$. This implies $\left[\ket{\psi_0}\bra{\psi_0},A^\dagger_{S}(\nu_2)A_S(\nu_1)\right]=0$.
\end{itemize}
These two cases give a contradiction.  This implies
\begin{equation}\label{eqn:coherent_term_small}
  \left[\ket{\psi_0}\bra{\psi_0},\widetilde{\mc{G}}_{A_S,f}(-\omega)\right]=\sum_{\left|\nu_2+\nu_1\right|\geq \Delta}\underbrace{F(\nu_1,\nu_2)}_{|F(\nu_1,\nu_2)|=\mc{O}(\sigma\exp(-\sigma^2\Delta^2/2))}\left[\ket{\psi_0}\bra{\psi_0},A^\dagger_S(\nu_2)A_S(\nu_1)\right]\,.
\end{equation}

Now, we are ready to show~\eqref{eqn:coherent_term_small} is small. First, let us assume $H$ has discrete eigenvalues in $[-\|H\|,\|H\|]$ with uniform gap $\eta$, meaning $|\lambda_i-\lambda_j|\geq \eta$ if $\lambda_i\neq \lambda_j$. This implies $|B(H)|=\mc{O}(\|H\|/\eta)$, where $|B(H)|$ means the number of elements in $B(H)$.  Then,
\[
\begin{aligned}
  &\left\|\sum_{\nu_2\leq 0,\nu_1\geq 0}\left(\dots\right)\right\|=\mc{O}\left(\|A_S\|^2|B(H)|\sigma\exp\left(-\sigma^2\Delta^2/2\right)\right)\\
  =&\mc{O}\left(\|A_S\|^2\|H\|\sigma\exp\left(-\sigma^2\Delta^2/2\right)/\eta\right)\,.
\end{aligned}
\]
Because every Hamiltonian can be approximated by a rounding Hamiltonian $H_\eta$ such that: 1. $\|H-H_\eta\|\leq \eta$; 2. $H_\eta$ has the same ground state; 3. $H_{\eta}$ has discrete eigenvalues in $[-\|H\|,\|H\|]$ with uniform gap $\eta$. We conclude that
\begin{equation}\label{eqn:coherent_term_small_final}
\begin{aligned}
&\left\|\sum_{\nu_2\leq 0,\nu_1\geq 0}\left(\dots\right)\right\|\\
=&\mc{O}\left(\|A_S\|^2\|H\|\sigma\exp\left(-\sigma^2\Delta^2/2\right)/\eta\right)+\mc{O}\left(\|A_S\|^2\int^\infty_{-\infty}\int^{s_1}_{-\infty}f(s_1)f(s_2)(|s_1|+|s_2|)\eta\mathrm{d}s_1\mathrm{d}s_2\right)\\
=&\mc{O}\left(\|A_S\|^2\min_{\eta}\left(\sigma\exp\left(-\sigma^2\Delta^2/2\right)\|H\|/\eta+\eta\sigma^2\right)\right)\\
=&\mc{O}\left(\|A_S\|^2\|H\|^{1/2}\sigma^{3/2}\exp\left(-\left(\sigma^2\Delta^2/4\right)\right)\right)
\end{aligned}
\end{equation}
Here, the second term arises from approximating $\left[\ket{\psi_0}\bra{\psi_0},\widetilde{\mathcal{G}}_{A_S,f}(-\omega)\right]$ by replacing $H$ with $H_{\eta}$. This concludes the calculation for the Lamb shift term.

\item Dissipative term: When $f(t)=\frac{1}{(2\pi)^{1/4}\sigma^{1/2}}\exp(-t^2/(4\sigma^2))$,
\[
\widetilde{V}_{A_S,f}(\omega)=\int^\infty_{-\infty}f(t)A_S(t)\exp(-i\omega t)\mathrm{d}t=\rev{2^{3/4}\pi^{1/4}}\sqrt{\sigma}\sum_{\nu\in B(H)}\exp(-\rev{(\nu-\omega)^2\sigma^2})A_S(\nu)\,.
\]
\rev{Define the component of $V$ that preserves the ground state as $V^+$:
\[
V^+_{A_S,f}(\omega)=\left\{
\begin{aligned}
&2^{3/4}\pi^{1/4}\sqrt{\sigma}\sum_{\nu\in B(H),\nu<0}\exp(-(\nu-\omega)^2\sigma^2)A(\nu),\quad \omega<-\frac{\Delta}{2}\\
&2^{3/4}\pi^{1/4}\sqrt{\sigma}\sum_{\text{If $i=0$ or $j=0$, then $i+j=0$}}\exp(-(\lambda_i-\lambda_j-\omega)^2\sigma^2)\bra{\psi_i}A\ket{\psi_j} \ket{\psi_i}\bra{\psi_j},\quad -\frac{\Delta}{2}\leq \omega<0
\end{aligned}
\right.
\]
Here, $\ket{\psi_i}$ is the eigenvector of $H$ with eigenvalue $\lambda_i$ with $\lambda_0,\ket{\psi_0}$ being the ground state energy and ground state. Recall~\eqref{eqn:L_tilde}:
\[
\widetilde{\mc{L}}(\rho)=\mathbb{E}_{A_S}\left(-i\left[\widetilde{H}_{\mathrm{LS},A_S},\rho\right]+\int^0_{-\infty}(g(\omega)+g(-\omega))\mathcal{D}_{\widetilde{V}_{A_S,f}(\omega)}(\rho)\mathrm{d}\omega\right)\,.
\]
We will show that the choice of $V^+_{A_S,f}(\omega)$ ensures that: 1. $\ket{\psi_0
}\bra{\psi_0}\in\mathrm{Ker}\left(\mathcal{D}_{V^+_{A_S,f}(\omega)}\right)$ for any $\omega<0$; 2. $\widetilde{V}_{A_S,f}(\omega)$ is close to $V^+_{A_S,f}(\omega)$. Consider two cases:
\begin{itemize}
    \item When $\omega<-\Delta/2$, $\ket{\psi_0
}\bra{\psi_0}\in\mathrm{Ker}\left(\mathcal{D}_{V^+_{A_S,f}(\omega)}\right)$ is straightforward because $V^+_{A_S,f}(\omega)\ket{\psi_0}=0$. To show that $\widetilde{V}_{A_S,f}(\omega)$ is close to $V^+_{A_S,f}(\omega)$, we use the rounding Hamiltonian technique similar to the calculation for the Lamb shift term. First, let us assume $H$ has discrete eigenvalues in $[-\|H\|,\|H\|]$ with uniform gap $\eta$, meaning $|\lambda_i-\lambda_j|=\eta$ if $\lambda_i\neq \lambda_j$. This implies $|B(H)|=\mc{O}(\|H\|/\eta)$. Then, for $\omega\geq 0$,
\begin{equation}\label{eqn:round_first}
\left\|\widetilde{V}_{A_S,f}(\omega)-V^+_{A_S,f}(\omega)\right\|=\mc{O}\left(\|A_S\||B(H)|\sqrt{\sigma}\exp\left(-\sigma^2\Delta^2/4\right)\right)=\mc{O}\left(\|A_S\|\|H\|\sqrt{\sigma}\exp\left(-\sigma^2\Delta^2/4\right)/\eta\right)\,.
\end{equation}
Similar to the Lamb shift term, we approximated the Hamiltonian by the rounding Hamiltonian $H_\eta$ such that $\|H-H_\eta\|\leq \eta$ and $H_\eta$ has discrete eigenvalues in $[-\|H\|,\|H\|]$ with uniform gap $\eta$. We conclude that, for general $H$,
\begin{equation}\label{eqn:V_close}
\begin{aligned}
&\left\|\widetilde{V}_{A_S,f}(\omega)-V^+_{A_S,f}(\omega)\right\|=\mc{O}\left(\min_\eta\left(\|A_S\|\|H\|\sqrt{\sigma}\exp\left(-\sigma^2\Delta^2/4\right)/\eta+\|A_S\|\eta \underbrace{\|tf(t)\|_{L^1}}_{=\mc{O}(\sigma^{3/2})}\right)\right)\\
=&\mc{O}\left(\|A_S\|\|H\|^{1/2}\sigma \exp(-\sigma^2\Delta^2/8)\right)
\end{aligned}\,.
\end{equation}
\item When $-\Delta/2<\omega\leq 0$, we can rewrite
\[
V^+_{A_S,f}(\omega)=(\dots)\ket{\psi_0}\bra{\psi_0}+\sum_{i,j\neq 0}(\dots)\ket{\psi_i}\bra{\psi_j}\,.
\]
This ensures that $[V^+_{A_S,f}(\omega),\ket{\psi_0}\bra{\psi_0}]$ and thus $\ket{\psi_0
}\bra{\psi_0}\in\mathrm{Ker}\left(\mathcal{D}_{V^+_{A_S,f}(\omega)}\right)$. Next, to show $\widetilde{V}_{A_S,f}(\omega)$ is close to $V^+_{A_S,f}(\omega)$, we note that
\[
\begin{aligned}
V_{A_S,f}(\omega)=&V^+_{A_S,f}(\omega)+2^{3/4}\pi^{1/4}\sigma^{1/2}\sum_{i\neq 0}\exp(-(\lambda_i-\lambda_0-\omega)^2\sigma^2)\bra{\psi_i}A\ket{\psi_0} \ket{\psi_i}\bra{\psi_0}\\
&+2^{3/4}\pi^{1/4}\sigma^{1/2}\sum_{i\neq 0}\exp(-(\lambda_0-\lambda_i-\omega)^2\sigma^2)\bra{\psi_0}A\ket{\psi_i} \ket{\psi_0}\bra{\psi_i}\,.
\end{aligned}
\]
In the above summation, since $i\neq 0$ and $H$ has spectral gap $\Delta$, we have $|\lambda_i-\lambda_0|\geq \Delta$ and $|\lambda_i-\lambda_0-\omega|\geq \Delta/2$ when $-\Delta/2<\omega\leq 0$. This guarantees that each term in the summation can be upper bounded, meaning
\[
\left\|2^{3/4}\pi^{1/4}\sigma^{1/2}\sum_{\lambda_i=\lambda}\exp(-(\lambda_0-\lambda_i-\omega)^2\sigma^2)\bra{\psi_0}A\ket{\psi_i} \ket{\psi_0}\bra{\psi_i}\right\|=\mc{O}\left(\|A_S\|\sqrt{\sigma}\exp\left(-\sigma^2\Delta^2/4\right)\right)
\]
for each eigenvalue $\lambda$. Thus, similar to the first case, we also have~\eqref{eqn:round_first} and~\eqref{eqn:V_close}.
\end{itemize}}
\rev{Because both cases satisfy~\eqref{eqn:V_close}}, we have
\begin{equation}\label{eqn:dissipative_small_final}
\begin{aligned}
  &\left\|\mc{L}_{\widetilde{V}_{A_S,f}(\omega)}\left(\ket{\psi_0}\bra{\psi_0}\right)\right\|_{1}=\rev{\mc{O}\left(\left\|\mc{L}_{\widetilde{V}_{A_S,f}(\omega)}-\mc{L}_{V^+_{A_S,f}(\omega)}\right\|_{1\leftrightarrow1}\right)}=\mc{O}\left(\rev{\left\|\widetilde{V}_{A_S,f}(\omega)-V^+_{A_S,f}(\omega)\right\|\left\|\widetilde{V}_{A_S,f}(\omega)\right\|}\right)\\
  =&\mc{O}\left(\|A_S\|^2\|H\|^{1/2}\sigma^{3/2} \exp(-\sigma^2\Delta^2/\rev{8})\right)\,.
\end{aligned}
\end{equation}
\end{itemize}
Combining~\eqref{eqn:coherent_term_small_final} and~\eqref{eqn:dissipative_small_final}, we have
\[
\left\|\widetilde{\mc{L}}\left(\ket{\psi_0}\bra{\psi_0}\right)\right\|_1=\mathcal{O}\left(\|A_S\|^2\|H\|^{1/2}\sigma^{3/2} \exp(-\sigma^2\Delta^2/\rev{8})\right)
\]
Plugging this into~\eqref{eqn:tilde_phi_bound},
\[
\left\|\widetilde{\Phi}\left(\ket{\psi_0}\bra{\psi_0}\right)-\ket{\psi_0}\bra{\psi_0}\right\|_1=\mathcal{O}\left(\alpha^2\sigma^{3/2} \exp(-\sigma^2\Delta^2/\rev{8})\|A_S\|^2\|H\|^{1/2}\right)
\]
This concludes the proof.
\end{proof}

\section{Mixing time and End-to-end efficiency analysis}\label{sec:mixing_time}

As concrete examples to guarantee fast mixing, in this section, we choose $\mathcal{A}$ to be the set of all single-qubit Pauli operators (and their negatives) for qubit systems, and the set of creation and annihilation operators (and their negatives) for fermionic systems. For the functions $g$ and $f$, we set
\[
\rev{g(\omega)=\frac{1}{\omega_{\max}}\mathbf{1}_{[0,\omega_{\max}]}},\  f(t) = \frac{1}{(2\pi)^{1/4} \sigma^{1/2}} \exp\left(-\frac{t^2}{4\sigma^2}\right)\,.
\]
The parameters $\omega_{\max}$ are selected so that the system-bath interaction can induce energy transitions effectively.
 The choice of $\omega_{\max}$ can be system-dependent and should generally be at least as large as the largest eigenvalue gap, and typically does not grow with system size. The parameter $\sigma$ in the filter function $f(t)$ is typically chosen to be sufficiently large to ensure that the Lamb shift term in Theorem 4 approximately commutes with the thermal or ground state, as discussed in~\cref{sec:thermal_ground_prep}.

According to \cref{thm:fix_thermal} and \cref{thm:fix_ground}, to establish end-to-end efficiency, it suffices to provide an upper bound on $t_{{\rm mix},\Phi}$. However, we emphasize that in \cref{thm:fix_thermal} and \cref{thm:fix_ground}, the mixing time $t_{{\rm mix},\Phi}$ and the parameter $\sigma$ are  \emph{not} independent of each other. The bath Hamiltonian $H_B$, the coupling operator $B_E$, and the filter function $f(t)$ must be carefully designed to ensure that the conditions required for the theorems are meaningfully satisfied.

In this section, we provide the result of upper bounding the mixing time of the map $\Phi$ defined in~(2) and a complete end-to-end efficiency analysis for preparing both the thermal state and the ground state. Specifically, we consider three examples of physical systems: a single qubit example (as a toy model), free fermionic systems, and commuting local Hamiltonians. In all three cases, we show that the mixing time of $\Phi$ can be upper bounded by a constant independent of $\sigma$, provided that $\sigma$ is sufficiently large. This enables us to achieve an arbitrarily small fixed-point error by appropriately choosing a large $\sigma$ and a small $\alpha$. For clarity, we first state the results, and defer all proofs to later sections.

\subsection{Single qubit example}

 We first consider a toy model to illustrate the key ideas. Assume the system Hamiltonian $H=-Z$. In~(2), we set \rev{$\mc{A}=\{X,-X\}$} and \rev{$g(\omega)=\frac{1}{3}\textbf{1}_{[0,3]}(\omega)$ ($\omega_{\max}=3$)}. Then, we have the following result:

\begin{thm}\label{thm:toy_model}
For thermal state preparation, given any $\beta,\epsilon>0$, there exists a constant $C=\mathrm{poly}(\beta,1/\epsilon)$ such that if $\sigma>C$, $T=\widetilde{\Omega}(\sigma)$, and $\alpha<\sigma^{-1}C^{-1}$, we have
\[
t_{{\rm mix},\Phi}(\epsilon)=\mc{O}(\log(1/\epsilon)).
\]

For ground state preparation ($\beta=\infty$), given $\epsilon>0$, there exists a constant $C=\mathrm{polylog}(1/\epsilon)$ such that if $\sigma>C$, $T=\widetilde{\Omega}(\sigma)$, and $\alpha<\sigma^{-1}\epsilon^{1/2}C^{-1}$, we have
\[
t_{{\rm mix},\Phi}(\epsilon)=\mc{O}(\log(1/\epsilon)).
\]
\end{thm}
Although this is a toy model, it highlights a key mechanism underlying the efficiency of our protocol when $\sigma \gg 1$: the design of the jump operator $V_{A_S,f,T}(\omega)$ should support a wide range of nondegenerate energy transitions. In the present setting, it suffices to have nondegenerate jumps between $\ket{0}$ and $\ket{1}$; see \cref{eqn:L_thermal_toy,eqn:L_ground_toy}. Even for this simplified model, achieving this property requires a careful choice of both the function $f(t)$ and the bath. A contrasting example that fails to meet this condition is discussed in \cref{sec:toy_uniform}, \cref{rem:filter_contrast}.

The proof of~\cref{thm:toy_model} is given in \cref{sec:toy_uniform}. We emphasize that, in this theorem, when $\sigma$ is sufficiently large, the mixing time $t_{{\rm mix},\Phi}(\epsilon)$ becomes independent of $\sigma$. Plugging this bound into \cref{thm:fix_thermal} and \cref{thm:fix_ground} yields a result demonstrating the end-to-end efficiency of our protocol; see \cref{cor:complete}.

\subsection{Free fermionic systems}

Consider a local fermionic Hamiltonian $H$ defined on a $D$-dimensional lattice of fermionic systems, $\Lambda = [0,L]^D$, given by
\begin{equation}\label{eqn:H_fermion}
H =\sum^N_{i,j=1} h_{i,j} c^\dagger_i c_j\,.
\end{equation}
where $N=(L+1)^D$ is the number of fermionic modes, $(h_{i,j})$ is a Hermitian matrix, and $c^\dagger_j$ and $c_j$ are the creation and annihilation operators at site $j$. We also assume that the coefficient matrix $h$ satisfies $\norm{h}=\Or(1)$. Note that the operator norm of the Hamiltonian $\norm{H}$ can still increase with respect to the system size $N$. We choose $A_S$ to be uniformly sampled from the set of all single fermionic operators $\{\pm c^\dagger_i,\pm  c_i\}^n_{i=1}$.

The mixing time analysis for ground state preparation is simpler, so we present it first. The rigorous version of~\cref{thm:gs_mixing} appears in~\cref{sec:gs_mixing} as~\cref{thm:gs_mixing_rigor}.

\begin{thm}[Ground state of quadratic fermionic Hamiltonian, informal]\label{thm:gs_mixing} Assume $H$ has a spectral gap $\Delta$. Let \rev{$g(\omega)=\frac{1}{\omega_{\max}}\mathbf{1}_{[0,\omega_{\max}]}$} with $\omega_{\max}=2\|h\|$. Given any $\epsilon>0$, if  $\sigma=\widetilde{\Theta}(\Delta^{-1})$, $T=\widetilde{\Theta}(\Delta^{-1})$, and $\alpha=\widetilde{\mathcal{O}}(\epsilon^{1/2} \Delta N^{-1/2})$, we have
\[
t_{{\rm mix},\Phi}(\epsilon)=\mathcal{O}\left(N\log(N/\epsilon)\right)\,.
\]
Here, $\widetilde{\Theta}$ suppresses logarithmic dependencies on $\Delta^{-1}$, $1/\epsilon$, and $N$.
\end{thm}

To prove this result, we adopt the strategy from~\cite[Section IV]{zhan2025rapidquantumgroundstate}, which analyzes the Heisenberg evolution of the number operator. Following the argument in~\cite[Section IV]{zhan2025rapidquantumgroundstate}, the convergence of the Lindblad dynamics to the ground state can be established by showing that the expectation of the number operator converges to zero. Moreover, since the unitary evolution commutes with the number operator, it does not affect this convergence. Finally, the convergence of the number operator can be directly related to the trace distance between the current state and the ground state using the Fuchs--van de Graaf inequality; see~\cref{sec:gs_mixing} for details.

For thermal state preparation, we have an analogous result.

\begin{thm}[Thermal state of quadratic fermionic Hamiltonian at constant temperature, informal]\label{thm:thermal_mixing}
For any constant temperature $\beta^{-1}$, with a proper choice of $g(\omega)$, let $\sigma=\widetilde{\Theta}(\epsilon^{-1} N^2), T=\widetilde{\Theta}(\epsilon^{-1} N^2), \alpha=\widetilde{\Theta}(\epsilon^{3/2}N^{-3})$, we have
\[
t_{{\rm mix},\Phi}(\epsilon)=\mc{O}\left(N^2 \log(N/\epsilon)\right)\,.
\]
Here, the notation $\widetilde{\Theta}$ suppresses logarithmic dependencies on $1/\epsilon$, and $N$.
\end{thm}

The rigorous version of ~\cref{thm:thermal_mixing} is presented in~\cref{sec:therm_mixing_free_fermion} as~\cref{thm:thermal_mixing_free_fermion_rigo}.~\rev{Compared to the ground state result, the additional $N$ factor in $t_{\rm mix}$ mainly arises from the initial dependence of the norm $\|\rho_{\beta}^{-1/4}[\cdot]\rho_{\beta}^{-1/4}\|_2$; see the detailed discussion at the end of~\cref{sec:therm_mixing_free_fermion}.} It is worth noting that the choice of $g(\omega)$ in~\cref{thm:thermal_mixing_free_fermion_rigo} is chosen to simplify the analysis, and can be suboptimal at large $\beta$~\cite[Section VII]{tong2024fast}.

\rev{In Theorem~\ref{thm:thermal_mixing}, it may be possible to further reduce the dependence of $t_{\mathrm{mix},\Phi}$ to linear in $N$ by employing advanced mixing time analysis techniques, such as the modified logarithmic Sobolev inequality or the oscillator norm method~\cite{KastoryanoTemme2013,BardetCapelGaoEtAl2023,rouz2024,zhan2025rapidquantumgroundstate,kochanowski2024rapid,rouze2024optimal}. However, due to the additional analytical challenges introduced by the Lamb-shift term, pursuing this improvement lies beyond the current scope of this work.
On the other hand, we believe that the linear $N$ dependence of $t_{\mathrm{mix},\Phi}$ in Theorems~\ref{thm:gs_mixing} and~\ref{thm:thermal_mixing} is intrinsic, since the algorithm samples only one jump operator per iteration. This situation closely parallels that of Lindbladian-dynamics-based algorithms: while rapid mixing can, in principle, be achieved when employing $O(N)$ jump operators, the total end-to-end simulation cost still scales linearly with $N$~\cite{cleve_2017,li_et_2023,PRXQuantum.5.020332}.
}

\subsection{Commuting local Hamiltonians}

Let $H = \sum_i h_i$ be a commuting local Hamiltonian defined on a $D$-dimensional lattice, where each local term $h_i$ commutes with all others and is supported on a ball of constant radius. Furthermore, each qubit $j$ is acted upon by only a constant number of terms $h_i$. Let $I_j$ denote the set of indices $i$ such that $h_i$ acts non-trivially on qubit $j$, and define $H_j = \sum_{i \in I_j} h_i$. Let $\Delta_{\lambda} = \max_{j,k} \left( \lambda_{k+1}(H_j) - \lambda_k(H_j) \right)$ be the maximal nearby eigenvalue difference among all $H_j$. We note that for local commuting Hamiltonians, $\Delta_{\lambda}$ is often a constant independent of the system size.

We choose $A_S$ to be randomly sampled from all local Pauli operators $\{\pm X_i,\pm Y_i,\pm Z_i\}_{i=1}^n$. We have the following result:

\begin{thm}[Commuting local Hamiltonian at high temperature, informal]\label{thm:thermal_commuting_local} Let $H$ be a commuting local Hamiltonian defined on a $D$-dimensional lattice and $g(\omega)=\frac{1}{\omega_{\max}}\mathbf{1}_{[0,\omega_{\max}]}$ with $\omega_{\max}=2\Delta_{\lambda}$. There exists a constant $\beta_c$ dependent on the Hamiltonian $H$ such that for every $\beta\leq \beta_c$ and any $\epsilon>0$, if $\sigma=\widetilde{\Theta}(\epsilon^{-1}N^2),T=\widetilde{\Theta}(\epsilon^{-1}N^2),\alpha=\widetilde{\Theta}(\epsilon^{3/2}N^{-3})$, we have
\[
t_{{\rm mix},\Phi}(\epsilon)=\mc{O}\left(N^2\log(1/\epsilon)\right)\,.
\]
Here, $\widetilde{\Theta}$ suppresses logarithmic dependencies on $1/\epsilon$, and $N$.
\end{thm}

 A more general version of~\Cref{thm:thermal_commuting_local} is given in \cref{sec:thermal_mixing_general} \cref{thm:thermal_mixing_rigor}. Here we use the result of~\cite{kastoryano2016quantum} stating that for commuting local Hamiltonians, there exists a critical inverse temperature $\beta_c$ such that, when $\beta \le \beta_c$, the spectral gap of the Davies generator is bounded below; see~\cref{rem:local_commuting}.

Although the mixing time bounds in \cref{thm:thermal_mixing} and \cref{thm:thermal_commuting_local} appear similar, their proof strategies differ substantially.  For the thermal state case, the main idea is to show that the dissipative part of the Lindbladian approximately satisfies the detailed balance condition, while the Lamb shift term approximately commutes with the thermal state when $\sigma \gg 1$. However, the  Lindbladian with the Lamb shift term does not satisfy the quantum detailed balance condition. Therefore existing techniques using the contraction of $\chi^2$-distance, relative entropy~\cite{KastoryanoTemme2013}, or local oscillator norm~\cite{rouze2024optimal,zhan2025rapidquantumgroundstate} are not directly applicable.

Instead, we follow the approach of~\cite[Appendix E.3.a, Proposition II.2]{ChenKastoryanoBrandaoEtAl2023}, which analyzes the spectral gap of the dissipative part of the generator after a similarity transformation, as introduced in~\cite[Appendix E.2]{ChenKastoryanoBrandaoEtAl2023}. In particular, we prove contraction under the weighted Hilbert-Schmidt norm $\|\rho_{\beta}^{-1/4}[\cdot]\rho_{\beta}^{-1/4}\|_2$. This contraction still holds in the presence of the unitary evolution in (7), and therefore also holds for the map $\Phi$. Further details are given in \cref{sec:thermal_mixing}, in particular \cref{thm:mixing_general} and \cref{cor:mixing_phi_alpha}.


\subsection{End-to-end efficiency analysis}\label{sec:complete_efficiency}

In the previous section, we have established that the fixed-point approximation error and the upper bound on the mixing time are independent of $\sigma$, when $\sigma$ is sufficiently large. This property is crucial for ensuring the validity of the fixed-point error bound in~\cref{sec:thermal_ground_prep}.

Combining the result in~\cref{sec:thermal_ground_prep}, we obtain the following corollary:

\begin{cor}\label{cor:complete} For the single-qubit, free-fermion, and (high-temperature) local commuting Hamiltonian problems above, for any $\epsilon>0$, it suffices to choose $\sigma,T,\alpha^{-1}=\mathrm{poly}(N,1/\epsilon)$ to ensure that
\begin{align*}
&\|\rho_{\rm fix}(\Phi)-\rho_\beta\|_1<\epsilon,\\
&\tau_{{\rm mix},\Phi}(\epsilon)=\frac{t_{{\rm mix},\Phi}(\epsilon)}{\alpha^2}=\mathrm{poly}(N,1/\epsilon)\,.
\end{align*}
For the single qubit and gapped free fermionic systems above, for any $\epsilon>0$, it suffices to choose $\sigma,T,\alpha^{-1}=\mathrm{poly}(N,1/\epsilon)$ to ensure that
\begin{align*}
&\|\rho_{\rm fix}(\Phi)-\ket{\psi_0}\bra{\psi_0}\|_1 <\epsilon,\\
&\tau_{{\rm mix},\Phi}(\epsilon)=\frac{t_{{\rm mix},\Phi}(\epsilon)}{\alpha^2}=\mathrm{poly}(N,1/\epsilon)\,.
\end{align*}
\end{cor}
In the above corollary, $\tau_{{\rm mix},\Phi}(\epsilon)$ denotes the number of
times the map $\Phi$ defined in~(2) should be applied to achieve $\epsilon$-mixing.

To establish end-to-end efficiency, it remains to analyze the simulation complexity of $\Phi$, which follows from a standard analysis of Trotter errors (see e.g.~\cite{PhysRevX.11.011020}).
Recall the quantum channel $\Phi^{\rm approx}_{\alpha}$ in (4). We have $\|\Phi^{\rm approx}_{\alpha} - \Phi\|_{1\leftrightarrow 1} = \mathcal{O}\left(\alpha T (\|H\|+\omega_{\max})^2\|A_S\| \tau^2/\sigma^{1/2}\right)$, where $\tau$ is the Trotter step size. Since $\|A_S\|\leq 1$, to achieve $\eta$-accuracy in each application of $\Phi$, the number of Trotter steps per iteration is $M = \Theta\left(\alpha^{1/2}T^{3/2} (\|H\|+\omega_{\max})\, \eta^{-1/2}\sigma^{-1/4}\right)$.  Given a mixing time of $\tau_{{\rm mix},\Phi}(\epsilon)$, we set $\eta=\epsilon/\tau_{\text{mix},\Phi}$ to ensure the total quantum channel error is bounded by $\epsilon$ in $1\leftrightarrow 1$ norm. This leads to  the total number of steps is $M_{\rm total} = M \cdot \tau_{{\rm mix},\Phi}(\epsilon) = \mathrm{poly}(N, 1/\epsilon)$. Each step involves a short-time ($\tau$) simulation of the system Hamiltonian, a single $Z$ rotation, and one simulation step for the system-bath interaction term whose gate complexity depends on the choice of $A_S$.~\rev{We note that, in~\cref{cor:complete}, the dependence of $\tau_{\mathrm{mix},\Phi}$ and $M_{\mathrm{total}}$ on $N$, $\beta$, and $1/\epsilon$ could potentially be further improved, not only by establishing a tighter upper bound on the mixing time, but also by allowing more relaxed choices of $\alpha$ and $\sigma$. For instance, although this work focuses on the weak-interaction regime—where $\alpha\sqrt{\sigma}$ is small—there is currently no evidence that this is the only regime that is valid (see~\cite{ramonescandell2025thermalstatepreparationrepeated} for example). Relaxing this assumption represents an interesting direction for future research.
}


\section{Mixing analysis of thermal and ground state preparation for the single qubit example}\label{sec:toy_uniform}

In this section, we consider a toy model $H=-Z$. In~(2), we set \rev{$\mc{A}=\{X,-X\}$} and \rev{$g(\omega)=\frac{1}{3}\textbf{1}_{[0,3]}(\omega)$ ($\omega_{\max}=3$)}. To prove Theorem~\ref{thm:toy_model}, it suffices to show that, for both thermal state and ground state preparation, the mixing time of $\Phi$ is independent of $\sigma$ when $\sigma$ is sufficiently large.

Similar to~\cref{sec:app_fix_thermal} and~\cref{sec:app_fix_gs}, we first consider a simplified CPTP map defined as follows:
\[
\widetilde{\Phi}=\mc{U}_S(T)\,\circ\,\exp\left(\widetilde{\mathcal{L}}\alpha^2\right)\,\circ\,\mc{U}_S(T)\,.
\]
Here $\widetilde{\mc{L}}$ omits the error in Theorem 4 and take the limit $T \to \infty$.
Specifically,
\[
\begin{aligned}
\widetilde{\mc{L}}(\rho)=-i\left[\widetilde{H}_{\mathrm{LS}},\rho\right]+\int^\infty_{-\infty}\gamma(\omega)\mathcal{D}_{\widetilde{V}_{X,f}(\omega)}(\rho)\mathrm{d}\omega\,,
\end{aligned}
\]
where
\[
\widetilde{H}_{\mathrm{LS}}=\rev{-\rm{Im}}\left(\int^\infty_{-\infty}\gamma(\omega)\widetilde{\mc{G}}_{X,f}(-\omega)\mathrm{d}\omega\right),\quad \widetilde{V}_{X,f}(\omega)=\int^\infty_{-\infty}f(t)X(t)\exp(-i\omega t)\mathrm{d}t\,,
\]
with
\[
\widetilde{\mc{G}}_{\rev{X},f}(\omega)=\int^\infty_{-\infty}\int^{s_1}_{-\infty}f(s_2)f(s_1) \rev{X}(s_2)\rev{X}(s_1)\exp(i\omega(s_2-s_1))\mathrm{d}s_2\mathrm{d}s_1\,.
\]
According to Lemma~\ref{lem:thermal_approx_first} or Lemma~\ref{lem:gd_approx_first}, for thermal and ground state preparation, respectively, we first have
\[
\left\|\Phi-\widetilde{\Phi}\right\|_{1\leftrightarrow 1}=\mathcal{O}\left(\alpha^2\sigma\exp\left(-T^2/(4\sigma^2)\right)+\alpha^4\rev{T^4}\sigma^{-2}\right)\,.
\]
According to~\cref{thm:almost_fixed_point}, when $\alpha$ is sufficiently small and $T$ is sufficiently large, it suffices to consider the mixing time of $\widetilde{\Phi}$. In this case, we can compute $\widetilde{V}$ and $\widetilde{H}_{\rm LS}$ explicitly.
Noticing,
\rev{\[
\widetilde{V}_{X,f}(\omega) = \int_{-\infty}^{\infty} f(t) X(t) e^{-i\omega t} \, \mathrm{d}t=2^{3/4} \sigma^{1/2}\pi^{1/4}\left(\exp\left(-\sigma^2(\omega - 2)^2\right)\ket{1}\bra{0}+\exp\left(-\sigma^2(\omega + 2)^2\right)\ket{0}\bra{1}\right)\,,
\]}
and
\[
\begin{aligned}
\widetilde{\mc{G}}_{\rev{X},f}(\omega)&=\sum_{\nu_1,\nu_2\in B(H)}\rev{X}(\nu_2)\rev{X}(\nu_1)\int^{\infty}_{-\infty}\int^{s_1}_{-\infty}f(s_2)f(s_1)\exp(i\nu_2s_2)\exp(i\nu_1s_1)\exp(i\omega(s_2-s_1))\mathrm{d}u\mathrm{d}v\\
&=C_{0,\sigma}(\omega)\ket{0}\bra{0}+C_{1,\sigma}(\omega)\ket{1}\bra{1}
\end{aligned}
\]
where $C_{0,\sigma}(\omega)$ and $C_{1,\sigma}(\omega)$ are functions of $\omega$ that depend on $\sigma$. Because the Lamb shift term does not effect the proof later, we do not specify the form of $C_{0}$ and $C_{1}$.

Now, we consider the thermal state and ground state separately:
\begin{itemize}
\item Thermal state: We notice that
\[
\rev{
\begin{aligned}
&\int_{-\infty}^\infty \gamma(\omega)\, \mathcal{D}_{V_{X,f,T}(\omega)}(\rho)\, \mathrm{d}\omega \\
=&2^{3/2}\pi^{1/2}\int^\infty_{-\infty}\gamma(\omega)\sigma\exp(-2\sigma^2(\omega-2)^2)\mathrm{d}\omega \mathcal{D}_{\ket{1}\bra{0}}(\rho)+2^{3/2}\pi^{1/2}\int^\infty_{-\infty}\gamma(\omega)\sigma\exp(-2\sigma^2(\omega+2)^2)\mathrm{d}\omega \mathcal{D}_{\ket{0}\bra{1}}(\rho)+\mc{O}\left(\exp(-8\sigma^2)\right)\\
=&2\pi \left( \gamma(2)\mathcal{D}_{\ket{1}\bra{0}}(\rho) + \gamma(-2)\mathcal{D}_{\ket{0}\bra{1}}(\rho) \right)+\mc{O}\left(\frac{\beta}{\sigma}\right)\,,
\end{aligned}}
\]
and
\[
\widetilde{H}_{\mathrm{LS}}=C_{0,\beta,\sigma}\ket{0}\bra{0}+C_{1,\beta,\sigma}\ket{1}\bra{1}\,,
\]
where $C_{0,\beta,\sigma}$ and $C_{1,\beta,\sigma}$ are constants that depend on $\sigma$. Define
\begin{equation}\label{eqn:L_thermal_toy}
\widehat{\mc{L}}_{\beta}=-i\left[\widetilde{H}_{\mathrm{LS}},\rho\right]+\rev{2\pi} \left( \gamma(2)\mathcal{D}_{\ket{1}\bra{0}} + \gamma(-2)\mathcal{D}_{\ket{0}\bra{1}} \right)
\end{equation}
and $\widehat{\Phi}_{\beta}=\mc{U}_S(T)\,\circ\,\exp\left(\widehat{\mathcal{L}}_{\beta}\alpha^2\right)\,\circ\,\mc{U}_S(T)$. Then, we have
\begin{equation}\label{eqn:toy_error_bound_thermal}
\left\|\Phi-\widehat{\Phi}_{\beta}\right\|_{1\leftrightarrow 1}=\mathcal{O}\left(\alpha^2\sigma\exp\left(-T^2/(4\sigma^2)\right)+\alpha^4\rev{T^4}\sigma^{-2}+\frac{\beta}{\sigma}\right)\,.
\end{equation}

In the case when $\sigma$ is sufficiently large,~\rev{according to~\cref{thm:almost_fixed_point}}, we only need to consider the mixing time of $\widehat{\Phi}_{\beta}$. Since this result follows from a more general theorem in~\cref{thm:thermal_mixing_rigor}, it is sufficient to demonstrate that the mixing time is independent of $\sigma$.

We express $\rho_n$ in the computational basis as $\rho_n = \sum_{a,b=0}^1 c_{a,b,n} \ket{a}\bra{b}$, where $c_{a,b,n}$ are the corresponding coefficients. To show convergence, it suffices to verify that $c_{0,0,n}$ and $c_{1,1,n}$ converge to $\frac{\exp(2\beta)}{1+\exp(2\beta)}$ and $\frac{1}{1+\exp(2\beta)}$, respectively, while $|c_{0,1,n}|^2$ and $|c_{1,0,n}|^2$ converge to zero. It is straightforward to check that the Lamb shift Hamiltonian $\widetilde{H}_{\mathrm{LS}}$ and the unitary dynamics do not affect the evolution of $c_{0,0,n}$, $c_{1,1,n}$, or $|c_{0,1,n}|^2$, $|c_{1,0,n}|^2$. After plugging $\rho_n$ into~\eqref{eqn:L_thermal_toy}, an ordinary differential equation (ODE) is obtained for the evolution of $c_{0,0,n}$, $c_{1,1,n}$, $c_{0,1,n}$, and $c_{1,0,n}$. A direct calculation verifies that the solution converges to the desired fixed point. Since the evolution is independent of $\sigma$, it follows that the mixing time of $\widehat{\Phi}_{\beta}$ is also independent of $\sigma$.
Combing this mixing time and ~\eqref{eqn:toy_error_bound_thermal} with~\cref{thm:almost_fixed_point}, we conclude the proof for the thermal state part.

\item Ground state:
\[
\int_{-\infty}^\infty \gamma(\omega)\, \mathcal{D}_{V_{\rev{X},f,T}(\omega)}(\rho)\, \mathrm{d}\omega = C_{\infty,\sigma}\mathcal{D}_{\ket{0}\bra{1}}+\mc{O}\left(\sigma\exp(-\rev{4}\sigma^2)\right)\,,
\]
where $C_{\infty,\sigma}$ is a constant that depends on $\sigma$. We note that, there exists a uniform constant $C_{\infty}$ such that $C_{\infty,\sigma}\geq C_{\infty}$ for $\sigma\geq 1$.
\[
\widetilde{H}_{\mathrm{LS}}=C_{0,\infty,\sigma}\ket{0}\bra{0}+C_{1,\infty,\sigma}\ket{1}\bra{1}\,,
\]
where $C_{0,\infty,\sigma}$ and $C_{1,\infty,\sigma}$ \rev{that only} depends on $\sigma$.

Define
\begin{equation}\label{eqn:L_ground_toy}
\widehat{\mc{L}}_{\infty}=-i\left[\widetilde{H}_{\mathrm{LS}},\rho\right]+C_{\infty,\sigma}\mathcal{D}_{\ket{0}\bra{1}}
\end{equation}
and $\widehat{\Phi}_{\infty}=\mc{U}_S(T)\,\circ\,\exp\left(\widehat{\mathcal{L}}_{\infty}\alpha^2\right)\,\circ\,\mc{U}_S(T)$. Then, we have
\begin{equation}\label{eqn:toy_error_bound_ground}
\left\|\Phi-\widehat{\Phi}_{\infty}\right\|_{1\leftrightarrow 1}=\mathcal{O}\left(\alpha^2\sigma\exp\left(-T^2/(4\sigma^2)\right)+\alpha^4\rev{T^4}\sigma^{-2}+\sigma\exp(-\rev{4}\sigma^2)\right)\,.
\end{equation}

Finally, we consider the mixing time of $\widehat{\Phi}_{\infty}$. Similar to the thermal state case, we express $\rho_n$ in the computational basis as $\rho_n = \sum_{a,b=0}^1 c_{a,b,n} \ket{a}\bra{b}$, where $c_{a,b,n}$ are the corresponding coefficients. To show convergence, it suffices to verify that $c_{0,0,n}$ and $c_{1,1,n}$ converge to $1$ and $0$, respectively, while $|c_{0,1,n}|^2$ and $|c_{1,0,n}|^2$ converge to zero. Same as before, the Lamb shift Hamiltonian $\widetilde{H}_{\mathrm{LS}}$ and the unitary dynamics do not affect the evolution of $c_{0,0,n}$, $c_{1,1,n}$, or $|c_{0,1,n}|^2$, $|c_{1,0,n}|^2$. Similar to before, a direct calculation verifies that the dissipative part of $\widehat{\mathcal{L}}_{\infty}$ converges and is independent of $\sigma$, which implies that the mixing time of $\widehat{\Phi}_{\infty}$ is independent of $\sigma$. Combining this mixing time and~\eqref{eqn:toy_error_bound_ground} with~\cref{thm:almost_fixed_point}, we conclude the proof for the ground state part.

\end{itemize}

\begin{rem}\label{rem:filter_contrast}
Different from our setting,~\cite[Section III]{hahn2025provably} considers the filter function
\begin{equation}\label{eqn:Hahn_f}
f(t) = \sqrt{\frac{2}{\pi\sigma^2}}\, \exp\left(-\frac{2}{\sigma^2}\left(t - \frac{i\beta}{4}\right)^2\right)
\end{equation}
Under this choice, the corresponding jump operator in the Lindblad dynamics is
\[
L_{A_S} = \sum_{\omega \in B(H)} \exp(-\beta\omega/4)\, \exp\left(-\frac{(\sigma \omega)^2}{8}\right) A_S(\omega)\,.
\]

As $\sigma \to \infty$, the support of $L_{A_S}$ effectively shrinks to a narrow energy window of width $\mathcal{O}(1/\sigma)$:
\[
L_{A_S}=\underbrace{\sum_{|\omega| \leq 1/2} \exp(-\beta\omega/4)\, \exp\left(-\frac{(\sigma \omega)^2}{8}\right) A_S(\omega)}_{=A_S(0)}+\underbrace{\sum_{|\omega| \geq 1/2} \exp(-\beta\omega/4)\, \exp\left(-\frac{(\sigma \omega)^2}{8}\right) A_S(\omega)}_{=\mathcal{O}(\exp(-\sigma^2))}\,,
\]
where we use $B(H)=\{2,0,-2\}$ in the above equality. This implies that when $\sigma \gg 1$, transitions between eigenvectors corresponding to different eigenvalues are strongly suppressed. For instance, when $A_S=X$, we have $A_S(0)=0$ and $\|L_{A_S}\| = \mathcal{O}(\exp(-\sigma^2))$, so the dissipative term becomes exponentially weak. This leads to a mixing time scaling as $t_{\rm mix} = \Omega(\exp(\sigma^2))$. Substituting this into the fixed-point error bound in~\cref{thm:fix_thermal} yields a vacuous upper bound on the error.

\rev{In our algorithm, since the bath is initialized in the thermal state of $H_E$, we do not need to choose an interaction function $f$ that simultaneously depends on both $\beta$ and $\sigma$ to satisfy an approximate detailed balance condition. This avoids the restriction—present in~\cref{eqn:Hahn_f}—that energy transitions must remain near $0$ when $\sigma$ is large. This key distinction prevents the mixing time from degrading in the large-$\sigma$ regime and enables substantially faster mixing.
}
\end{rem}

\section{Rigorous version of\texorpdfstring{~\cref{thm:gs_mixing}}{Lg}}\label{sec:gs_mixing}
Recall $H=\sum^N_{j,k} h_{j,k}c^\dagger_jc_k$, where $c^\dagger_j$, $c_k$ are creation and annihilation operators, respectively. Because $h$ is a Hermitian matrix, there exists an unitary matrix $U$ such that $\Lambda=U^\dagger hU$ is diagonal. Specifically,
\[
H=\sum^N_{k=1} \lambda_k \left(\sum_j (U^\dagger)_{k,j}c_j\right)^\dagger\left(\sum_j (U^\dagger)_{k,j}c_j\right):=\sum^N_{k=1}\lambda_k b^\dagger_kb_k\,,
\]
where $b^\dagger_k=\left(\sum_j (U^\dagger)_{k,j}c_j\right)^\dagger$, $b_k=\sum_j (U^\dagger)_{k,j}c_j$ formulate a new set of creation and annihilation operators after the unitary transformation. Then, the spectral gap $\Delta=\min_{i}|\lambda_i|$.

Now, we are ready to introduce the rigorous version of~\cref{thm:gs_mixing}:
\begin{thm}\label{thm:gs_mixing_rigor}
Let \rev{$g(\omega)=\frac{1}{\omega_{\max}}\mathbf{1}_{[0,\omega_{\max}]}$} with $\omega_{\max}=2\|h\|$. Given any $\epsilon>0$, when
\[
\left(\frac{N\|h\|}{\alpha^2}\log(N/\epsilon)\right)\left(\alpha^2\sigma\exp\left(-T^2/(4\sigma^2)\right)+\alpha^2\sigma\sqrt{N}\exp(-\rev{\Delta^2\sigma^2})+\alpha^4T^4\sigma^{-2}\right)=\mc{O}(\epsilon)\,,
\]
we have
\[
t_{{\rm mix},\Phi}(\epsilon)=\mathcal{O}\left(\|h\|N\log(N/\epsilon)\right)\,.
\]
\end{thm}
\begin{proof}[Proof of~\cref{thm:gs_mixing_rigor}]
First, according to~\cref{thm:almost_fixed_point}~\cref{eqn:mixing_time_bound} and Lemma~\ref{lem:gd_approx_first}, it suffices to prove the mixing time of $~\widetilde{\Phi}$ defined in~\eqref{eqn:Phi_tilde}. Recall
\[
\widetilde{\Phi}=\mc{U}_S(T)\,\circ\,\exp\left(\widetilde{\mathcal{L}}\alpha^2\right)\,\circ\,\mc{U}_S(T)\,.
\]
Here
\[
\begin{aligned}
\widetilde{\mc{L}}(\rho)=\mathbb{E}_{A_S}\left(-i\left[\widetilde{H}_{\mathrm{LS},A_S},\rho\right]+\int^0_{-\infty}(g(\omega)+g(-\omega))\mathcal{D}_{\widetilde{V}_{A_S,f}(\omega)}(\rho)\mathrm{d}\omega\right)\,,
\end{aligned}
\]
where
\[
\widetilde{H}_{\mathrm{LS},A_S}=\rev{-}\mathrm{Im}\left(\int^0_{-\infty}g(\omega)\widetilde{\mc{G}}_{A^\dagger_S,f}(\omega)\mathrm{d}\omega+\int^\infty_{0}g(\omega)\widetilde{\mc{G}}_{A_S,f}(-\omega)\mathrm{d}\omega\right),\quad \widetilde{V}_{A_S,f}(\omega)=\int^\infty_{-\infty}f(t)A_S(t)\exp(-i\omega t)\mathrm{d}t\,,
\]
with
\[
\widetilde{\mc{G}}_{A_S,f}(\omega)=\int^\infty_{-\infty}\int^{s_1}_{-\infty}f(s_2)f(s_1) A^\dagger_S(s_2)A_S(s_1)\exp(i\omega(s_2-s_1))\mathrm{d}s_2\mathrm{d}s_1\,.
\]

Define the number operator:
\[
\hat{\mathsf{N}}=\sum_{\lambda_k>0}b^\dagger_kb_k+\sum_{\lambda_k<0}b_kb^\dagger_k\,.
\]
Let $H$ has eigendecomposition $\{(\lambda_i,\ket{\psi_i})\}^{d-1}_{i=0}$ with $\lambda_0\leq \lambda_1\leq \dots,\lambda_{d-1}$. We note that $1-\bra{\psi_0}\rho\ket{\psi_0}\leq \mathrm{Tr}(\rho \hat{\mathsf{N}})$. Using the Fuchs-van de Graaf inequality, we obtain
\[
\|\rho-\ket{\psi_0}\bra{\psi_0}\|_1\leq \rev{2\sqrt{\left(1-\bra{\psi_0}\rho\ket{\psi_0}\right)}}\leq \rev{2\sqrt{\mathrm{Tr}(\rho \hat{\mathsf{N}})}}.
\]
we can show the decaying of $\mathrm{Tr}\left(\widetilde{\Phi}[\rho]\hat{\mathsf{N}}\right)$, compared with $\mathrm{Tr}\left(\rho\hat{\mathsf{N}}\right)$. Furthermore, because $\hat{\mathsf{N}}$ commutes with $H$, the unitary evolution part $\mc{U}_S$ does not effect the expectation. Thus, if we can show
\[
\mathrm{Tr}\left(\exp\left(\widetilde{\mc{L}}\alpha^2\right)[\rho]\hat{\mathsf{N}}\right)\leq (1-\delta)\mathrm{Tr}(\rho\hat{\mathsf{N}})
\]
for some $0<\delta<1$ and any $\rho$, then we have
\[
\mathrm{Tr}\left(\widetilde{\Phi}^K_{\alpha}[\rho]\hat{\mathsf{N}}\right)=\mathrm{Tr}\left(\exp\left(\widetilde{\mc{L}}\alpha^2\right)\circ \widetilde{\Phi}^{K-1}_{\alpha}[\rho]\hat{\mathsf{N}}\right)\leq (1-\delta)\mathrm{Tr}\left(\widetilde{\Phi}^{K-1}_{\alpha}[\rho]\hat{\mathsf{N}}\right)\leq (1-\delta)^K\mathrm{Tr}(\rho\hat{\mathsf{N}})\,.
\]
for any $K>0$ and $\rho$. This implies the fast decaying of $\mathrm{Tr}(\rho_K\hat{\mathsf{N}})$.

However, because $\widetilde{\mc{L}}$ does not exactly preserve the ground state, it is difficult to directly show the exponential decay of $\mathrm{Tr}\left(\exp\left(\widetilde{\mc{L}}\alpha^2\right)[\rho]\hat{\mathsf{N}}\right)$ in the above form. Instead, we will construct a new Lindbladian operator $\widehat{\mc{L}}$ in the proof so that  such that $\left\|\widehat{\mc{L}} - \widetilde{\mc{L}}\right\|_{1\leftrightarrow 1}$ is bounded, $\widehat{\mc{L}}$ exactly fixes the ground state, and $\widehat{\mc{L}}$ satisfies a decay property, namely,
\begin{equation}\label{eqn:number_decay}
\mathrm{Tr}\left(\exp\left(\widehat{\mc{L}}\alpha^2\right)[\rho]\hat{\mathsf{N}}\right)\leq (1-\delta)\,\mathrm{Tr}(\rho\hat{\mathsf{N}}).
\end{equation}

Recall that $A_S$ is uniformly sampled from $\{\pm c^\dagger_k,\pm c_k\}^N_{k=1}$. We first deal with the dissipative part and define $\widehat{\mc{L}}$. We note that
\begin{equation}\label{eqn:b_j_evolution}
\exp(iH t)b_j\exp(-iHt)=\exp(-i\lambda_jt)b_j,\quad \exp(iH t)b^\dagger_j\exp(-iHt)=\exp(i\lambda_jt)b^\dagger_j\,.
\end{equation}
When $A_S=c^\dagger_k=\left(\sum_j U_{k,j}b_j\right)^\dagger$, we have
\[
A_S(t)=\exp(iH t)A_S\exp(-iHt)=\sum_j \overline{U_{k,j}}\exp(i\lambda_j t)b^\dagger_j\,.
\]
\rev{Because the integral in $\widetilde{V}$ is restricted to the regime $\omega\leq 0$, we have}
\[
\begin{aligned}
\widetilde{V}_k(\omega)&:=\widetilde{V}_{A_S,f}(\omega)=\int^\infty_{-\infty}f(t)A_S(t)\exp(-i\omega t)\mathrm{d}t=\sum_j \overline{U_{k,j}}\hat{f}(-\omega+\lambda_j)b^\dagger_j\\
&=\sum_{\lambda_j<0} \overline{U_{k,j}}\hat{f}(-\omega+\lambda_j)b^\dagger_j+\underbrace{\mc{O}\left(\sqrt{\sigma N}\exp(-\rev{\Delta^2\sigma^2})\right)}_{\rev{\text{contains the part with $\lambda_j\geq 0$}}}
\end{aligned}
\]
when $\omega\leq 0$. We note that the Lindbladian with jump operator $\widehat{V}_k=\sum_{\lambda_j<0} \overline{U_{k,j}}\hat{f}(-\omega+\lambda_j)b^\dagger_j$ preserves the ground state, since $\widehat{V}_k\ket{\psi_0}=0$. In addition, for \rev{$\omega\leq 0$}, we note
\begin{equation}\label{eqn:Vhat_tile_error}
\left\|\widetilde{V}_k(\omega)-\widehat{V}_k(\omega)\right\|=\mc{O}\left(\sqrt{\sigma N}\exp(-\rev{\Delta^2\sigma^2})\right)\,.
\end{equation}
Now, defining $\widehat{\mc{L}}$ with $\widehat{V}_k$ and the same Lamb shift term $\widetilde{H}_{\rm LS}$, we obtain
\begin{equation}\label{eqn:hat_L_approx}
\left\|\widehat{\mc{L}} - \widetilde{\mc{L}}\right\|_{1\leftrightarrow 1}=\rev{\mc{O}\left(\sup_{k}\left\|\widetilde{V}_k(\omega)-\widehat{V}_k(\omega)\right\|\left\|\widetilde{V}_k(\omega)\right\|\right)}=\mc{O}\left(\sigma\sqrt{N}\exp(-\rev{\Delta^2\sigma^2})\right)
\end{equation}
Define $\widehat{\Phi}$ with $\widehat{\mc{L}}$ similar to~\cref{eqn:Phi_tilde}. We have
\begin{equation}\label{eqn:tile_hat_diff}
\left\|\widehat{\Phi}-\widetilde{\Phi}\right\|_{1\leftrightarrow 1}=\mc{O}\left(\alpha^2\sigma\sqrt{N}\exp(-\rev{\Delta^2\sigma^2})\right)\,.
\end{equation}
Combining~\eqref{eqn:tile_hat_diff} and Lemma~\ref{lem:gd_approx_first}, we have
\begin{equation}\label{eqn:tile_origi_diff}
\left\|\Phi-\widetilde{\Phi}\right\|_{1\leftrightarrow 1}=\mc{O}\left(\alpha^2\sigma\exp\left(-T^2/(4\sigma^2)\right)+\alpha^2\sigma\sqrt{N}\exp(-\rev{\Delta^2\sigma^2})+\alpha^4T^4\sigma^{-2}\right)\,.
\end{equation}

Now, given an observable $O=b^\dagger_ib_i$ with $\lambda_i>0$, we notice $[b_j,O]=\delta_{ij}b_i,\quad [b^\dagger_j,O]=-\delta_{ij}b^\dagger_j$. Then, we have
\[
\mathcal{L}^\dagger_{\widehat{V}_k}(O)=\frac{1}{2}\left([\widehat{V}^\dagger_k,O]\widehat{V}_k-\widehat{V}^\dagger_k[\widehat{V}_k,O]\right)=0
\]

Given an observable $O=b_ib^\dagger_i$ with $\lambda_i<0$, we notice $[b_j,O]=-\delta_{ij}b_i,\quad [b^\dagger_j,O]=\delta_{ij}b^\dagger_i$. Then, we have
\[
\begin{aligned}
&\mathcal{L}^\dagger_{\widehat{V}_k}(O)=\frac{1}{2}\left([\widehat{V}^\dagger_k,O]\widehat{V}_k-\widehat{V}^\dagger_k[\widehat{V}_k,O]\right)\\
=&\frac{1}{2}\left(\left(-U_{k,i}\hat{f}(\rev{-}\omega+\lambda_i)b_i\right)\widehat{V}_k-\widehat{V}_k^\dagger\left(\overline{U_{k,i}}\hat{f}(\rev{-}\omega+\lambda_i)b^\dagger_i\right)\right)\\
=&-\frac{1}{2}\sum_{\lambda_j<0}U_{k,i}\overline{U_{k,j}}\hat{f}(-\omega+\lambda_i)\hat{f}(-\omega+\lambda_j)b_ib^\dagger_j-\frac{1}{2}\sum_{\lambda_j<0}\overline{U_{k,i}}U_{k,j}\hat{f}(\rev{-}\omega+\lambda_i)\hat{f}(\rev{-}\omega+\lambda_j)b_jb^\dagger_i
\end{aligned}\,.
\]
Because $\sum_k U_{k,i}\overline{U_{k,j}}=\delta_{i,j}$, this implies
\begin{equation}\label{eqn:summation_annihilation}
\sum_k\left(\sum_{\lambda_i>0}\mathcal{L}^\dagger_{\widehat{V}_k}(b^\dagger_ib_i)+\sum_{\lambda_i<0}\mathcal{L}^\dagger_{\widehat{V}_k}(b_ib^\dagger_i)\right)=-\sum_{\lambda_i<0}\left|\hat{f}(-\omega+\lambda_i)\right|^2b_ib^\dagger_i
\end{equation}
Similarly, when $A_S=c_k$, we can also define $\hat{V}_k$ that preserves the ground state and satisfies~\eqref{eqn:Vhat_tile_error} to~\eqref{eqn:tile_origi_diff}. Further more, similar to~\eqref{eqn:summation_annihilation}, we have
\[
\sum_k\left(\sum_{\lambda_i>0}\mathcal{L}^\dagger_{\widehat{V}_k}(b^\dagger_ib_i)+\sum_{\lambda_i<0}\mathcal{L}^\dagger_{\widehat{V}_k}(b_ib^\dagger_i)\right)=-\sum_{\lambda_i>0}\left|\hat{f}(-\omega-\lambda_i)\right|^2b^\dagger_ib_i
\]
Because \rev{$g(\omega)=\frac{1}{2\|h\|}\mathbf{1}_{[0,2\|h\|]}$}, we have $\mathbb{E}_{\omega}|\hat{f}(-\omega-\mathrm{sign}(\lambda_i)\lambda_i)|^2=\Omega(\|h\|^{-1})$. Thus,
\begin{equation}\label{eqn:hat_L_decay}
\widehat{\mathcal{L}}^\dagger(\hat{\mathsf{N}})\leq -\frac{C}{\|h\|N}\hat{\mathsf{N}}\,,
\end{equation}
with a uniform constant $C$.~\rev{Here, $N$ comes from the expectation of $V_k$, which gives an $\frac{1}{N}$ factor before the summation of $k$.}

Next, for the Lamb shift term, when $A_S=c^\dagger_k=\left(\sum_j U_{k,j}b_j\right)^\dagger$, we have
\[
\begin{aligned}
&\widetilde{\mc{G}}_{A_S,f}(-\omega)=\int^\infty_{-\infty}\int^{s_1}_{-\infty}f(s_2)f(s_1) A^\dagger_S(s_2)A_S(s_1)\exp(-i\omega(s_2-s_1))\mathrm{d}s_2\mathrm{d}s_1\\
=&\sum^N_{\nu_1,\nu_2=1}U_{k,\nu_2}\overline{U_{k,\nu_1}}b_{\nu_2}b^\dagger_{\nu_1}\int^{\infty}_{-\infty}\int^{s_1}_{-\infty}f(s_2)f(s_1)\exp(-i\lambda_{\nu_2}s_2)\exp(i\lambda_{\nu_1}s_1)\exp(-i\omega(s_2-s_1))\mathrm{d}u\mathrm{d}v
\end{aligned}\,.
\]
After summing in $k$, the remaining terms commute with $\widehat{\mathsf{N}}$ and thus, does not change $\mathrm{Tr}(\rho(t)\widehat{\mathsf{N}})$. Similarly, when $A_S$ is chosen to be $c_k$, we have the same commuting properties.

In conclusion, using~\cref{eqn:hat_L_decay} and the commuting property of $\widetilde{\mc{G}}_{A_S,f}$, we have
\[
\mathrm{Tr}\left(\exp\left(\widehat{\mc{L}}\alpha^2\right)[\rho]\hat{\mathsf{N}}\right)\leq \left(1-\frac{C\alpha^2}{\|h\|N}\right)\mathrm{Tr}(\rho\hat{\mathsf{N}}).
\]
This implies that
\[
\left\|\widehat{\Phi}^k_{\alpha}[\rho]-\ket{\psi_0}\bra{\psi_0}\right\|_1\le2\sqrt{\operatorname{Tr}(\widehat\Phi_\alpha^k(\rho)\widehat{ \mathsf{N}})}\le2\sqrt N\left(1-C\alpha^2/(\|h\|N)\right)^{k/2}\,.
\]
Thus, given $\epsilon>0$, we have
\[
\tau_{{\rm mix},\widehat{\Phi}}(\epsilon)=\mc{O}\left(\frac{N\|h\|}{\alpha^2}\log(N/\epsilon)\right)
\]
Combining this,~\eqref{eqn:tile_origi_diff}, and~\cref{thm:almost_fixed_point}, we conclude the proof.
\end{proof}

\section{Mixing time of \texorpdfstring{$\Phi$}{Lg} for  thermal state preparation}\label{sec:thermal_mixing}

Before showing~\cref{thm:thermal_mixing} and~\cref{thm:thermal_commuting_local}, we provide a framework for studying the mixing time of the CPTP maps that take the form of
\[
\Phi=\mc{U}_S(T)\,\circ\,\exp\left(\mathcal{M}\alpha^2\right)\,\circ\,\mc{U}_S(T)\,,
\]
where $\mc{M}$ is an arbitrary Lindbladian that preserves the thermal state $\rho_\beta$. This framework is inspired by~\cite{ChenKastoryanoBrandaoEtAl2023,ChenKastoryanoGilyen2023}.

To start, we first introduce the detailed balance condition that allows coherent term:

\begin{defn}[Detailed balance condition with unitary drift~\cite{ChenKastoryanoGilyen2023,Fag_2007}]\label{def:DBC}
For any Lindbladian $\mathcal{M}$ and full-rank state $\rho_\beta$, take a similarity transformation and decompose into the Hermitian and the anti-Hermitian parts
\[
\begin{aligned}
\mathcal{K}(\rho_\beta, \mathcal{M})=\rho_{\beta}^{-1 / 4} \mathcal{M}\left[\rho_{\beta}^{1 / 4} \cdot \rho_{\beta}^{1 / 4}\right] \rho_{\beta}^{-1 / 4} & =\mathcal{H}(\rho_{\beta}, \mathcal{M})+\mathcal{A}(\rho_{\beta}, \mathcal{M}) \\
\mathcal{K}(\rho_\beta, \mathcal{M})^{\dagger}=\rho_{\beta}^{1 / 4} \mathcal{M}^{\dagger}\left[\rho_{\beta}^{-1 / 4} \cdot \rho_{\beta}^{-1 / 4}\right] \rho_{\beta}^{1 / 4} & =\mathcal{H}(\rho_{\beta}, \mathcal{M})-\mathcal{A}(\rho_{\beta}, \mathcal{M})
\end{aligned}
\]
We say the Lindbladian $\mathcal{M}$ satisfies the detailed balance with unitary drift if there exists a Hermitian operator $H_C$ such that
\[
\mathcal{A}(\rho_{\beta}, \mathcal{M})=-i\rho_\beta^{1/4}[H_C,\rho_\beta^{-1/4}(\cdot)\rho_\beta^{-1/4}]\rho_\beta^{1/4}\,.
\]
\end{defn}

We note that the above detailed balance condition allows a coherent term that commutes with $\rho_\beta$ in $\mc{M}$. It is straightforward to check that if $\mc{M}$ satisfies the detailed balance with unitary drift, then $\mc{H}(\rho_\beta,\mc{M})(\sqrt{\rho_\beta})=0$ and $\mc{M}(\rho_\beta)=0$. Furthermore, if $\mc{M}$ approximately satisfies the detailed balance with unitary drift, we have the following result to quantify the mixing time of $\Phi$:
\begin{thm}\label{thm:mixing_general} Assume $\rev{\mathcal{H}}(\rho_\beta, \mathcal{M})=\mathcal{H}_1(\rho_{\beta}, \mathcal{M})+\mathcal{H}_2(\rho_{\beta}, \mathcal{M})$. If \rev{$\mathcal{H}_1$ is a self-adjoint operator under Hilbert-Schmidt such that} $\mathcal{H}_1(\rho_{\beta}, \mathcal{M})(\sqrt{\rho_\beta})=0$ and $\mathcal{H}_1(\rho_{\beta}, \mathcal{M})$ has a spectral gap $\lambda_{\rm gap}(\mc{H}_1)>\left\|\mc{H}_2\right\|_{2\leftrightarrow2}$. Given any $\rho_1,\rho_2$, we have
\[
\left\|\Phi^k(\rho_1-\rho_2)\right\|_1\leq 2\exp\left(\left(-\lambda_{\rm gap}(\mc{H}_1)+\left\|\mc{H}_2\right\|_{2\leftrightarrow2}\right)k\alpha^2\right)\left\|\rho^{-1/2}_\beta\right\|\left\|\rho_1-\rho_2\right\|_1\,.
\]
Specifically, for any $\epsilon>0$, we have 
\[
t_{{\rm mix},\Phi}(\epsilon)\leq \frac{1}{\lambda_{\rm gap}(\mc{H}_1)-\left\|\mc{H}_2\right\|_{2\leftrightarrow2}}\log\left(\frac{2\left\|\rho^{-1/2}_\beta\right\|}{\epsilon}\right)+1
\]

\end{thm}
We emphasize that~\cref{thm:mixing_general} does not guarantee the correctness of the fixed point. However, it still provides an upper bound on the mixing time of $\Phi$. In the regime where $\mathcal{H}_2 \ll 1$, it is possible to establish a small fixed-point error.

\begin{proof}[Proof of~\cref{thm:mixing_general}] Given any density operator $\rho_1,\rho_2$, we define $\mathcal{E}=\rho_1-\rho_2$. We consider the change of $\left\|\rho_\beta^{-1/4}\mathcal{E}\rho_\beta^{-1/4}\right\|_2$ after applying $\Phi$, where $\|\cdot\|_2$ is the Schatten-2 norm (Hilbert-Schmidt norm). First, because $\mc{U}_S$ commutes with $\rev{\rho}^{-1/4}_\beta(\cdot)\rev{\rho}^{-1/4}_\beta$, we have
\[
\left\|\rho_\beta^{-1/4}\mc{U}_S(\mathcal{E})\rho_\beta^{-1/4}\right\|_2=\left\|\mc{U}_S\left(\rho_\beta^{-1/4}\mathcal{E}\rho_\beta^{-1/4}\right)\right\|_2=\left\|\rho_\beta^{-1/4}\mathcal{E}\rho_\beta^{-1/4}\right\|_2
\]
Thus,
\[
\left\|\rho_\beta^{-1/4}\Phi(\mathcal{E})\rho_\beta^{-1/4}\right\|_2=\left\|\rho_\beta^{-1/4}\exp(\mc{M}\alpha^2)(\mathcal{E})\rho_\beta^{-1/4}\right\|_2=\left\|\exp(\mc{K}(\rho_\beta,\mc{M})\alpha^2)\left[\rho^{-1/4}_\beta \mathcal{E}\rho^{-1/4}_\beta\right]\right\|_2
\]
Let $\mathcal{E}(t)=\exp(\mc{M}t)\mathcal{E}$. Because $\mathcal{E}(t)$ is traceless, we have $\rho^{-1/4}_\beta \mathcal{E}(t)\rho^{-1/4}_\beta$ is orthogonal to $\sqrt{\rho_\beta}$ under Hilbert Schemitz inner product. This implies that
\[
\begin{aligned}
&\frac{\mathrm{d}}{\mathrm{d}t}\left\|\rho_\beta^{-1/4}\exp(\mc{M}t)(\mathcal{E})\rho_\beta^{-1/4}\right\|^2_2=\frac{\mathrm{d}}{\mathrm{d}t}\left\|\exp(\mc{K}(\rho_\beta,\mc{M})t)\left[\rho^{-1/4}_\beta \mathcal{E}\rho^{-1/4}_\beta\right]\right\|_2^2\\
=&2\left\langle\rho^{-1/4}_\beta \mathcal{E}\rho^{-1/4}_\beta,(\mc{H}_1+\mc{H}_2)\left[\rho^{-1/4}_\beta \mathcal{E}(t)\rho^{-1/4}_\beta\right]\right\rangle_2\leq 2\left(-\lambda_{\rm gap}(\mc{H}_1)+\left\|\mc{H}_2\right\|_{2\leftrightarrow2}\right)\left\|\rho^{-1/4}_\beta \mathcal{E}(t)\rho^{-1/4}_\beta\right\|^2_2
\end{aligned}
\]
This implies that
\[
\left\|\rho_\beta^{-1/4}\Phi(\mathcal{E})\rho_\beta^{-1/4}\right\|_2=\left\|\exp(\mc{K}(\rho_\beta,\mc{M})\alpha^2)\left[\rho^{-1/4}_\beta \mathcal{E}\rho^{-1/4}_\beta\right]\right\|_2\leq \exp\left(\left(-\lambda_{\rm gap}(\mc{H}_1)+\left\|\mc{H}_2\right\|_{2\leftrightarrow2}\right)\alpha^2\right)\left\|\rho_\beta^{-1/4}\mathcal{E}\rho_\beta^{-1/4}\right\|_2\,.
\]
In summary, we have
\[
\left\|\rho_\beta^{-1/4}\Phi^k(\mathcal{E})\rho_\beta^{-1/4}\right\|_2\leq \exp\left(\left(-\lambda_{\rm gap}(\mc{H}_1)+\left\|\mc{H}_2\right\|_{2\leftrightarrow2}\right)k\alpha^2\right)\left\|\rho_\beta^{-1/4}\mathcal{E}\rho_\beta^{-1/4}\right\|_2\,.
\]
Finally, using $\|BAB\|_1\leq \|B\|^2_4\|A\|_2$, we have
\[
\|\mathcal{E}\|_1\leq \left\|\rho^{1/4}_\beta\right\|^2_4\left\|\rho_\beta^{-1/4}\mathcal{E}\rho_\beta^{-1/4}\right\|_2=\left\|\rho_\beta^{-1/4}\mathcal{E}\rho_\beta^{-1/4}\right\|_2\leq \left\|\rho^{-1/4}_\beta\right\|^2\|\mathcal{E}\|_2\leq \left\|\rho^{-1/2}_\beta\right\|\|\mathcal{E}\|_1\,.
\]
This implies
\[
\left\|\Phi^k(\mathcal{E})\right\|_1\leq 2\exp\left(\left(-\lambda_{\rm gap}(\mc{H}_1)+\left\|\mc{H}_2\right\|_{2\leftrightarrow2}\right)k\alpha^2\right)\|\rho^{-1/2}_\beta\|\left\|\mc{E}\right\|_1
\,.
\]
This concludes the proof.
\end{proof}

Let $H$ has eigendecomposition $\{(\lambda_i,\ket{\psi_i})\}^{2^N-1}_{i=0}$ with $\lambda_0\leq \lambda_1\leq \dots,\lambda_{2^N-1}$. Given a coupling operator $A$, for any $\omega>0$, define $A(\omega)=\sum_{\lambda_i-\lambda_j=\omega} \ket{\psi_i}\bra{\psi_j}\bra{\psi_i}A\ket{\psi_j}$. The Davies generator of a set of coupling operator $\mc{A}$ is defined as
\[
\mc{L}_{D,\mc{A}}[\rho]=\sum_{A\in\mc{A}}\sum_{\omega}A(\omega)\rho A(\omega)^\dagger-\frac{1}{2}\left\{A^\dagger(\omega)A(\omega),\rho\right\}\,.
\]

A direct corollary of~\cref{thm:mixing_general} is in the following:

\begin{cor}\label{cor:mixing_phi_alpha}
For $\mathcal{M}=-i[H_C,\cdot]+\mathcal{L}_D(\cdot)$, where $\mathcal{L}_D$ is a generator that satisfies GNS detailed balance condition or KMS detailed balance condition and has a gap $\lambda_{\rm gap}(\mc{L}_D)$. If
\[
\left\|\rho^{-1/4}_{\beta}H_C \rho^{1/4}_{\beta}-\rho^{1/4}_{\beta}H_C \rho^{-1/4}_{\beta}\right\|\leq \delta<\lambda_{\rm gap}(\mc{L}_D)\,.
\]
Then, the mixing time of $\mathcal{M}$ is
\[
t_{{\rm mix},\Phi}(\epsilon)\leq \frac{1}{\lambda_{\rm gap}(\mc{L}_D)-\delta}\log\left(\frac{2\left\|\rho^{-1/2}_\beta\right\|}{\epsilon}\right)+1
\]
\end{cor}
\begin{proof}[Proof of Corollary~\ref{cor:mixing_phi_alpha}] Because $\mc{L}$ satisfies GNS/KMS detailed balance condition, we have
\[
\mathcal{K}(\rho_\beta, \mathcal{L}_D)=\mathcal{K}(\rho_\beta, \mathcal{L}_D)^\dagger=\mc{H}(\rho_\beta, \mathcal{L}_D),\quad \rev{\lambda_{\rm gap}(\mc{H}(\rho_\beta, \mathcal{L}_D))=\lambda_{\rm gap}(\mc{L}_D)}\,.
\]
Thus,
\[
\mc{H}(\rho_\beta, \mathcal{M})=\mc{H}(\rho_\beta, \mc{L}_D)-\frac{i}{2}\left\{\rho^{-1/4}_{\beta}H_C \rho^{1/4}_{\beta}-\rho^{1/4}_{\beta}H_C \rho^{-1/4}_{\beta},\rho\right\}\,.
\]
Noticing
\[
\left\|\frac{i}{2}\left\{\rho^{-1/4}_{\beta}H_C \rho^{1/4}_{\beta}-\rho^{1/4}_{\beta}H_C \rho^{-1/4}_{\beta},[\cdot]\right\}\right\|_{2\leftrightarrow 2}\leq \left\|\rho^{-1/4}_{\beta}H_C \rho^{1/4}_{\beta}-\rho^{1/4}_{\beta}H_C \rho^{-1/4}_{\beta}\right\|\leq \delta\,,
\]
we conclude the proof using~\cref{thm:mixing_general}~\rev{with $\mc{H}_1=\mc{H}(\rho_\beta, \mc{L}_D)$.}
\end{proof}

Next, we show that, with a proper choice of $g(\omega)$,  the dissipative part of (8) approximates a Lindbladian dynamics satisfying the KMS detailed balance condition when $\sigma$ and $T$ are sufficiently large. This can be used to show the mixing time of the free fermions in \cref{sec:therm_mixing_free_fermion}.

\begin{thm}\label{thm:KMS_general} Given  \rev{$x=\Omega(\frac{\beta}{\sigma^2})$} such that \rev{$\frac{\beta^{2}}{\sigma^2}\frac{1+x/\sqrt{2x/\beta-1/(4\sigma^2)}}{x-\beta/(8\sigma^2)}=\mathcal{O}(1)$}, we set
\[
g_x(\omega)=\frac{1}{Z_x}\exp\left(-\frac{(\omega+x)^2}{2\left(\frac{2x}{\beta}-\frac{1}{4\sigma^2}\right)}\right),\quad Z_x=\sqrt{2\pi\left(\frac{2x}{\beta}-\frac{1}{4\sigma^2}\right)}\,.
\]
Then, there exists a Lindbladian $\widehat{\mc{L}}_{\text{KMS},x}$ that satisfies KMS detailed balance condition and a Hermitian operator $H_x$ such that
\[
\left\|\mc{L}-\left(-i[H_x,\rho]+\widehat{\mc{L}}_{\text{KMS},x}\right)\right\|_{1\leftrightarrow1}=\mc{O}\left(\sigma\exp\left(-T^2/(4\sigma^2)\right)+\rev{\frac{1}{Z_x}\left(\frac{\beta^{2}}{\sigma^2}\frac{x+\sqrt{2x/\beta-1/(4\sigma^2)}}{x-\beta/(8\sigma^2)}+\frac{\beta}{\sigma}\right)}\right)\,,
\]
and
\[
\left\|\sigma^{-1/4}_{\beta}H_x \sigma^{1/4}_{\beta}-\sigma^{1/4}_{\beta}H_x \sigma^{-1/4}_{\beta}\right\|=\mc{O}\left(\frac{\beta}{\sigma}\sqrt{\frac{\beta}{x-\frac{\beta}{8\sigma^2}}}+\rev{\frac{\beta^{3}}{\sigma^2}\frac{x+\sqrt{2x/\beta-1/(4\sigma^2)}}{x-\beta/(8\sigma^2)}}\right)\,.
\]
Here, $\widehat{\mc{L}}_{\text{KMS},x}$ takes the form of
\begin{equation}\label{eqn:KMS_Lindbladian}
\widehat{\mc{L}}_{\text{KMS},x}[\rho]=\mathbb{E}_{A_S}\left(-i\left[\frac{B_{A_S}}{Z_x},\rho\right]+\int^\infty_{-\infty}\widehat{\gamma}_x(\omega)\mathcal{D}_{V_{A_S,f,\infty}(\omega)}(\rho)\mathrm{d}\omega\right)\,,
\end{equation}
with $\widehat{\gamma}_x(\omega)=g_x(\omega)$ and
\[
B_{A_S}=-\int^\infty_{-\infty}h_1(t_1)e^{-iHt_1}\left(\int^\infty_{-\infty}h_2(t_2)A_S(t_2)A_S(-t_2)\mathrm{d}t_2\right)e^{iHt_1}\mathrm{d}t_1\,,
\]
where
\[
h_1(t)=\frac{1}{2\sigma \pi \beta}\exp\left(\frac{\beta^2}{32\sigma^2}\right)\left(\frac{1}{\cosh{2\pi t/\beta}}*_t\sin(-\beta t/
(4\sigma^2))\exp\left(-t^2/(2\sigma^2)\right)\right)\,,
\]
and
\[
h_2(t)=\rev{2}\sqrt{\frac{2x}{\beta}-\frac{1}{4\sigma^2}}\exp\left(\left(-\frac{4t^2}{\beta}-2it\right)x\right)\,.
\]
Furthermore, when $x=\frac{\beta}{8\sigma^2}+\omega_{\max}$ with $\omega_{\max}=\Omega(\beta)$, we have error bounds
\[
\left\|\mc{L}-\left(-i[H_x,\rho]+\widehat{\mc{L}}_{\text{KMS},x}\right)\right\|_{1\leftrightarrow1}=\mc{O}\left(\sigma\exp\left(-T^2/(4\sigma^2)\right)+\frac{\beta}{\sigma}\right)\,,
\]
and
\[
\left\|\sigma^{-1/4}_{\beta}H_x \sigma^{1/4}_{\beta}-\sigma^{1/4}_{\beta}H_x \sigma^{-1/4}_{\beta}\right\|=\mc{O}\left(\frac{\beta}{\sigma}\right)\,.
\]
\end{thm}
\begin{rem} We notice that $\widehat{\gamma}_x$ satisfies~\cite[Eqn. (1.4)]{ChenKastoryanoGilyen2023} up to a normalization factor. According to~\cite[Lemma II.2]{ChenKastoryanoGilyen2023} with $\sigma_E=1/(2\sigma)$, $\sigma_r=\sqrt{2x/\beta-\sigma^2_E}$, and $\omega_r=x$, and the above choice of $g_x$, the transition part of the Lindbladian $\widehat{\mc{L}}_{\text{KMS},x}$ can be written as
\[
\mathcal{T}(\rho)=\sum_{\nu_1,\nu_2\in B(H)}\gamma_{\nu_1,\nu_2}A_{\nu_1}\rho A_{\nu_2}^\dagger\,,
\]
where
\[
\gamma_{\nu_1,\nu_2}=\frac{1}{2\sqrt{4\pi x/\beta}}\exp\left(-\frac{(\nu_1+\nu_2+2x)^2}{16x/\beta}\right)\exp\left(-\frac{(\nu_1-\nu_2)^2\sigma^2}{2}\right)\,.
\]
This implies that, when $\sigma$ changes, it only affects the $\nu_1 - \nu_2$ term.
\end{rem}
\begin{proof}[Proof of~\cref{thm:KMS_general}] The formula of $g$ gives
\[
\gamma_x(\omega)=\frac{g_x(\omega)+g_x(-\omega)}{1+\exp(\beta\omega)}=\rev{g_x}(\omega)\frac{1+g_x(-\omega)/g_x(\omega)}{1+\exp(\beta\omega)}=\rev{g_x}(\omega)\frac{1+\exp(\beta\omega)\exp\left(\frac{\beta^2\omega}{8x\sigma^2-\beta}\right)}{1+\exp(\beta\omega)}\,.
\]
Let $\widehat{\gamma}_x(\omega)=g_x(\omega)$, we then have
\rev{\begin{equation}\label{eqn:gamma_diff}
\begin{aligned}
 &\left\|\gamma_x-\hat{\gamma}_x\right\|_{L^\infty}\leq  \left\|g_x(\omega)\left|1-\exp\left(\frac{\beta^2\omega}{8x\sigma^2-\beta}\right)\right|\right\|_{L^\infty_\omega}\\
 =&\left\|\frac{1}{Z_x}\exp\left(-\frac{(\omega+x)^2}{2\left(\frac{2x}{\beta}-\frac{1}{4\sigma^2}\right)}\right)\left|1-\exp\left(\frac{\beta^2\omega}{8x\sigma^2-\beta}\right)\right|\right\|_{L^\infty_\omega}\\
 =&\left\|\frac{1}{Z_x}\exp\left(-\frac{u^2}{2}\right)\left|1-\exp\left(\frac{\beta^2\left(u\sqrt{2x/\beta-1/(4\sigma^2)}-x\right)}{8x\sigma^2-\beta}\right)\right|\right\|_{L^\infty_u}\\
 =&\mc{O}\left(\frac{\beta^{2}}{\sigma^2}\frac{1}{\sqrt{2x/\beta-1/(4\sigma^2)}}\frac{\sqrt{2x/\beta-1/(4\sigma^2)}+x}{x-\beta/(8\sigma^2)}\right)\\
 =&\mc{O}\left(\frac{\beta^{2}}{\sigma^2}\frac{1+x/\sqrt{2x/\beta-1/(4\sigma^2)}}{x-\beta/(8\sigma^2)}\right)\,.
\end{aligned}
\end{equation}
when $\frac{\beta^{2}}{\sigma^2}\frac{1+x/\sqrt{2x/\beta-1/(4\sigma^2)}}{x-\beta/(8\sigma^2)}=\mathcal{O}(1)$. In the second equality, we let $u=\omega+x/\sqrt{2x/\beta-1/(4\sigma^2)}$. Similarly,
\begin{equation}\label{eqn:gamma_diff_L_1}
\begin{aligned}
 &\left\|\gamma_x-\hat{\gamma}_x\right\|_{L^1}\leq \int^\infty_{-\infty}g_x(\omega)\frac{\exp(\beta \omega)}{1+\exp(\beta\omega)}\left|1-\exp\left(\frac{\beta^2\omega}{8x\sigma^2-\beta}\right)\right|\mathrm{d}\omega\\
 \leq &\frac{1}{Z_x}\int^\infty_{-\infty}\exp\left(-\frac{(\omega+x)^2}{2\left(\frac{2x}{\beta}-\frac{1}{4\sigma^2}\right)}\right)\left|1-\exp\left(\frac{\beta^2\omega}{8x\sigma^2-\beta}\right)\right|\mathrm{d}\omega\\
 =&\frac{1}{\sqrt{2\pi\left(\frac{2x}{\beta}-\frac{1}{4\sigma^2}\right)}}\int^\infty_{-\infty}\exp\left(-\frac{(\omega+x)^2}{2\left(\frac{2x}{\beta}-\frac{1}{4\sigma^2}\right)}\right)\left|1-\exp\left(\frac{\beta^2\omega}{8x\sigma^2-\beta}\right)\right|\mathrm{d}\omega\\
 =&\frac{1}{\sqrt{2\pi}}\int^\infty_{-\infty}\exp\left(-u^2/2\right)\left|1-\exp\left(\frac{\beta^2\left(u\sqrt{2x/\beta-1/(4\sigma^2)}-x\right)}{8x\sigma^2-\beta}\right)\right|\mathrm{d}\omega\\
 =&\mc{O}\left(\frac{\beta^{2}}{\sigma^2}\frac{x+\sqrt{2x/\beta-1/(4\sigma^2)}}{x-\beta/(8\sigma^2)}\right)\,.
\end{aligned}
\end{equation}
}

Define
\[
\widehat{\mc{D}}_{A_S,x}(\rho)=\int^\infty_{-\infty}\widehat{\gamma}_x(\omega)\mathcal{D}_{V_{A_S,f,\infty}(\omega)}(\rho)\mathrm{d}\omega\,.
\]
Using the above estimation, we have
\[
\left\|\widehat{\mc{D}}_{A_S,x}(\rho)-\int^\infty_{-\infty}\rev{\gamma_x}(\omega)\mathcal{D}_{V_{A_S,f,\infty}(\omega)}(\rho)\mathrm{d}\omega\right\|_{1\leftrightarrow1}=\rev{\mc{O}\left(\|\gamma_x-\hat{\gamma}_x\|_{L^\infty}\|f(t)\|_{L^2}\right)=\mc{O}\left(\frac{\beta^{2}}{\sigma^2}\frac{1+x/\sqrt{2x/\beta-1/(4\sigma^2)}}{x-\beta/(8\sigma^2)}\right)}\,,
\]
where the second term is the original dissipative part with $\gamma_x(\omega)$.~\rev{Here, the first equality is a result of~\cite[Proposition A.1]{ChenKastoryanoBrandaoEtAl2023} and~\cite[Lemma A.1]{ChenKastoryanoBrandaoEtAl2023}.}

Next, it is straightforward to check that, $\widehat{\gamma}_x$ satisfies~\cite[Eqn. (1.4)]{ChenKastoryanoGilyen2023} up to a normalization factor. According to~\cite[Appendix A Corollaries A.1,A.2]{ChenKastoryanoGilyen2023}~\rev{with $\sigma_E=1/(2\sigma)$, $\sigma_r=\sqrt{2x/\beta-\sigma^2_E}$, $g(\omega)=\delta(\omega-x)$},~\eqref{eqn:KMS_Lindbladian} satisfies the KMS detailed balance. Furthermore, we have
\[
\|h_1\|_{L^1}=\mc{O}\left(\left\|\frac{1}{\cosh{2\pi t/\beta}}\right\|_{L^1}\left\|(\sigma\beta)^{-1}\sin(- t\beta/(4\sigma^2))\exp\left(-t^2/(2\sigma^2)\right)\right\|_{L^1}\right)=\mc{O}(\beta/\sigma),\quad \|h_2\|_{L^1}=\mc{O}\left(1\right)\,,
\]
which implies $\|B_{A_S}\|\leq \|h_1\|_{L^1}\|h_2\|_{L^1}=\mc{O}(\rev{\beta}/\sigma)$. In summary, we have
\[
\left\|\underbrace{-i\left[\frac{B_{A_S}}{Z_x},\rho\right]+\widehat{\mc{D}}_{A_S,x}(\rho)}_{\widehat{\mc{L}}_{\text{KMS},x}}-\int^\infty_{-\infty}\gamma_x(\omega)\mathcal{D}_{V_{A_S,f,\infty}(\omega)}(\rho)\mathrm{d}\omega\right\|_{1\leftrightarrow1}=\mc{O}\left(\rev{\frac{\beta^{2}}{\sigma^2}\frac{1+x/\sqrt{2x/\beta-1/(4\sigma^2)}}{x-\beta/(8\sigma^2)}+\frac{\beta}{\sigma Z_x}}\right)\,.
\]
Furthermore, we can verify that
\[
\widehat{R}:=\int^\infty_{0}\left|\int^\infty_{-\infty}\widehat{\gamma}_x(\omega)\exp(i\omega \sigma q)\mathrm{d}\omega\right|\exp(-q^2/8)\mathrm{d}q=\mc{O}\left(\frac{\sqrt{\beta}}{\sigma \sqrt{x-\frac{\beta}{8\sigma^2}}}\right)\,,
\]
which implies
\begin{equation}\label{eqn:R_bound_fermion}
R:=\int^\infty_{0}\left|\int^\infty_{-\infty}\gamma_x(\omega)\exp(i\omega \sigma q)\mathrm{d}\omega\right|\exp(-q^2/8)\mathrm{d}q=\mc{O}\left(\frac{\sqrt{\beta}}{\sigma \sqrt{x-\frac{\beta}{8\sigma^2}}}+\rev{\frac{\beta^{2}}{\sigma^2}\frac{x+\sqrt{2x/\beta-1/(4\sigma^2)}}{x-\beta/(8\sigma^2)}}\right)\,.
\end{equation}
Here, $\rev{\hat{R}},R$ are defined according to Lemma~\ref{lem:Lamb_shift_commute}~\eqref{eqn:condition_R_thermal}. According to the proof of Lemma~\ref{lem:thermal_approx_first} and Lemma~\ref{lem:Lamb_shift_commute}, there exists a Hermitian matrix $H_x$ such that
\[
\left\|\mc{L}-\left(-i[H_x,\rho]+\widehat{\mc{L}}_{\text{KMS},x}\right)\right\|_{1\leftrightarrow1}=\mc{O}\left(\sigma\exp\left(-T^2/(4\sigma^2)\right)+\rev{\frac{1}{Z_x}\left(\frac{\beta^{2}}{\sigma^2}\frac{x+\sqrt{2x/\beta-1/(4\sigma^2)}}{x-\beta/(8\sigma^2)}+\frac{\beta}{\sigma}\right)}+\frac{\beta}{\sigma}\right)\,,
\]
and
\[
\left\|\sigma^{-1/4}_{\beta}H_x \sigma^{1/4}_{\beta}-\sigma^{1/4}_{\beta}H_x \sigma^{-1/4}_{\beta}\right\|=\mc{O}\left(\frac{\beta}{\sigma}\sqrt{\frac{\beta}{x-\frac{\beta}{8\sigma^2}}}+\rev{\frac{\beta^{3}}{\sigma^2}\frac{x+\sqrt{2x/\beta-1/(4\sigma^2)}}{x-\beta/(8\sigma^2)}}\right)\,.
\]
This concludes the proof of the first part of~\cref{thm:KMS_general}.

Now, let $x=\frac{\beta}{8\sigma^2}+\omega_{\max}$ with $\omega_{\max}=\Omega(\beta)$. We can provide a better estimation for~\eqref{eqn:gamma_diff_L_1}. Let $c=2\omega_{\max}/\beta$. Then $c=\Omega(1)$ and
\[
\begin{aligned}
 &\left\|\gamma_x-\hat{\gamma}_x\right\|_{L^1}\leq \int^\infty_{-\infty}g_x(\omega)\frac{\exp(\beta \omega)}{1+\exp(\beta\omega)}\left|1-\exp\left(\frac{\beta^2\omega}{8x\sigma^2-\beta}\right)\right|\mathrm{d}\omega\\
 \leq &\frac{1}{\sqrt{2\pi c}}\int^\infty_{-\infty}\exp\left(-\frac{(\omega+x)^2}{2c}\right)\frac{\exp(\beta \omega)}{1+\exp(\beta\omega)}\left|1-\exp\left(\frac{\beta \omega}{4\sigma^2 c}\right)\right|\mathrm{d}\omega\\
 =&\frac{1}{\sqrt{2\pi c}}\int^0_{-\infty}\exp\left(-\frac{(\omega+x)^2}{2c}\right)\frac{\exp(\beta \omega)}{1+\exp(\beta\omega)}\left|1-\exp\left(\frac{\beta \omega}{4\sigma^2 c}\right)\right|\mathrm{d}\omega\\
 &+\frac{1}{\sqrt{2\pi c}}\int^\infty_{0}\exp\left(-\frac{(\omega+x)^2}{2c}\right)\frac{\exp(\beta \omega)}{1+\exp(\beta\omega)}\left|1-\exp\left(\frac{\beta \omega}{4\sigma^2 c}\right)\right|\mathrm{d}\omega
\end{aligned}
\]
For the first term, we have
\[
\begin{aligned}
&\frac{1}{\sqrt{2\pi c}}\int^0_{-\infty}\exp\left(-\frac{(\omega+x)^2}{2c}\right)\frac{\exp(\beta \omega)}{1+\exp(\beta\omega)}\left|1-\exp\left(\frac{\beta \omega}{4\sigma^2 c}\right)\right|\mathrm{d}\omega\\
\leq &\max_{\omega\in(-\infty,0)}\exp(\beta \omega)\left|1-\exp\left(\frac{\beta \omega}{4\sigma^2 c}\right)\right|= \max_{\omega\in(-\infty,0)}\exp(\omega)\left|1-\exp\left(\frac{\omega}{4\sigma^2 c}\right)\right|=\mc{O}(1/(c\sigma^2))\,.
\end{aligned}
\]
For the second term, we have
\[
\begin{aligned}
&\frac{1}{\sqrt{2\pi c}}\int^\infty_{0}\exp\left(-\frac{(\omega+x)^2}{2c}\right)\frac{\exp(\beta \omega)}{1+\exp(\beta\omega)}\left|1-\exp\left(\frac{\beta \omega}{4\sigma^2 c}\right)\right|\mathrm{d}\omega\\
=&\frac{1}{\sqrt{2\pi c}}\int^\infty_{0}\exp\left(-\frac{(\omega+x)^2}{2c}\right)\left|1-\exp\left(\frac{\beta \omega}{4\sigma^2 c}\right)\right|\mathrm{d}\omega\\
\leq &\frac{1}{\sqrt{2\pi}}\int^\infty_{0}\exp\left(-\frac{(u+x/\sqrt{c})^2}{2}\right)\left|1-\exp\left(\frac{\beta u }{4\sigma^2 \sqrt{c}}\right)\right|\mathrm{d}u\\
= &\frac{1}{\sqrt{2\pi}}\left(\int^{\log(\sigma)}_{0}\exp\left(-\frac{(u+x/\sqrt{c})^2}{2}\right)\left|1-\exp\left(\frac{\beta u }{4\sigma^2 \sqrt{c}}\right)\right|\mathrm{d}u+\int^\infty_{\log(\sigma)}\exp\left(-\frac{(u+x/\sqrt{c})^2}{2}\right)\left|1-\exp\left(\frac{\beta u }{4\sigma^2 \sqrt{c}}\right)\right|\mathrm{d}u\right)
\end{aligned}
\]
For the first part, we have
\[
\begin{aligned}
&\int^{\log(\sigma)}_{0}\exp\left(-\frac{(u+x/\sqrt{c})^2}{2}\right)\left|1-\exp\left(\frac{\beta u }{4\sigma^2 \sqrt{c}}\right)\right|\mathrm{d}u\\
\leq &\log(\sigma) \max_{u\in [0,\log(\sigma)]}\exp\left(-\frac{(u+x/\sqrt{c})^2}{2}\right)\left|1-\exp\left(\frac{\beta u }{4\sigma^2 \sqrt{c}}\right)\right|\\
=&\mathcal{O}\left(\frac{\log(\sigma)}{\sqrt{c}}\max_{u\in [0,\log(\sigma)]}\frac{\beta u}{\sigma^2}\exp(-(u+\beta \sqrt{c}/2)^2/2\right)\\
=&\mathcal{O}\left(\frac{\log^2(\sigma)}{\sqrt{c}}\max_{u\in [0,\log(\sigma)]}\frac{\beta}{\sigma^2}\exp(-(u+\beta \sqrt{c}/2)^2/2\right)=\mathcal{O}\left(\sigma^{-2}\log^2(\sigma)/\sqrt{c}\right)
\end{aligned}
\]
For the second part, we have
\[
\begin{aligned}
&\int^\infty_{\log(\sigma)}\exp\left(-\frac{(u+x/\sqrt{c})^2}{2}\right)\left|1-\exp\left(\frac{\beta u }{4\sigma^2 \sqrt{c}}\right)\right|\mathrm{d}u\\
\leq & 2 \int^\infty_{\log(\sigma)}\exp\left(-\frac{(u+x/\sqrt{c})^2-\beta u/(2\sigma\sqrt{c})}{2}\right)\mathrm{d}u\\
\leq & 2 \int^\infty_{\log(\sigma)}\exp\left(-\frac{(u+x/\sqrt{c})^2-u/(2\sqrt{c})}{2}\right)\mathrm{d}u\leq 2 \int^\infty_{\log(\sigma)}\exp\left(-\frac{u^2/2}{2}\right)\mathrm{d}u=\mathcal{O}\left(\exp(-\log(\sigma)^2/4)\right)
\end{aligned}
\]
where we use $u>1/\sqrt{c}$ when $u>\log(\sigma)$. Plugging this back, we have
\[
\left\|\gamma_x-\hat{\gamma}_x\right\|_{L^1}=\mc{O}\left(\sigma^{-2}\log^2(\sigma)\right)\,.
\]
when $\sigma=\Omega(\beta)$. The remaining part of the calculation is very similar to the first part of the proof. Thus, we omitted.
\end{proof}

In the following section, we will upper bound the mixing time of free fermion in \cref{sec:therm_mixing_free_fermion} using KMS DBC and the mixing time of commuting local Hamiltonian in \cref{sec:thermal_mixing_general} using GNS DBC.

\section{Mixing time for thermal state preparation of free fermions}\label{sec:therm_mixing_free_fermion}

In this section, we give a rigorous version of~\cref{thm:thermal_mixing} and provide the proof.
\begin{thm}\label{thm:thermal_mixing_free_fermion_rigo} Assume $\beta=\Theta(1)$ and $\|h\|=\mc{O}(1)$. We set
\begin{equation}\label{eqn:formula_g}
g(\omega)=\frac{1}{Z_x}\exp\left(-\frac{(\omega+\beta/(8\sigma^2)+\omega_{\max})^2}{4\omega_{\max}/\beta}\right),\quad Z_x=\sqrt{4\pi\omega_{\max}/\beta}\,,
\end{equation}
where $\omega_{\max}=\Theta(1)$. Then, if
\[
\sigma=\widetilde{\Omega}(N^2\epsilon^{-1}),\quad T=\widetilde{\Omega}(\sigma),\quad \alpha=\widetilde{\Or}(\sigma^{-1}N^{-1}\epsilon^{1/2})\,
\]
we have
\[
t_{{\rm mix},\Phi}(\epsilon)=\mc{O}\left(N\log\left(\frac{2\left\|\sigma^{-1/2}_\beta\right\|}{\epsilon}\right)\right)\rev{=\mc{O}(N(N+\log(1/\epsilon)))}\,.
\]
\end{thm}
\begin{rem} \rev{We note that since $g$ is not a uniform distribution, we cannot directly apply~\cref{thm:fix_thermal} to control the fixed-point error. However, using~\eqref{eqn:R_bound_fermion} and the assumption $\beta=\Theta(1)$, we obtain
\[
R=\int^\infty_{0}\left|\int^\infty_{-\infty}\gamma(\omega)\exp(i\omega \sigma q)\mathrm{d}\omega\right|\exp(-q^2/8)\mathrm{d}q=\mc{O}\left(\frac{1}{\sigma}\right)\,.
\]
Plugging this into~\cref{thm:fix_thermal_rigor}, we obtain a fixed-point error bound that is essentially the same as in~\cref{thm:fix_thermal}. Specifically, for any $\epsilon>0$, if
\[
\sigma=\widetilde{\Omega}\left(\epsilon^{-1}t_{{\rm mix},\Phi}(\epsilon)\right),\quad T=\Omega(\sigma\log(\sigma/\epsilon))\,,
\]
and $\alpha=\mc{O}\left(\sigma T^{-2}\epsilon^{1/2}t^{-1/2}_{{\rm mix},\Phi}(\epsilon)\right)$,
then
\[
\|\rho_{\rm fix}(\Phi)-\rho_\beta\|_1<\epsilon\,.
\]
The upper bound on $t_{\rm mix,\Phi}(\epsilon)$ established in~\cref{thm:thermal_mixing_free_fermion_rigo} then implies that~\cref{cor:complete} also holds in this setting.}
\end{rem}
\begin{proof}[Proof of~\cref{thm:thermal_mixing_free_fermion_rigo}]
The proof of the theorem is based on Corollary~\ref{cor:mixing_phi_alpha} and~\cref{thm:KMS_general}. Specifically, in~\cref{thm:KMS_general}, we choose $x=\frac{\beta}{8\sigma^2}+\omega_{\max}$ and define $\widehat{\mc{L}}_{\text{KMS},c}=\widehat{\mc{L}}_{\text{KMS},x}$ and $H_{c}=H_x$. According to~\cref{thm:KMS_general}, we have
\[
\left\|\mc{L}-\left(-i[H_c,\rho]+\widehat{\mc{L}}_{\rm KMS,c}\right)\right\|_{1\leftrightarrow1}=\mc{O}\left(\sigma\exp\left(-T^2/(4\sigma^2)\right)+\frac{\beta}{\sigma}\right)\,,
\]
and
\begin{equation}\label{eqn:Lamb_diff}
\left\|\sigma^{-1/4}_{\beta}H_c \sigma^{1/4}_{\beta}-\sigma^{1/4}_{\beta}H_c \sigma^{-1/4}_{\beta}\right\|=\mc{O}(\beta/\sigma)\,.
\end{equation}
Thus, it suffices to study the spectral gap of $\widehat{\mc{L}}_{\text{KMS},c}$ defined in the above lemma. For this part, we mainly follow the approach in~\cite[Section III.A]{smid2025polynomial}.

First, given a set of creation and annihilation operators $\{c_k,c^\dagger_k\}_{k=1}^N$, we define the Majorana operators as
\[
m_{2j-1}=c_j+c^\dagger_j,\quad m_{2j}=i(c_j-c^\dagger_j),\quad j=1,\dots,N\,,
\]
which satisfies $\{m_i,m_j\}=2\delta_{i,j}$. Let $\vec{m}=[m_1,\dots,m_{2N}]^T$. Then, we have
\[
H=\sum_{i,j}h_{i,j}c^\dagger_ic_j=\sum^{2N}_{i,j=1}h^m_{i,j}m_{i}m_j-C I_{2^N\times 2^N}=\vec{m}^T\cdot h^m\cdot \vec{m}-C I_{2^N\times 2^N}\,.
\]
We note that the eigenvalues of $h^m$ is a Hermitian and antisymmetric matrix with eigenvalues $\{\lambda_k(h)/4,-\lambda_k(h)/4\}^N_{k=1}$, and $C$ is a constant.

Next, for each creation and annihilation operator pair $(c_j,c^\dagger_j)$, we have
\[
\begin{bmatrix}
c^\dagger_j\\
c_j
\end{bmatrix}=\begin{bmatrix}
1/2 & i/2\\
1/2 & -i/2
\end{bmatrix}
\begin{bmatrix}
m_{2j-1}\\
m_{2j}
\end{bmatrix}
\]
Define the coupling operator vector $\vec{A}=[c^\dagger_1, c_1, c^\dagger_2, c_2,\dots, c^\dagger_N,c_N]^T$. Then,
\begin{equation}\label{eqn:coupling_vector}
\vec{A}=\frac{M}{\sqrt{2}}\vec{m}\,,
\end{equation}
where $M$ is a unitary matrix.

\rev{In $\hat{L}_{\rm KMS,c}$, we first evaluate the coherent component $B$. By~\eqref{eqn:coupling_vector}, the coupling operator $\sqrt{2},\vec{A}$ can be written as a unitary transformation of $\vec{m}$. As shown in~\cite[Lemma III.1]{smid2025polynomial}, the coherent part under KMS detailed balance satisfies $B=0$, meaning that the coherent contribution vanishes.}

Next, we follow the proof of~\cite[Lemma III.2]{smid2025polynomial} to calculate the spectral gap of the dissipative term. We first formulate
\[
\mc{H}_0[\rho]=\sigma^{-1/4}_{\beta}\cdot \widehat{\mc{L}}^\dagger_{\rm KMS,c}[\sigma^{1/4}_{\beta}(\rho)\sigma^{1/4}_{\beta}]\sigma^{-1/4}_{\beta}\,.
\]
Because $\widehat{\mc{L}}_{\rm KMS,c}$ satisfies the KMS detailed balance condition, $\mc{H}_0$ is a self-adjoint operator with respect to the HilbertSchmidt inner product. Furthermore, $\mc{H}_0$ is a similarity transformation of the Lindbladian $\widehat{\mc{L}}^\dagger_{\rm KMS,c}$, which means that the spectral gap of $\mc{H}_0$ is the same as the spectral gap of $\widehat{\mc{L}}^\dagger_{\rm KMS,c}$.

To calculate each term in the parent Hamiltonian $\mc{H}_0[\rho]$, we first note that, for each $\omega$,
\[
\begin{bmatrix}
  \sigma^{-1/4}_{\beta}V_{c^\dagger_1}(\omega)\sigma^{1/4}_{\beta} \\
  \sigma^{-1/4}_{\beta}V_{c_1}(\omega)\sigma^{1/4}_{\beta} \\
  \vdots \\
  \sigma^{-1/4}_{\beta}V_{c^\dagger_N}(\omega)\sigma^{1/4}_{\beta} \\
  \sigma^{-1/4}_{\beta}V_{c_N}(\omega)\sigma^{1/4}_{\beta}
\end{bmatrix}=\frac{M}{\sqrt{2}}\cdot \hat{f}(-4h^m-\omega) e^{-\beta h^m}\cdot \vec{m}\,,
\]
and
\[
\sum_j \sigma^{-1/4}_{\beta}V^\dagger_{c^\dagger_j}(\omega)V_{c^\dagger_j}(\omega)\sigma^{1/4}_{\beta}+\sigma^{-1/4}_{\beta}V^\dagger_{c_j}(\omega)V_{c_j}(\omega)\sigma^{1/4}_{\beta}=\frac{1}{2}\vec{m}^\dagger\cdot \left|\hat{f}(-4h^m-\omega)\right|^2\cdot \vec{m}\,.
\]
Plugging this equality into the expression for $\mc{H}_0[\rho]$,
\[
\begin{aligned}
&\mc{H}_0[\rho]\\
=&\frac{1}{2N}\int \frac{\widehat{\gamma}(\omega)}{2}\left(\vec{m}^\dagger\cdot \hat{f}(-4h^m-\omega) e^{-\beta h^m} \cdot M^\dagger\cdot \rho\cdot M\cdot \hat{f}(-4h^m-\omega) e^{-\beta h^m}\cdot \vec{m}\right.\\
&-\left.\frac{1}{2}\vec{m}^\dagger\cdot \left|\hat{f}(-4h^m-\omega)\right|^2\cdot \vec{m}\cdot \rho-\rho\cdot \frac{1}{2}\vec{m}^\dagger\cdot \left|\hat{f}(-4h^m-\omega)\right|^2\cdot \vec{m}\right)\mathrm{d}\omega\,.
\end{aligned}
\]
Because $\beta=\Theta(1)$, $\|h\|=\mc{O}(1)$, and $\omega_{\max}=\Theta(1)$, it is straightforward to check that
\[
\left(\int \widehat{\gamma}(\omega)\left|\hat{f}(-4h^m-\omega)\right|^2\mathrm{d}\omega\right)\exp(-\beta h^m)\geq C\,,
\]
where $C$ is a constant independent of $\sigma$, meaning the coefficients of the above quadratic expansion does not generate when $\sigma$ approaches to zero. Following the proof of~\cite[Proposition III.2]{smid2025polynomial}, the spectral gap of $\mc{H}_0$ is lower bounded by a constant over $N$ independent of $\sigma$.

\rev{Let
\[
\widehat{\Phi}=\mc{U}_S(T)\,\circ\,\exp\left(\mathcal{M}\alpha^2\right)\,\circ\,\mc{U}_S(T),\quad \mathcal{M}=-i[H_c,\rho]+\widehat{\mc{L}}_{\rm KMS,c}\,.
\]
Combining Lemma~\ref{lem:thermal_approx_first} and~\cref{thm:KMS_general}, we first have
\[
\left\|\widehat{\Phi}-\Phi_{\alpha}\right\|_{1\leftrightarrow1}=\mathcal{O}\left(\alpha^2\left(\sigma\exp\left(-T^2/(4\sigma^2)\right)+\alpha^2T^4\sigma^{-2}+\frac{\beta}{\sigma}\right)\right)\,.
\]
In addition, according to Corollary~\ref{cor:mixing_phi_alpha} and~\eqref{eqn:Lamb_diff}, when $\beta/\sigma=\mathcal{O}(1)$, we have
\[
t_{{\rm mix},\widehat{\Phi}}(\epsilon)=\mc{O}\left(N\log\left(\frac{2\left\|\sigma^{-1/2}_\beta\right\|}{\epsilon}\right)\right)\,.
\]
We note that $\log\left(\left\|\sigma^{-1/2}_\beta\right\|\right)=\mc{O}(\beta\|H\|+N)=\mc{O}(\beta N)=\mc{O}(N)$ in our case. Applying~\cref{thm:almost_fixed_point} to $\Phi$ and $\widehat{\Phi}$, we obtain that, if
\[
\sigma\exp\left(-T^2/(4\sigma^2)\right)+\alpha^2T^4\sigma^{-2}+\frac{1}{\sigma}=\mc{O}(\epsilon N^{-1}(N+\log(1/\epsilon))^{-1})\,,
\]
then
\[
t_{{\rm mix},\Phi_{\alpha}}(\epsilon)=\mc{O}\left(N\log\left(\frac{2\left\|\sigma^{-1/2}_\beta\right\|}{\epsilon}\right)\right)\,.
\]
This concludes the proof.
}


\end{proof}

\section{Mixing time for commuting local Hamiltonians}\label{sec:thermal_mixing_general}

In this section, we prove a general version of~\cref{thm:thermal_commuting_local}. Unlike the previous section, here we show that, the dissipative part of $\mc{L}$ can converge to the Davies generator. The Davies generator satisfies the GNS detailed balance condition (defined in~\cref{sec:notation}), and therefore also satisfies the KMS detailed balance condition.

In the general case, establishing rigorous convergence to the Davies generator requires $\sigma$ to scale exponentially with the system size, since one may need to resolve exponentially close Bohr frequencies $\omega_1$ and $\omega_2$ in order to generate distinct jump operators $A_S(\omega_1)$ and $A_S(\omega_2)$ that appear in the Davies generator. However, in cases where the effective Bohr frequencies are not exponentially close, such as for local commuting Hamiltonians, $\sigma$ does not need to be exponentially large, allowing for an efficient approximation. This property is mainly summarized in the following theorem.
\begin{thm}\label{thm:thermal_mixing_rigor} Given a coupling set $\mc{A}$. Assume
\begin{itemize}
    \item The Davies generator $\mc{L}_{D,\mc{A}}[\rho]$ has a spectral gap $\lambda_{\rm gap}>0$.
    \item There exists a constant $\delta>0$ such that: for any $A\in\mc{A}$, $\omega_1\neq \omega_2$, if $A(\omega_1)\neq 0$ and $A(\omega_2)\neq 0$, we must have $|\omega_1-\omega_2|\geq \delta$.
     \item There exists a constant $M$ such that $\sup_{A\in\mc{A}}\left|\{\omega|A(\omega)\neq 0\}\right|\leq M$.
\end{itemize}
Given any $\epsilon>0$, if
\[
\sigma=\Omega\left(\beta|\mc{A}|\lambda^{-1}_{\rm gap}\delta^{-1}\epsilon^{-1}M\log\left(\left\|\sigma^{-1/2}_\beta\right\|/\epsilon\right)\mathbb{E}(\|A_S\|^2)\right),\ T=\widetilde{\Omega}(\sigma)\,,
\]
and
\[
\alpha=\mc{O}\left(\epsilon^{1/2}\sigma^{-1}|\mc{A}|^{-1/2}\lambda^{1/2}_{\rm gap}\log^{-1/2}\left(\left\|\sigma^{-1/2}_\beta\right\|/\epsilon\right)\mathbb{E}^{-1/2}(\|A_S\|^4)\right)\,,
\]
then
\[
t_{{\rm mix},\widehat{\Phi}}(\epsilon)\leq \frac{|\mc{A}|}{\lambda_{\rm gap}}\log\left(\frac{8\left\|\sigma^{-1/2}_\beta\right\|}{\epsilon}\right)+1
\]
\end{thm}
In~\cref{thm:thermal_mixing_rigor}, we note that the second condition is a technical assumption used to bound the difference between $\widetilde{\mathcal{L}}$ and the Davies generator. When $\delta$ is not exponentially small, it is not necessary to fully resolve all Bohr frequencies to ensure that the Lindbladian dynamics converges to the Davies generator. This condition is easily satisfied in the case of local commuting Hamiltonians: $A(\omega)$ is determined by the interaction between $A$ and a constant number of local Hamiltonians. Hence, it suffices to choose $\delta$ to be a constant to guarantee that different components $A(\omega)$ are well separated.

\begin{rem}\label{rem:local_commuting} According to~\cite[Section VIII]{kastoryano2016quantum}, when $H$ is a local commuting Hamiltonian as stated in~\cref{thm:thermal_commuting_local}, there exists a constant $\beta_c$ dependent on the Hamiltonian $H$ such that for every $\beta\leq \beta_c$, the spectral gap of the Davies generator can be lower bounded by a constant, meaning $\lambda_{\mathrm{gap}} = \Theta(1)$. Furthermore, we also have $\delta = \Omega(1)$, $M = \mathcal{O}(1)$, $|\mathcal{A}| = \mathcal{O}(N)$, $\|H\|=\mc{O}(N)$. Noticing $\log\left(\left\|\sigma^{-1/2}_\beta\right\|\right)=\mc{O}(\beta\|H\|+N)=\mc{O}((\beta+1)N)$, we can choose
\[
\sigma = \widetilde{\Omega}\left(\epsilon^{-1}(\beta+1)^2 N^2\right), \quad T = \widetilde{\Omega}(\sigma), \quad \text{and} \quad
\alpha = \widetilde{\mathcal{O}}\left(\epsilon^{3/2}(\beta+1)^{-5/2}N^{-3}\right),
\]
to obtain
\[
t_{\mathrm{mix}, \Phi}(\epsilon) = \mathcal{O}\left( N(\|H\|\beta+N)\log(1/\epsilon) \right)=\mathcal{O}\left( N^2(\beta+1)\log(1/\epsilon) \right).
\]
This gives~\cref{thm:thermal_commuting_local}.
\end{rem}
\begin{proof} Recall $\widetilde{\Phi}$ defined in~\cref{eqn:Phi_tilde_thermal}:
\begin{equation}\label{eqn:Phi_tilde_thermal_recall}
\widetilde{\Phi}=\mc{U}_S(T)\,\circ\,\exp\left(\widetilde{\mathcal{L}}\alpha^2\right)\,\circ\,\mc{U}_S(T)\,.
\end{equation}
Here
\begin{equation}\label{eqn:L_tilde_thermal_recall}
\begin{aligned}
\widetilde{\mc{L}}(\rho)=-i\left[\widetilde{H}_{\mathrm{LS}},\rho\right]+\mathbb{E}_{A_S}\left(\int^\infty_{-\infty}\gamma(\omega)\mathcal{D}_{\widetilde{V}_{A_S,f}(\omega)}(\rho)\mathrm{d}\omega\right)\,,
\end{aligned}
\end{equation}
where
\[
\widetilde{H}_{\mathrm{LS}}=-\rev{\mathbb{E}_{A_S}\left(\mathrm{Im}\left(\int^\infty_{-\infty}\gamma(\omega)\widetilde{\mc{G}}_{A_S,f}(-\omega)\mathrm{d}\omega\right)\right)},\quad \widetilde{V}_{A_S,f}(\omega)=\int^\infty_{-\infty}f(t)A_S(t)\exp(-i\omega t)\mathrm{d}t\,,
\]
with
\begin{equation}\label{eqn:G_S_thermal_recall}
\widetilde{\mc{G}}_{A,f}(\omega)=\int^\infty_{-\infty}\int^{s_1}_{-\infty}f(s_2)f(s_1) A^\dagger(s_2)A(s_1)\exp(i\omega(s_2-s_1))\mathrm{d}s_2\mathrm{d}s_1\,.
\end{equation}

Similar to the proof of~\cref{thm:gs_mixing_rigor}, we will construct
\begin{equation}\label{eqn:Phi_hat_thermal}
\widehat{\Phi}=\mc{U}_S(T)\,\circ\,\exp\left(\widehat{\mathcal{L}}\alpha^2\right)\,\circ\,\mc{U}_S(T)
\end{equation}
with
\begin{equation}\label{eqn:L_tilde_thermal_hat}
\begin{aligned}
\widehat{\mc{L}}(\rho)=-i\left[\widehat{H}_{\mathrm{LS}},\rho\right]+\mathbb{E}_{A_S}\left(\sum_{\omega}\mc{L}_{D,A_S}\right)\,,
\end{aligned}
\end{equation}
such that $\left\|\widehat{\Phi}-\widetilde{\Phi}\right\|_{1\leftrightarrow1}$ is small. Here $\mc{L}_{D,A_S}$ is the Davies generator associated with the coupling operator $A_S$.

We first deal with the Lamb shift term. For each $\omega$, we have
\[
\begin{aligned}
\widetilde{\mc{G}}_{\rev{A_S},f}(\omega)=&\int^\infty_{-\infty}\int^{s_1}_{-\infty}f(s_2)f(s_1) A_S(s_2)A_S(s_1)\exp(-i\omega(s_2-s_1))\mathrm{d}s_2\mathrm{d}s_1\\
=&\sum_{\nu_1,\nu_2\in B(H_S)}A^\dagger_S(\nu_2)A_S(\nu_1)\int^{\infty}_{-\infty}\int^{s_1}_{-\infty}f(s_2)f(s_1)\exp(i\nu_2s_2)\exp(i\nu_1s_1)\exp(-i\omega(s_2-s_1))\mathrm{d}u\mathrm{d}v\\
=&\frac{\sigma}{2\sqrt{2\pi}}\sum_{\nu_1,\nu_2\in B(H_S)}A^\dagger_S(\nu_2)A_S(\nu_1)\\
&\cdot\underbrace{\int^{\infty}_{-\infty}\exp\left(i\frac{\sigma p}{2}(\nu_1+\nu_2)\right)\exp\left(-\frac{p^2}{8}\right)\mathrm{d}p}_{=\mc{O}(\exp(-\sigma^2(\nu_1+\nu_2)^2/2))}\underbrace{\int^{\infty}_{0}\exp\left(-\frac{q^2}{8}\right)\exp\left(i\frac{\sigma q}{2}(\nu_1-\nu_2)\right)\exp(i\sigma\omega q)\mathrm{d}q}_{=\mc{O}(1)}
\end{aligned}
\]
Define the commuting part as $\widehat{\mc{G}}_{\rev{A_S},f}(\omega)$:
\[
\begin{aligned}
\widehat{\mc{G}}_{\rev{A_S},f}(\omega)
=&\frac{\sigma}{2\sqrt{2\pi}}\sum_{\nu_1+\nu_2=0}A^\dagger_S(\nu_2)A_S(\nu_1)\\
&\cdot\int^{\infty}_{-\infty}\exp\left(i\frac{\sigma p}{2}(\nu_1+\nu_2)\right)\exp\left(-\frac{p^2}{8}\right)\mathrm{d}p\int^{\infty}_{0}\exp\left(-\frac{q^2}{8}\right)\exp\left(i\frac{\sigma q}{2}(\nu_1-\nu_2)\right)\exp(i\sigma\omega q)\mathrm{d}q
\end{aligned}
\]
According to the assumption of $A_S$, we have $|\nu_1+\nu_2|\geq \delta$ in the summation of $\widetilde{\mc{G}}_{\rev{A_S},f}(\omega)$. Thus,
\[
\left\|\widetilde{\mc{G}}_{\rev{A_S},f}(\omega)-\widehat{\mc{G}}_{\rev{A_S},f}(\omega)\right\|=\mathcal{O}\left(\sigma\exp(-\sigma^2\delta^2/2)\|\rev{A_S}\|^2\rev{M^2}\right)\,,
\]
where $M$ comes from the total number of terms in the summation.

Define \rev{$\widehat{H}_{\mathrm{LS}}=-\mathbb{E}_{A_S}\left(\mathrm{Im}\left(\int^\infty_{-\infty}\gamma(\omega)\widehat{\mc{G}}_{A_S,f}(-\omega)\mathrm{d}\omega\right)\right)$}. We have
\[
\left\|\widetilde{H}_{\mathrm{LS}}-\widehat{H}_{\mathrm{LS}}\right\|=\mathcal{O}\left(\sigma\exp(-\sigma^2\delta^2/2)\rev{M^2}\mathbb{E}(\|A_S\|^2)\right)
\,.
\]

Next, for the dissipative operator, we have
\[
\widetilde{V}_{A,f}(\omega)=\int^\infty_{-\infty}f(t)A(t)\exp(-i\omega t)\mathrm{d}t=2^{3/4}\rev{\sigma^{1/2}\pi^{1/4}\sum_{\xi\in B(H)}\exp(-(\xi-\omega)^2\sigma^2)A(\xi)}
\]
Then
\[
\begin{aligned}
&\int^\infty_{-\infty}\gamma(\omega)\mathcal{D}_{\widetilde{V}_{A_S,f}(\omega)}(\rho)\mathrm{d}\omega-\mc{L}_{D,A_S}\\
=&\sum_{\xi\in B(H)} \left(\rev{\underbrace{\int^\infty_{-\infty}2^{3/2}\pi^{1/2}\sigma \exp(-2(\xi-\omega)^2\sigma^2)\gamma(\omega)\mathrm{d}\omega-2\pi\gamma(\xi)}_{|\cdot|=\mc{O}\left(\beta\sigma^{-1}\right)}}\right)\mc{L}_{A_S(\xi)}+\underbrace{\int^\infty_{-\infty}\gamma(\omega)\left(\sum_{\xi_1\neq \xi_2}\cdots\right)\mathrm{d}\omega}_{\|\cdot\|_{1\leftrightarrow1}=\mathcal{O}\left(\sigma \exp(-\sigma^2\delta^2/\rev{2})M^2\|A_S\|^2\right)}\,.
\end{aligned}
\]
This implies that
\[
\left\|\int^\infty_{-\infty}\gamma(\omega)\mathcal{D}_{\widetilde{V}_{A_S,f}(\omega)}(\rho)\mathrm{d}\omega-\mc{L}_{D,A_S}\right\|_{1\leftrightarrow1}=\mc{O}\left(\left(\beta\sigma^{-1}M+\sigma \exp(-\sigma^2\delta^2/\rev{2})M^2\right)\|A_S\|^2\right)\,.
\]
In conclusion,
\[
\left\|\widehat{\Phi}-\widetilde{\Phi}\right\|_{1\leftrightarrow1}=\mc{O}\left(\alpha^2\left(\beta\sigma^{-1}M+\sigma \exp(-\sigma^2\delta^2/\rev{2})M^2\right)\mathbb{E}(\|A_S\|^2)\right)\,.
\]
Combining this and Lemma~\ref{lem:thermal_approx_first}, we have
\[
\left\|\widehat{\Phi}-\Phi\right\|_{1\leftrightarrow1}=\mc{O}\left(\alpha^2\left(\beta\sigma^{-1}M+\sigma \exp(-\sigma^2\delta^2/\rev{2})M^2+\sigma\exp\left(-T^2/(4\sigma^2)\right)\right)\mathbb{E}(\|A_S\|^2)
+\alpha^4T^4\sigma^{-2}\mathbb{E}\left(\|A_S\|^4\right)\right)\,.
\]
Furthermore, according to Corollary~\ref{cor:mixing_phi_alpha}, we have
    \[
t_{{\rm mix},\widehat{\Phi}}(\epsilon)\leq \frac{|\mc{A}|}{\lambda_{\rm gap}}\log\left(\frac{2\left\|\sigma^{-1/2}_\beta\right\|}{\epsilon}\right)+1\,.
\]
Finally, when
\[
\sigma=\Omega\left(|\mc{A}|\lambda^{-1}_{\rm gap}\delta^{-1}\epsilon^{-1}\beta M\log\left(\left\|\sigma^{-1/2}_\beta\right\|/\epsilon\right)\mathbb{E}\left(\|A_S\|^2\right)\right),\ T=\widetilde{\Omega}(\sigma)\,,
\]
and
\[
\alpha=\mc{O}\left(\epsilon^{1/2}\sigma^{-1}|\mc{A}|^{-1/2}\lambda^{1/2}_{\rm gap}\log^{-1/2}\left(\left\|\sigma^{-1/2}_\beta\right\|/\epsilon\right)\mathbb{E}^{-1/2}\left(\|A_S\|^4\right)\right)\,,
\]
we have
\[
t_{{\rm mix},\widehat{\Phi}}(\epsilon)\left\|\widehat{\Phi}-\Phi\right\|_{1\leftrightarrow1}\leq \epsilon\,.
\]
Applying~\cref{thm:almost_fixed_point}, we conclude the proof.
\end{proof}

\bibliographystyle{unsrtnat}
\bibliography{ref}